\documentclass[runningheads]{llncs}

\usepackage[T1]{fontenc}
\usepackage{amsmath,amssymb,amsfonts}
\usepackage{booktabs}
\usepackage{longtable}
\usepackage{xspace}
\usepackage{xcolor}
\usepackage{tikz}
\usepackage{hyperref}
\usepackage[capitalize]{cleveref}

\usepackage{thmtools}
\usepackage{thm-restate}

\makeatletter
\def\@unspthm#1{%
  \expandafter\let\csname #1\endcsname\@undefined
  \expandafter\let\csname end#1\endcsname\@undefined
  \expandafter\let\csname #1name\endcsname\@undefined
  \expandafter\let\csname the#1\endcsname\@undefined
  \expandafter\let\csname p@#1\endcsname\@undefined
  \expandafter\let\csname c@#1\endcsname\@undefined}
\@unspthm{definition}
\@unspthm{example}
\@unspthm{theorem}
\@unspthm{proposition}
\@unspthm{lemma}
\@unspthm{corollary}
\@unspthm{remark}
\makeatother

\declaretheorem[numberwithin=section, name=Definition]{definition}
\declaretheorem[sibling=definition, name=Example]{example}
\declaretheorem[sibling=definition, name=Theorem]{theorem}
\declaretheorem[sibling=definition, name=Proposition]{proposition}
\declaretheorem[sibling=definition, name=Lemma]{lemma}
\declaretheorem[sibling=definition, name=Corollary]{corollary}
\declaretheorem[sibling=definition, style=remark, name=Remark]{remark}

\newcommand{\mypara}[1]{\vspace{0pt}\noindent\textbf{#1.}}
\newcommand{\temph}[1]{\textbf{#1}}

\newcommand{\remove}[1]{}

\newcommand{\calV}{\mathcal{V}}
\newcommand{\calG}{\mathcal{G}}
\newcommand{\calT}{\mathcal{T}}
\newcommand{\calF}{\mathcal{F}}

\newcommand{\calS}{\mathcal{S}}

\newcommand{\calC}{\mathcal{C}}

\newcommand{\ia}{\textit{i}}
\newcommand{\ib}{\textit{ii}}
\newcommand{\ic}{\textit{iii}}

\newcounter{pc}
\newcommand\spc{\refstepcounter{pc}\thepc}
\newcommand{\Program}[1]{\smallskip\noindent\textbf{Program \spc: #1}\vspace{-6pt}}

\newif\ifappendix
\appendixtrue

\usepackage{comment}
\ifappendix\else
\excludecomment{proof}
\fi

\newcommand{\arxivref}{the full paper~\cite{shapiro2026implementing}}
\newcommand{\appref}[2]{\ifappendix\cref{#1}\else #2\fi}
\newcommand{\proofref}{\par\noindent\textit{Proof.} See \arxivref.\medskip}

\usepackage{etoolbox}
\makeatletter
\pretocmd{\@verbatim}{\topsep=2\p@ \partopsep=\z@}{}{}
\makeatother

\AtBeginDocument{%
  \setlength{\abovedisplayskip}{2pt plus 1pt minus 1pt}%
  \setlength{\belowdisplayskip}{2pt plus 1pt minus 1pt}%
  \setlength{\abovedisplayshortskip}{1pt plus 1pt}%
  \setlength{\belowdisplayshortskip}{1pt plus 1pt}%
}

\makeatletter
\def\@spthm#1#2#3#4{\topsep 2\p@ \@plus1\p@ \@minus1\p@
\refstepcounter{#1}%
\@ifnextchar[{\@spythm{#1}{#2}{#3}{#4}}{\@spxthm{#1}{#2}{#3}{#4}}}
\def\@Thm#1#2#3{\topsep 2\p@ \@plus1\p@ \@minus1\p@
\@ifnextchar[{\@Ythm{#1}{#2}{#3}}{\@Xthm{#1}{#2}{#3}}}
\makeatother

\usepackage{enumitem}
\setlist{nosep, leftmargin=*, itemsep=0pt, topsep=0pt}

\title{Implementing Grassroots Logic Programs with Multiagent Transition Systems and AI
\ifappendix\texorpdfstring{\\}{ }(Full Version)\fi}
\titlerunning{Implementing Grassroots Logic Programs}

\author{Ehud Shapiro}
\authorrunning{E.~Shapiro}
\institute{London School of Economics and Weizmann Institute of Science\\
\email{e.shapiro2@lse.ac.uk}}

\begin{document}

\maketitle

\begin{abstract}
Grassroots Logic Programs (GLP) is a concurrent logic programming language in which logic variables are partitioned into paired readers and writers. An assignment is produced at most once via a writer and consumed at most once via its paired reader, and may contain additional readers and/or writers.  This enables the concise expression of rich multidirectional communication modalities.

The language was introduced together with concurrent (cGLP) and multiagent (maGLP)  operational semantics.  Here, we derive from these (\ia)~dGLP, a deterministic counterpart of cGLP, and (\ib)~madGLP, a counterpart of maGLP in which deterministic agents communicate solely by asynchronous message passing, and prove them correct against their abstract counterparts.  maGLP shared variable pairs spanning agents can be implemented by two local variable pairs joined by a \emph{global link}, with correctness following from disjoint substitution commutativity (a consequence of GLP's single-occurrence invariant).  We further prove that madGLP is grassroots.   Both dGLP and madGLP serve as formal specifications for an AI-driven implementation discipline (math $\to$ informal spec $\to$ Dart) employed and described here: from dGLP, AI (Claude) developed a workstation-based GLP implementation in Dart, and from madGLP it is developing a smartphone-based multiagent one.

\keywords{concurrent logic programming \and multiagent systems \and transition systems \and operational semantics \and peer-to-peer \and grassroots platforms}
\end{abstract}


\section{Introduction}
\label{sec:introduction}

\mypara{Grassroots platforms}
A digital platform is \emph{grassroots}~\cite{shapiro2023grassrootsBA,shapiro2025atomic} if it can have multiple instances that can (\ia)~operate independently of each other and of any global resource other than the network, and (\ib)~coalesce into ever larger instances, possibly resulting in a single global instance.

\mypara{GLP}
Grassroots Logic Programs (GLP)~\cite{shapiro2025glp} is a concurrent logic programming language in which logic variables are partitioned into paired \emph{readers} $X?$ and \emph{writers} $X$, each occurring at most once (single-occurrence, SO), with a variable occurring in a clause iff its paired variable also does (single-reader single-writer, SRSW). GLP eschews unification in favour of simple term matching, is linear-logic-like~\cite{girard1987linear}, and is futures/promises-like~\cite{baker1977future,friedman1976impact}: each assignment $X := T$ is produced at most once via the sole occurrence of a writer (promise) $X$, and consumed at most once via the sole occurrence of its paired reader (future) $X?$; the term $T$ may contain further readers and/or writers, enabling the concise expression of rich multidirectional communication modalities.

Being relational and nondeterministic (and thus not functional), GLP---like its predecessor concurrent logic languages~\cite{shapiro1989family}---can express nondeterministic asynchronous concurrent processes such as the fair merging of streams. The SRSW restriction is what makes GLP suitable for grassroots multiagent distributed execution: each variable assignment has exactly one producer and one consumer, so an assignment can be communicated asynchronously from its writer to its reader without consensus.
The intended deployment of GLP is grassroots systems consisting of people operating their smartphones. 

\mypara{From specifications to implementations}
Prior work has defined concurrent and multiagent nondeterministic operational semantics for GLP, cGLP and maGLP~\cite{shapiro2025glp}, via transition systems and multiagent transition systems~\cite{shapiro2021multiagent,lewis2026volitional}, and a type and module system for GLP~\cite{shapiro2026types}. \ifappendix Here, we derive from cGLP and maGLP deterministic counterparts dGLP and madGLP, and prove each of them correct with respect to its abstract counterpart, madGLP for \emph{authorised runs} (\Cref{app:held-links}).\fi 

\mypara{Implementing a shared variable, by example}
\ifappendix
A maGLP shared pair couples a writer at one agent with its reader at another---say writer $X$ at $p$, reader $X?$ at $q$, so values flow $p\to q$. madGLP realises this with two local pairs, $(X_p, X_p?)$ at $p$ and $(X_q, X_q?)$ at $q$, joined by a \emph{global link}: a spawned goal at $p$ suspends on $X_p?$ and, when $X_p$ is assigned, sends the value to $q$, which assigns $X_q$. The maGLP writer and reader are the outer ends ($X_p$, $X_q?$); the link bridges the inner ends ($X_p?$, $X_q$). Each shared variable is realised by two local pairs, the reader of one linked to the writer of the other.
\fi

\mypara{AI development methodology}
The AI-based development employs three layers of abstraction---math, an English+code informal specification, and Dart code. Authority flows math$\to$spec$\to$Dart but harmonised via back-and-forth in which running the implementation and deriving the specification surfaced defects, also at the mathematical level, prompting major redesigns of madGLP. madGLP and its single-agent specialisation dGLP serve as the formal specifications: from dGLP, AI (Claude) developed a workstation-based GLP implementation in Dart~\cite{dart2024}; from madGLP, it is developing a smartphone-based multiagent one. See Section~\ref{sec:ai-methodology}.

\mypara{Related work}
GLP belongs to the family of concurrent logic languages~\cite{shapiro1989family,shapiro1983subset,ueda1986guarded,clark1986parlog,mierowsky1985fcp}; it can be understood as Flat Concurrent Prolog with SRSW added, simplifying read-only unification~\cite{levi1985readonly}, with strong connections to mode and linearity analysis~\cite{ueda1994moded,ueda2001resource}, linear logic and session types~\cite{girard1987linear,caires2010session,wadler2014propositions,honda2016multiparty}, and futures and promises~\cite{baker1977future,pruiksma2022futures}. Our correctness framework draws on I/O automata~\cite{lynch1996forward}, refinement mappings~\cite{abadi1991existence}, and atomic-transaction theory~\cite{lynch1988atomic,herlihy1990linearizability}, using disjoint substitution commutativity (a consequence of SO) where linearizability uses linearization points. Grassroots~\cite{shapiro2023grassrootsBA,lewis2026volitional} extends the CRDT~\cite{shapiro2011conflict} tradition with compositional Byzantine tolerance via the Blocklace~\cite{almeida2024blocklace}. We prove that an agent-deterministic, peer-to-peer, message-passing implementation correctly implements maGLP's nondeterministic shared-variable semantics. Session types establish analogous implementation-correctness results for abstract models~\cite{toninho2013higher}; we are not aware of an actual concurrent programming language implementation~\cite{armstrong2010erlang,pike2012go,matsakis2014rust,cooper2006links,clebsch2015deny,johnsen2011abs} accompanied by such a proof. AI-driven derivation of implementations from formal specifications~\cite{fowler2025sdd,mundler2025type,blinn2024typed} is an emerging discipline; we contribute a three-layer (math$\to$spec$\to$Dart) variant at greater scale (Section~\ref{sec:ai-methodology}). See \appref{sec:related-work}{\arxivref} for fuller discussion.

\mypara{Paper outline}
Section~\ref{sec:ts-implementations} recalls transition systems and the notion of one transition system correctly implementing another.
Section~\ref{sec:glp} presents GLP syntax and its concurrent operational semantics cGLP\ifappendix, together with the deterministic dGLP that correctly implements it\fi.
Section~\ref{sec:multiagent-ts} recalls multiagent transition systems and atomic transactions, the setting in which agents share variable pairs across boundaries.
Section~\ref{sec:maglp} recalls the multiagent maGLP, defines its agent-deterministic, message-passing counterpart madGLP, and proves madGLP correctly implements maGLP\ifappendix{} for authorised runs\fi.
\ifappendix
Section~\ref{app:madglp-spec} provides the complete madGLP specification: full definitions, remarks, and correctness proofs.
\fi
Section~\ref{sec:ai-methodology} describes the three-layer (math$\to$spec$\to$Dart) AI development methodology and the implementations derived through it\ifappendix, with concrete examples of revisions driven by implementation feedback\fi.
Section~\ref{sec:grassroots} proves that madGLP is grassroots.
\ifappendix
Section~\ref{sec:related-work} discusses related work.
\fi
Section~\ref{sec:conclusion} concludes.
\ifappendix
Appendix~\ref{app:madglp-trace} provides detailed example traces, Appendix~\ref{app:code-format} the byte-level code format, Appendix~\ref{app:implementation-notes} implementation notes on the runtime, and Appendix~\ref{app:dev-history} the development history.
\fi


\section{Transition Systems and Implementations}
\label{sec:ts-implementations}

This section presents transition systems and implementations among them, providing the mathematical foundation for proving that one transition system correctly implements another.

\subsection{Transition Systems}
\label{subsec:ts-basic}

\begin{definition}[Transition System, Computation, Run, Safe, Live, Correct~\cite{shapiro2021multiagent,lewis2026volitional}]
\label{def:ts-basic}
A \temph{transition system} is a tuple $TS = (C, c_0, T, {\sim})$ where $C$ is an arbitrary set of \temph{configurations}, $c_0 \in C$ a designated \temph{initial configuration}, $T \subseteq C \times C$ a set of \temph{transitions}, each a pair $c \rightarrow c'$ of non-identical configurations $c \ne c' \in C$, and $\sim$ an equivalence relation on $T$; we write $[t]$ for the class of $t \in T$ under $\sim$.

A \temph{computation} is a (nonempty, finite or infinite) sequence of configurations $c_1, c_2, \ldots$; it is a \temph{run} if $c_1 = c_0$, and \temph{safe} if $c_i \rightarrow c_{i+1} \in T$ for every two consecutive configurations. We write $c \xrightarrow{*} c'$ for the existence of a safe computation from $c$ to $c'$ (empty if $c = c'$). A class $[t] \in T/{\sim}$ is \temph{enabled} in $c$ if $c \rightarrow c' \in [t]$ for some $c'$, and $c$ is \temph{terminal} if no class is enabled in it. A run is \temph{live} if no class $[t]$ is enabled in every configuration of some suffix in which no member of $[t]$ occurs, and \temph{correct} if it is safe and live.
\end{definition}

\mypara{Outcomes}
A run is \temph{complete} if it is infinite or ends in a terminal configuration. GLP transition systems additionally carry an \temph{outcome}: a domain-specific function from complete runs to an outcome space---for GLP, the computed answer (Section~\ref{sec:glp}). Outcomes serve only to state outcome-completeness of implementations (Definition~\ref{def:implementation-properties}).

\ifappendix
\mypara{Liveness}
The equivalence $\sim$ is the special case of the partial equivalence of~\cite{lewis2026volitional} in which every transition lies in the domain, hence every class carries a liveness obligation. It is instantiated for cGLP in Section~\ref{sec:glp} and for maGLP in Section~\ref{sec:maglp}; until then, liveness and correctness are as above.
\fi

\subsection{Implementations}
\label{subsec:implementations}

\begin{definition}[Implementation]
\label{def:implementation}
Given two transition systems $TS = (C, c_0, T, {\sim})$ (the \temph{specification}) and $TS' = (C', c'_0, T', {\sim'})$, an \temph{implementation of $TS$ by $TS'$} is a function $\tau : C' \rightarrow C$ where $\tau(c'_0) = c_0$, in which case the pair $(TS', \tau)$ is referred to as \temph{an implementation of $TS$}.

Given a computation $r' = c'_1 \rightarrow c'_2 \rightarrow \ldots$ of $TS'$, $\tau(r')$ is the (possibly empty) computation $\tau(c'_1) \rightarrow \tau(c'_2) \rightarrow \ldots$ obtained from the sequence by removing consecutively repetitive elements so that $\tau(r')$ has no \temph{stutter transitions} of the form $c \rightarrow c$.
\end{definition}

The implementation mapping $\tau$ (written $\sigma$ in~\cite{shapiro2021multiagent,lewis2026volitional}; we keep $\sigma$ for substitutions) need not preserve transitions: an implementation transition $c' \rightarrow d'$ may map to a stutter ($\tau(c') = \tau(d')$), or two consecutive specification transitions may correspond to a single implementation transition. Correctness of the implementation is captured by the conditions of Definition~\ref{def:implementation-properties}, which depend in turn on the outcome function defined above.

\begin{definition}[Correct and Outcome-Complete Implementation]
\label{def:implementation-properties}
The implementation $(TS', \tau)$ of $TS$ is:
\begin{itemize}
\item \temph{correct} if $\tau$ maps every correct run of $TS'$ to a correct run of $TS$
\item \temph{outcome-complete} if for every complete $TS$ run with outcome $O$, some complete $TS'$ run has outcome $O$
\end{itemize}
\end{definition}

Correctness carries two obligations on the non-stutter image $\tau(r')$: that it is a run of $TS$---its consecutive distinct configurations are transitions of $TS$---and that it is live. We require outcome-completeness rather than run-completeness (every correct $TS$ run being a $\tau$-image), because the agent-deterministic implementations below realise every specification outcome without matching the specification's reduction order.

\section{GLP and Its Implementation}
\label{sec:glp}

\subsection{GLP Syntax and Operational Semantics}
\label{subsec:concurrent-glp}

This section presents Grassroots Logic Programs (GLP), a concurrent logic programming language in which logic variables are partitioned into paired \emph{readers} and \emph{writers}.

\begin{definition}[GLP Variables]
\label{def:glp-variables}
Let $\calV$ denote the set of \temph{writers} (identifiers beginning with uppercase). Define $\calV? = \{X? \mid X \in \calV\}$, called \temph{readers}. The set of all GLP variables is $\hat\calV = \calV \cup \calV?$. A writer $X$ and its reader $X?$ form a \temph{variable pair}.
\end{definition}

\begin{definition}[GLP Terms and Goals]
\label{def:glp-syntax}
A \temph{term} is a variable in $\hat\calV$, a constant (numbers, strings, or the empty list \verb|[]|), or a compound term $f(T_1,\ldots,T_n)$ with functor $f$ and subterms $T_i$. Let $\calT$ denote the set of all terms. We use standard list notation: \verb=[X|Xs]= for a list cell, \verb|[X1,...,Xn]| for finite lists. A term is \temph{ground} if it contains no variables.

A \temph{unit goal} is a compound term or a string. A \temph{goal} is a multiset of unit goals; the empty goal is written \verb|true|. A \temph{clause} $A$~\verb|:-|~$B$ has head $A$ (a unit goal) and body $B$ (a goal); a \emph{unit clause} has empty body.
\end{definition}

\begin{definition}[Single-Occurrence (SO) Invariant]
\label{def:so-invariant}
A term, goal, or clause satisfies the \temph{single-occurrence (SO) invariant} if every variable occurs in it at most once.
\end{definition}

\begin{definition}[GLP Program]
\label{def:glp-program}
A clause $C$ satisfies the \temph{single-reader/single-writer (SRSW) restriction} if it satisfies SO and a variable occurs in $C$ iff its paired variable also occurs in $C$.
A \temph{GLP program} is a finite sequence of clauses satisfying SRSW; clauses for the same predicate form a \temph{procedure}.
The set of GLP goals $\hat\calG(P)$ includes all goals over $\hat\calV$ and the vocabulary of $P$ that satisfy SO.
\end{definition}

The SO and SRSW restrictions ensure that GLP eschews unification in favour of simple term matching. A guard that implies its argument is ground relaxes the single-reader half of SRSW for that reader: a ground value has no writer left to pair with, so further occurrences of the reader in the clause consume nothing~\cite{shapiro2025glp}. Program~\ref{program:monitor} relies on this.

\begin{definition}[Substitutions and Assignments~\cite{shapiro2025glp}]
\label{def:writers-assignment}
A GLP \temph{writer assignment} is a term of the form $X := T$, $X\in\calV$, $T\notin\calV$, satisfying SO. Similarly, a GLP \temph{reader assignment} is a term of the form $X? := T$, $X?\in\calV?$, $T\notin\calV$, satisfying SO. A \temph{writers (readers) substitution} $\sigma$ is the substitution implied by a set of writer (reader) assignments that jointly satisfy SO. Given a writers assignment $X := T$, its \temph{readers counterpart} is $X? := T$, and given a writers substitution $\sigma$, its \temph{readers counterpart} $\sigma?$ is the readers substitution defined by $X?\sigma? = X\sigma$. Given a reader assignment $X? := T$, its \temph{writers counterpart} is $X := T$, and given a readers substitution $\sigma$, its \temph{writers counterpart} $\sigma!$ is the writers substitution defined by $X\sigma! = X?\sigma$. The \temph{pair completion} of a readers substitution $\sigma$ is $\sigma^\star = \sigma \cup \sigma!$, applied to a fixed point.
\end{definition}
A writers substitution applied to a term assigns only writers; readers are left unchanged.

\begin{definition}[GLP Renaming, Renaming Apart]
\label{def:glp-renaming}
A \temph{GLP renaming} is an injective substitution $\rho: \hat\calV \to \hat\calV$ such that for each $X \in \calV$: $X\rho \in \calV$ and $X?\rho = (X\rho)?$. Two GLP terms \temph{have a variable in common} if for some writer $X \in \calV$, either $X$ or $X?$ occurs in both. A GLP renaming $\sigma$ \temph{renames $T'$ apart from} $T$ if $T'\sigma$ and $T$ have no variable in common.
\end{definition}

\begin{definition}[Writer MGU~\cite{shapiro2025glp}]
\label{def:writer-mgu}
Given two GLP unit goals $A$ and $H$, a \temph{writer mgu} is a writers substitution $\sigma$ such that $A\sigma = H\sigma$ and $\sigma$ is most general among such substitutions.
\end{definition}

Here ``most general'' is under the standard subsumption ordering on substitutions: $\sigma \leq \tau$ iff there exists $\theta$ such that $\sigma\theta = \tau$.
\ifappendix The operational algorithm for computing the writer mgu via term matching, including a table of all writer/reader/term combinations, is given below.\fi

If two terms $T_1$ and $T_2$ that jointly satisfy SO are unifiable with an mgu $\sigma$, then $\sigma$ maps any variable in $T_1$ to a subterm of $T_2$ and vice versa. Hence, the SO invariant of GLP allows eschewing unification in favour of \emph{term matching} that performs joint term-tree traversal and collects variable assignments along the way.
\ifappendix
\begin{definition}[Term Matching]
\label{def:term-matching}
Given two terms $T_1$ and $T_2$ that jointly satisfy SO, their \temph{term matching} proceeds via the joint traversal of the term-trees of $T_1$ and $T_2$, consulting the following table at each pair of joint vertices, where $X_1, X_2$ denote writers, $X_1?, X_2?$ denote readers, and $f/n$ denotes a non-variable term, a constant when $n=0$ and a compound term when $n>0$:
\begin{center}
\begin{tabular}{l|lll}
$T_1 \backslash T_2$ & Writer $X_2$ & Reader $X_2?$ & Term $f_2/n_2$ \\
\hline
Writer $X_1$ & fail & $X_1 := X_2?$ & $X_1 := T_2$ \\
Reader $X_1?$ & $X_2 := X_1?$ & fail & suspend on $X_1?$\\
Term $f_1/n_1$ & $X_2 := T_1$ & fail & fail if $f_1 \ne f_2$ or $n_1 \ne n_2$\\
\end{tabular}
\end{center}
The writer mgu is the union of all writer assignments if no \emph{fail} was encountered and the suspension set is empty.
\end{definition}
\fi

\begin{example}[Fair Merge]
\label{ex:merge}
The concurrent logic program fairly merges two streams:
\begin{verbatim}
merge([X|Xs], Ys, [X?|Zs?]) :- merge(Ys?, Xs?, Zs).
merge(Xs, [Y|Ys], [Y?|Zs?]) :- merge(Xs?, Ys?, Zs).
merge(Xs, [], Xs?).
merge([], Ys, Ys?).
\end{verbatim}
In the goal \verb=merge([1,2,3|In1?],[a,b|In2?],Out)=, the writer \verb|Out| is the output stream \verb|merge| produces, and the readers \verb|In1?|, \verb|In2?| are open input-stream tails that other processes extend by assigning their paired writers \verb|In1|, \verb|In2|. The first two clauses swap inputs in their recursive calls, ensuring fairness when both streams are available; when neither input has a list cell (both start with a reader), the goal suspends until one is assigned. Different scheduling of assignments to \verb|In1| and \verb|In2| may yield different output interleavings.
\end{example}

\ifappendix
\begin{example}[Writer MGU and Reduction]
\label{ex:writer-mgu}
Consider reducing the goal of Example~\ref{ex:merge} with its first \verb=merge= clause. The goal and clause have no variable in common, so no renaming apart is needed. Term matching proceeds position by position. In position~1, \verb=[1,2,3|In1?]= matches \verb=[X|Xs]=, yielding \verb|X := 1| and \verb!Xs := [2,3|In1?]!. In position~2, \verb=[a,b|In2?]= matches \verb|Ys|, yielding \verb!Ys := [a,b|In2?]!. In position~3, \verb|Out| matches \verb=[X?|Zs?]=, yielding \verb!Out := [X?|Zs?]!. The writer mgu is $\hat\sigma = \{$\verb|X|$:=$\verb|1|$,\;$\verb|Xs|$:=$\verb=[2,3|In1?]=$,\;$\verb|Ys|$:=$\verb=[a,b|In2?]=$,\;$\verb|Out|$:=$\verb=[X?|Zs?]=$\}$. Note that the readers \verb|In1?|, \verb|In2?|, \verb|X?|, \verb|Zs?| appear in assigned terms but are not themselves assigned---only writers are assigned by the writer mgu.
\end{example}
\fi

\begin{definition}[GLP Goal/Clause Reduction]
\label{def:glp-reduction}
Given GLP unit goal $A$ and clause $C$, with $H$ \verb|:-| $B$ being the result of the GLP renaming of $C$ apart from $A$, the \temph{GLP reduction} of $A$ with $C$ \temph{succeeds with result} $(B,\sigma)$ if $A$ and $H$ have a writer mgu $\sigma$.
\end{definition}

\ifappendix
GLP clauses may include \emph{guards}---tests that determine clause applicability~\cite{shapiro2025glp}.

\begin{definition}[Guarded Clause~\cite{shapiro2025glp}]
\label{def:guarded-clause}
A \temph{guarded clause} has the form $H$ \verb|:-| $G$ \verb"|" $B$, where $H$ is the head, $G$ is a conjunction of guard predicates, and $B$ is the body. Guard arguments are readers paired to head writers.
\end{definition}

Guards have three-valued semantics: each guard predicate defines its \emph{success} condition; a guard \emph{suspends} if it does not succeed but some instance of it under a readers substitution would succeed; it \emph{fails} if no such instance exists. A conjunction succeeds if all members succeed, suspends if any member suspends and none fail, and fails if any member fails. The GLP reduction of Definition~\ref{def:glp-reduction} is augmented to succeed if the guard also succeeds, and the suspension set of Definition~\ref{def:reduction-suspend-fail} below likewise collects the readers blocking a suspended guard. The guard and system-predicate reference appears in~\cite{shapiro2025glp}.
\fi

\begin{definition}[Goal Identity~\cite{shapiro2025glp}]
\label{def:goal-identity}
Each unit goal in a cGLP computation carries a unique \temph{identifier} assigned when it is spawned: the goals of the initial goal $G_0$ receive distinct identifiers, and each goal spawned by a reduction receives a fresh identifier. Instantiating a reader preserves the identifier of the goal in which it occurs.
\end{definition}

Preserving identity under reader instantiation lets equivalent Reduce transitions operate on differently instantiated forms of the same goal, on which the correctness results below rest.

The following definition provides a standard concurrent (interleaving-based) operational semantics for GLP, which we refer to as cGLP.

\begin{definition}[cGLP Transition System~\cite{shapiro2025glp}]
\label{def:cglp-ts}
Given a GLP program $P$, an \temph{asynchronous resolvent} over $P$ is a pair $(G, \sigma)$ where $G \in \hat\calG(P)$ and $\sigma$ is a readers substitution.

A transition system $cGLP = (\calC, c_0, \calT, {\sim})$ is a \temph{cGLP transition system} over $P$ and initial goal $G_0$ satisfying SO if:
\begin{enumerate}
    \item $\calC$ is the set of all asynchronous resolvents over $P$
    \item $c_0 = (G_0, \emptyset)$
    \item $\calT$ is the set of all transitions $(G, \sigma) \rightarrow (G', \sigma')$ satisfying either:
    \begin{enumerate}
        \item \textbf{Reduce:} there exists a unit goal $A \in G$ such that $C \in P$ is the first clause for which the GLP reduction of $A$ with $C$ succeeds with result $(B, \hat\sigma)$, $G' = (G \setminus \{A\} \cup B)\hat\sigma$, and $\sigma' = \sigma \circ \hat\sigma?$.
        \item \textbf{Communicate:} $\{X? := T\} \in \sigma$, $X?\in G$, $G' = G\{X? := T\}$, and $\sigma' = \sigma$.
    \end{enumerate}
    \item $\sim$, the \temph{cGLP transition equivalence}, relates $t_1 \sim t_2$ iff either both are Reduce transitions reducing the same goal (by identity, Definition~\ref{def:goal-identity}) with the same clause, or both are Communicate transitions applying the counterpart of the same writer assignment to the same goal (by identity)
\end{enumerate}
\end{definition}

cGLP Reduce uses a writer mgu instead of standard unification, and chooses the first applicable clause instead of any clause. The first is the fundamental use of GLP readers for synchronisation. The second allows writing fair concurrent programs.

The cGLP Communicate rule realises the use of reader/writer pairs for asynchronous communication: it communicates an assignment from its writer to its paired reader.

cGLP has two forms of nondeterminism: the choice of $A \in G$ (\emph{and-nondeterminism}), and the choice of which Communicate to apply.

\ifappendix
\begin{example}[Reduce and Communicate]
\label{ex:reduce-communicate}
Continuing Example~\ref{ex:writer-mgu}, the Reduce transition produces the new goal $G' = \{$\verb=merge(Ys?, Xs?, Zs)=$\}$ (the body of the clause, with $\hat\sigma$ applied---but $\hat\sigma$ assigns only writers, and the body contains only readers and the fresh writer \verb|Zs|, so the body is unchanged). The readers counterpart $\hat\sigma? = \{$\verb|Ys?|$:=$\verb=[a,b|In2?]=$,\;$\verb|Xs?|$:=$\verb=[2,3|In1?]=$,\;\ldots\}$ is added to $\sigma$.

Two Communicate transitions are now enabled: applying \verb=Ys?:=[a,b|In2?]= and \verb=Xs?:=[2,3|In1?]= to the goal (since both \verb|Ys?| and \verb|Xs?| occur in it). After both, the goal becomes \verb=merge([a,b|In2?], [2,3|In1?], Zs)=---ready for the next Reduce, which will match the second clause (swapping the inputs back).
\end{example}
\fi

\begin{definition}[Proper Run, Outcome~\cite{shapiro2025glp}]
\label{def:proper-run}
A cGLP run $r: (G_0,\sigma_0) \rightarrow \cdots \rightarrow (G_n, \sigma_n)$ is \temph{proper} if for any $1\le i< n$, a variable that occurs in $G_{i+1}$ but not in $G_i$ also does not occur in any $G_j$, $j<i$. If proper, the \temph{outcome} of $r$ is $(G_0$ \verb|:-| $G_n)\sigma_n^\star$.
\end{definition}
Pair completion (Definition~\ref{def:writers-assignment}) is employed since $\sigma_n$ accumulates only readers counterparts, leaving the writers of $G_0$ unbound.

\ifappendix
\begin{restatable}[Disjoint Substitution Commutativity]{lemma}{lemSubstitutionCommutativity}
\label{lem:substitution-commutativity}
Let $\hat\sigma_1$ and $\hat\sigma_2$ be writers substitutions produced by two cGLP Reduce transitions on distinct goals $A_1$ and $A_2$ in the same run. Then $\hat\sigma_1$ and $\hat\sigma_2$ assign disjoint sets of variables, and $\hat\sigma_1 \circ \hat\sigma_2 = \hat\sigma_2 \circ \hat\sigma_1$.
\end{restatable}

\begin{proof}
By the SO invariant, each writer occurs at most once in the goal multiset $G$. Since $A_1 \neq A_2$ are distinct goals, they have disjoint writers. Each Reduce assigns only writers from its reduced goal (by Definition~\ref{def:writer-mgu}). Therefore $\hat\sigma_1$ and $\hat\sigma_2$ have disjoint domains. Substitutions with disjoint domains commute.
\end{proof}

\begin{restatable}[Reduction Set Determines Outcome]{lemma}{lemReductionSetOutcome}
\label{lem:reduction-set-outcome}
Let $r_1$ and $r_2$ be two proper cGLP runs from $(G_0, \emptyset)$ that perform pairwise equivalent sets of Reduce transitions---for each reduction of a goal (by identity) with a clause in $r_1$, there is an equivalent reduction in $r_2$, and vice versa. Then $r_1$ and $r_2$ have the same outcome.
\end{restatable}

\begin{proof}
Let $r_1$ perform reductions in order $R_{\pi_1(1)}, \ldots, R_{\pi_1(k)}$ and $r_2$ in order $R_{\pi_2(1)}, \ldots, R_{\pi_2(k)}$ for permutations $\pi_1, \pi_2$, where $R_i$ and the corresponding reduction in the other run are equivalent (same goal by identity, same clause).

Since each writer mgu assigns only writers (Definition~\ref{def:writers-assignment}), equivalent Reduce transitions produce the same writer mgu $\hat\sigma_i$ regardless of which reader substitutions have been applied. By Lemma~\ref{lem:substitution-commutativity}, these substitutions have disjoint domains and commute. Therefore:
$\hat\sigma_{\pi_1(1)} \circ \cdots \circ \hat\sigma_{\pi_1(k)} = \hat\sigma_{\pi_2(1)} \circ \cdots \circ \hat\sigma_{\pi_2(k)}$

The final readers substitution $\sigma_n$ is the composition of the readers counterparts, which likewise commute.

The final goal $G_n$ consists of goals added by clause bodies (which depends only on which reductions occurred, not their order) with the accumulated substitution applied (which is the same by the above).

Communicate transitions apply reader assignments but do not affect the outcome: they do not change $\sigma_n$ and merely propagate existing assignments within goals. Therefore both runs have the same $\sigma_n$, hence the same pair completion $\sigma_n^\star$ and the same outcome $(G_0$ \verb|:-| $G_n)\sigma_n^\star$.
\end{proof}
\fi

\ifappendix
Two equivalent Reduce transitions may operate on different instantiations of the same goal---with different reader substitutions applied---but produce the same writer mgu (since the writer mgu assigns only writers and is unaffected by reader substitutions).

\begin{remark}[Properness]
\label{rem:properness}
The properness condition ensures fresh variables introduced by clause renaming are indeed fresh with respect to the entire computation history. This is guaranteed by the GLP renaming-apart requirement in Definition~\ref{def:glp-renaming}.
\end{remark}
\fi

A \emph{cGLP run} is a run of the cGLP transition system; it starts at $c_0$ and follows only cGLP transitions, hence is safe, collapsing Definition~\ref{def:ts-basic}'s run/safe-run distinction. Correctness of a cGLP run thus reduces to liveness.

Deterministic implementations of GLP refine reduction with explicit suspension and failure, and reactivate a suspended goal when one of its blocking readers is assigned.

\begin{definition}[Reduction Suspension and Failure]
\label{def:reduction-suspend-fail}
Given unit goal $A$ and clause $C$, with $H$ \verb|:-| $B$ being the result of the GLP renaming of $C$ apart from $A$:
\begin{itemize}
\item The reduction \temph{succeeds} if $A$ and $H$ have a writer mgu $\sigma$ (Definition~\ref{def:glp-reduction})
\item The reduction \temph{suspends with suspension set $W$} if no writer mgu exists, but there exists a set $W \subseteq \calV?$ of readers occurring in $A$ such that for some ground substitution $\theta$ on $W$, the terms $A\theta$ and $H$ have a writer mgu. $W$ is the set of readers of $A$ matched against a non-variable subterm of $H$.
\item The reduction \temph{fails} if it neither succeeds nor suspends
\end{itemize}
\end{definition}

\begin{definition}[Reactivation Set]\label{def:reactivation-set}
Given a set $S$ of pairs $(G, W)$, each a suspended goal $G$ with its suspension set $W \subseteq \calV?$ (Definition~\ref{def:reduction-suspend-fail}), and a readers substitution $\sigma?$, the \temph{reactivation set} is:
\[
\textrm{reactivate}(S, \sigma?) = \{G : (G, W) \in S \wedge \exists X? \in W.\; X?\sigma? \neq X?\}
\]
These are goals suspended on readers that have been instantiated by $\sigma?$.
\end{definition}

\ifappendix
\subsection{Deterministic Implementation of cGLP}
\label{subsec:dglp}

\mypara{Motivation}
dGLP serves as the formal specification for AI to implement GLP on workstations. The discipline is: from the mathematical definition (this section), AI derives an English+code specification, and from this specification derives Dart code.

While cGLP (Definition~\ref{def:cglp-ts}) is nondeterministic in goal selection and abstracts away suspension and failure, dGLP provides:
\begin{enumerate}
\item Deterministic FIFO scheduling of active goals
\item Explicit tracking of suspended goals with their suspension sets
\item Explicit tracking of failed goals
\item Automatic reactivation of suspended goals when blocking readers are instantiated
\end{enumerate}

The key simplification compared to maGLP is that all variables are local, so reader assignments can be applied immediately rather than through asynchronous communication.

\begin{example}[Suspension and Failure]
\label{ex:suspension}
Consider the goal \verb=merge(Xs?, Ys?, Zs)= where both input streams start with readers. Matching against the first clause head \verb=merge([X1|Xs1], Ys1, [X1?|Zs1?])= fails at position~1: \verb|Xs?| is a reader and \verb=[X1|Xs1]= is a compound term, so no writer mgu exists. However, if \verb|Xs?| were instantiated to a list cell (e.g., \verb=[1|Xs1?]=), a writer mgu would exist. Hence the reduction suspends with suspension set $W = \{$\verb|Xs?|$\}$. The goal is suspended until \verb|Xs?| is assigned.

Conversely, \verb|merge(hello, Ys?, Zs)| fails outright against all clauses: the constant \verb|hello| cannot match any clause head's first argument, regardless of what readers are instantiated to.
\end{example}

\begin{definition}[dGLP Configuration]\label{def:dglp-configuration}
A \temph{dGLP configuration} over program $P$ is a triple $(Q, S, F)$ where:
\begin{itemize}
\item $Q$ is a sequence (queue) of \temph{active} unit goals
\item $S$ is a set of \temph{suspended} unit goals, each paired with its suspension set of blocking readers
\item $F$ is a set of \temph{failed} unit goals
\end{itemize}
\end{definition}

\begin{definition}[dGLP Transition System]\label{def:dglp-ts}
The transition system $\textrm{dGLP}(P) = (\calC, c_0, \calT, {\mathrm{id}})$ over GLP program $P$ and initial goal $G_0$ satisfying SO is defined by:
\begin{enumerate}
\item $\calC$ is the set of all dGLP configurations over $P$
\item $c_0 = (G_0, \emptyset, \emptyset)$
\item $\calT$ is the set of all transitions $(Q, S, F) \rightarrow (Q', S', F')$ where $Q = A \cdot Q_r$ (with $\cdot$ denoting the queue with head element $A$ and remainder $Q_r$) and exactly one of the following holds:

\begin{enumerate}
\item \textbf{Reduce:} There exists a first clause $C \in P$ for which the GLP reduction of $A$ with $C$ (Definition~\ref{def:glp-reduction}) succeeds with $(B, \hat\sigma)$. Let $R = \textrm{reactivate}(S, \hat\sigma?)$. Then:
\[
Q' = (Q_r \cdot B \cdot R)\hat\sigma\hat\sigma?, \quad S' = S \setminus \{(G, W) : G \in R\}, \quad F' = F
\]

\item \textbf{Suspend:} No clause succeeds, but the GLP reduction of $A$ with at least one clause $C \in P$ suspends. Let $W = \bigcup_{C \in P} W_C$ where $W_C$ is the suspension set for clause $C$. Then:
\[
Q' = Q_r, \quad S' = S \cup \{(A, W)\}, \quad F' = F
\]

\item \textbf{Fail:} No clause succeeds and no clause suspends (all clauses fail outright). Then:
\[
Q' = Q_r, \quad S' = S, \quad F' = F \cup \{A\}
\]
\end{enumerate}
\item $\sim$ is the identity relation: dGLP is deterministic, so each transition forms its own class
\end{enumerate}
\end{definition}

\begin{remark}[Deterministic Scheduling]
\label{rem:dglp-deterministic}
The dGLP transition system is deterministic: given any non-terminal configuration, exactly one transition is enabled. Goal selection follows FIFO order from the active queue $Q$, and clause selection uses the first applicable clause as in cGLP.
\end{remark}

\begin{remark}[Immediate Reader Application]
\label{rem:dglp-immediate-reader}
In dGLP, reader substitutions $\hat\sigma?$ are applied immediately during reduction rather than through separate Communicate transitions. This is sound for single-agent execution where all variables are local.
\end{remark}

\mypara{dGLP Correctly implements cGLP}

\begin{definition}[Implementation Mapping $\tau$]
\label{def:dglp-sigma}
Given a dGLP computation $r' = c'_0 \rightarrow c'_1 \rightarrow \ldots \rightarrow c'_n$ where $c'_n = (Q, S, F)$, the implementation mapping $\tau$ is defined by:
$$\tau(c'_n) = (G, \sigma_r)$$
where $G = Q \cup \{A : (A, W) \in S\} \cup F$ is the multiset union of active, suspended, and failed goals, and $\sigma_r = \hat\sigma_1? \circ \hat\sigma_2? \circ \ldots \circ \hat\sigma_k?$ is the composition of all readers substitutions produced by Reduce transitions along $r'$.
\end{definition}

\begin{restatable}[$\tau$ is Correct]{lemma}{lemdGLPLive}
\label{lem:dglp-live}
The implementation mapping $\tau$ is correct.
\end{restatable}

\begin{proof}
Since dGLP is deterministic (Remark~\ref{rem:dglp-deterministic}), every complete dGLP run is trivially live---there is at most one enabled transition at any configuration, and it is taken.

\emph{Reduce:} Since dGLP's FIFO scheduling eventually processes every goal that enters the queue, for every cGLP Reduce equivalence class $[t]$ that becomes enabled, eventually some equivalent reduction is taken in $\tau(r')$.

\emph{Communicate:} In dGLP, reader substitutions $\hat\sigma?$ are applied immediately during Reduce (Remark~\ref{rem:dglp-immediate-reader}). Each such application corresponds to the cGLP Communicate equivalence classes that would apply those reader assignments. Hence every enabled cGLP Communicate equivalence class is taken when the corresponding dGLP Reduce fires.

Therefore $\tau$ maps correct dGLP runs to correct cGLP runs.
\end{proof}

While dGLP cannot match cGLP's arbitrary reduction order step-by-step (cGLP is nondeterministic in goal selection; dGLP uses FIFO), it achieves the same \emph{outcomes}. This follows from the commutativity of disjoint substitution compositions.

\begin{restatable}[dGLP Performs Equivalent Reductions]{lemma}{lemdGLPSameReductions}
\label{lem:dglp-same-reductions}
For every proper cGLP run $r$ that performs Reduce transitions $\{R_1, \ldots, R_k\}$, there exists a dGLP run $r'$ that performs a pairwise equivalent set of Reduce transitions.
\end{restatable}

\ifappendix
\begin{proof}
By induction on spawn order. Goals in the initial goal $G_0$ are in the initial dGLP queue with the same identifiers. For any $R_i$ reducing goal $A_i$ (by identity) with clause $C_i$: if $A_i$ was spawned by an earlier reduction $R_j$, by induction dGLP performs an equivalent reduction, so $A_i$ enters the queue with the same identity. Since $R_i$ succeeds in $r$, $A_i$ and the head of $C_i$ have a writer mgu. Since the writer mgu assigns only writers (Definition~\ref{def:writers-assignment}), this writer mgu is invariant under any reader substitutions applied to $A_i$, so an equivalent Reduce---reducing $A_i$ (by identity) with $C_i$---is enabled in dGLP. When $A_i$ reaches the queue head, the equivalent reduction succeeds with the same writer mgu.
\end{proof}
\else
\proofref
\fi

\begin{restatable}[dGLP Outcome-Completeness]{lemma}{lemdGLPComplete}
\label{lem:dglp-complete}
For every proper cGLP run $r$ with outcome $O$, there exists a dGLP run $r'$ with outcome $O$.
\end{restatable}

\begin{proof}
Let $r$ be a proper cGLP run with outcome $O = (G_0$ \verb|:-| $G_n)\sigma_n^\star$.

Let $r$ perform Reduce transitions $\{R_1, \ldots, R_k\}$.

By Lemma~\ref{lem:dglp-same-reductions}, there exists a dGLP run $r'$ that performs a pairwise equivalent set of reductions, possibly in a different order and with interleaved Suspend/reactivation steps.

By Lemma~\ref{lem:reduction-set-outcome}, since $r$ and $r'$ perform pairwise equivalent reductions, they have the same outcome $O$.
\end{proof}

\begin{theorem}[dGLP Correctly Implements cGLP]
\label{thm:dglp-implements-glp}
The implementation $(dGLP, \tau)$ of cGLP is correct and outcome-complete.
\end{theorem}
\begin{proof}
\emph{Correct:} By Lemma~\ref{lem:dglp-live}. \emph{Outcome-complete:} By Lemma~\ref{lem:dglp-complete}.
\end{proof}

\begin{remark}[Step-by-Step vs.\ Outcome Equivalence]
\label{rem:dglp-outcome}
Step-by-step completeness would require dGLP to match cGLP's transitions in order, which is impossible due to cGLP's nondeterministic goal selection vs.\ dGLP's FIFO scheduling. Outcome-completeness is the appropriate notion: what matters for program correctness is the final result (the clause derived and final substitution), not the intermediate reduction order.
\end{remark}

\mypara{Bounded Tail Reduction}
Strict FIFO appends every body goal to the back of $Q$ (Definition~\ref{def:dglp-ts}), so a tail-recursive goal---one whose body is a single continuation---is enqueued and dequeued at every reduction. A workstation implementation avoids this churn by reducing the continuation in place, ahead of the queue; left unbounded, this lets one goal monopolise the engine and starve the goals already queued behind it, breaking liveness. dGLP therefore admits a \emph{bounded} tail optimisation, realised in the bytecode by the \texttt{requeue} instruction (Section~\ref{app:cf-instructions}).

\begin{definition}[$b$-Bounded Tail Schedule]
\label{def:bounded-tail}
Fix $b \geq 1$. In a Reduce of the queue head $A$, at most one goal of the body $B$ is designated the \temph{tail goal} $A'$; the remaining goals of $B$ and the reactivated goals $R$ are appended to the back of $Q$ as in Definition~\ref{def:dglp-ts}. A \temph{$b$-bounded tail schedule} relaxes the FIFO order of Remark~\ref{rem:dglp-deterministic} only as follows: the tail goal $A'$ may be reduced ahead of the goals standing before it in $Q_r$, for at most $b$ such tail reductions in succession; after $b$ consecutive tail reductions the next tail goal is appended to the back of $Q$, restoring FIFO.
\end{definition}

The bound is finite but otherwise free: its value trades scheduling overhead against the latency the queued goals see, and carries no semantic weight; the implementation fixes it (currently $b = 26$). At $b = 1$ the schedule is strict FIFO---each tail goal is reduced once and then yields---correct, but forgoing the optimisation.

\begin{corollary}[Bounded Tail Reduction Preserves Correctness]
\label{cor:bounded-tail}
For every $b \geq 1$, $(dGLP, \tau)$ confined to $b$-bounded tail schedules is a correct and outcome-complete implementation of cGLP.
\end{corollary}
\begin{proof}
A $b$-bounded tail schedule performs exactly the Reduce transitions of the corresponding strict-FIFO run, only in a different order: each tail goal is reduced earlier than its FIFO position, and every goal standing behind it is delayed by at most $b$ further reductions. Outcome-completeness is therefore immediate from Lemma~\ref{lem:reduction-set-outcome}, which depends on the set of reductions performed and not on their order. For correctness (liveness), $b$ is finite, so the in-place continuation always yields within $b$ steps; hence every goal that enters $Q$ reaches the head in finitely many transitions and every enabled Reduce class is eventually taken, as in Lemma~\ref{lem:dglp-live}.
\end{proof}

\fi


\section{Multiagent Transition Systems}
\label{sec:multiagent-ts}

This section recalls the multiagent transition systems and atomic transactions needed to specify maGLP, drawing on the framework of volitional multiagent atomic transactions~\cite{lewis2026volitional} and the multiagent atomic transactions framework on which it builds~\cite{shapiro2025atomic}.

We assume a potentially infinite set of \emph{agents} $\Pi$, but consider only finite subsets of it, so when we refer to a particular set of agents $P \subset \Pi$ we assume $P$ to be nonempty and finite. We use $\subset$ to denote the strict subset relation and $\subseteq$ when equality is also possible.

\ifappendix
\begin{remark}[Agent Names]
\label{rem:agent-names}
Over a real network, the name of an agent $p \in \Pi$ is its public key, held by the networking layer (Appendix~\ref{app:in-networking}); symbolic agent names are the simulation realisation's choice.
\end{remark}
\fi

We use $S^P$ to denote the set $S$ indexed by the set $P$, and if $c \in S^P$ we use $c_p$ to denote the member of $c$ indexed by $p \in P$.

\begin{definition}[Local States, Configuration, Transaction, Participants]
\label{def:transaction}
Given agents $Q \subset \Pi$ and an arbitrary set $S$ of \temph{local states}, a \temph{configuration} over $Q$ and $S$ is a member of $C := S^Q$. An \temph{atomic transaction}, or just \emph{transaction}, over $Q$ and $S$ is any pair of configurations $t = c \rightarrow c' \in C^2$ such that $c \ne c'$, with $t_p := c_p \rightarrow c'_p$ for any $p \in Q$, and with $p$ being an \temph{active participant} in $t$ if $c_p \ne c'_p$, \temph{stationary participant} otherwise.
\end{definition}

\subsection{Multiagent Transition Systems}
\label{subsec:mats}

\begin{definition}[Multiagent Transition System]
\label{def:mts}
Given agents $P \subset \Pi$ and an arbitrary set $S$ of \temph{local states} with a designated \temph{initial local state} $s_0 \in S$, a \temph{multiagent transition system} over $P$ and $S$ is a transition system $TS = (C, c_0, T, {\sim})$ with \temph{configurations} $C := S^P$, \temph{initial configuration} $c_0 := \{s_0\}^P$, \temph{transitions} $T \subseteq C^2$ being a set of transactions over $P$ and $S$, and $\sim$ an equivalence on $T$.
\end{definition}

\begin{definition}[Transaction Closure]
\label{def:closure}
Let $P \subset \Pi$, $S$ a set of local states, and $C := S^P$. For a transaction $t = (c \rightarrow c')$ over local states $S$ with participants $Q \subseteq P$, the \temph{$P$-closure of $t$}, $t{\uparrow}P$, is the set of transitions over $P$ and $S$ defined by:
$$
t{\uparrow}P := \{ t' \in C^2 : \forall q \in Q.(t_q = t'_q) \wedge \forall p \in P \setminus Q.(p \text{ is stationary in } t') \}
$$
If $R$ is a set of transactions, each $t \in R$ over some $Q \subseteq P$ and $S$, then the \temph{$P$-closure of $R$}, $R{\uparrow}P$, is the set of transitions over $P$ and $S$ defined by:
$$
R{\uparrow}P := \bigcup_{t \in R} t{\uparrow}P
$$
Given a relation ${\sim}$ on $R$, its \temph{$P$-closure} ${\sim}{\uparrow}P$ is the relation on $R{\uparrow}P$ with $\hat t \mathrel{({\sim}{\uparrow}P)} \hat t'$ iff $\hat t \in t{\uparrow}P$ and $\hat t' \in t'{\uparrow}P$ for some $t \sim t'$.
\end{definition}

\ifappendix
The closure carries an equivalence on transactions to one on the induced transitions: since distinct transactions over the same participants have disjoint closures, ${\sim}{\uparrow}P$ is an equivalence whenever ${\sim}$ is. A transaction equivalence relates only transactions with the same participants~\cite{lewis2026volitional}; the maGLP and madGLP equivalences below satisfy this.
\fi

\ifappendix
\begin{example}[Configurations, Transactions, Closure]
\label{ex:closure}
Let $P = \{p, q, o\}$ and $S = \{s_0, s_1\}$, writing a configuration as $(c_p, c_q, c_o)$. Let $t = c \rightarrow c'$ be the transaction over $Q = \{p, q\}$ with $c_p = c_q = s_0$ and $c'_p = c'_q = s_1$; both participants are active. Its $P$-closure $t{\uparrow}P$ has two transitions, one per local state of the stationary $o$: $(s_0, s_0, s_0) \rightarrow (s_1, s_1, s_0)$ and $(s_0, s_0, s_1) \rightarrow (s_1, s_1, s_1)$.
\end{example}
\fi

\begin{definition}[Transactions-Based Multiagent Transition System]
\label{def:tbmts}
Given agents $P \subset \Pi$, local states $S$ with initial local state $s_0 \in S$, a set of transactions $R$ each $t \in R$ over some $Q \subseteq P$ and $S$, and an equivalence $\sim$ on $R$, the \temph{transactions-based multiagent transition system} over $P$, $S$, $R$, and $\sim$ is the multiagent transition system $TS = (S^P, \{s_0\}^P, R{\uparrow}P, {\sim}{\uparrow}P)$.
\end{definition}

\section{Multiagent GLP and Its Implementation}\label{sec:maglp}

\subsection{Multiagent GLP}
\label{subsec:maglp}

This section recalls maGLP---the multiagent operational semantics of GLP~\cite{shapiro2025glp}---a transactions-based multiagent transition system (Section~\ref{sec:multiagent-ts}).

\mypara{From GLP to Multiagent GLP}
\ifappendix
In extending GLP to multiple agents, each agent maintains its own asynchronous resolvent as its local state. GLP's variable pairs provide natural asynchronous channels for inter-agent communication: when agent $p$ assigns a writer $X$ for which the paired reader $X?$ is held by agent $q$, the assignment $X := T$ must be communicated to $q$.

A key difference between cGLP and maGLP is in the initial state. In a multiagent transition system all agents must have the same initial local state $s_0$ (Definition~\ref{def:mts}). This precludes setting up an initial configuration in which agents share logic variables, as this would imply different initial states for different agents.

We resolve this in two steps. First, we employ only anonymous logic variables ``\verb|_|'' in the initial local states of agents: Anonymous variables are, on the one hand, syntactically identical, hence allow all initial states to be syntactically identical, and on the other hand represent unique variables, hence semantically all initial goals have unique, local, non-shared variables. The initial state of all agents is the atomic goal \verb|agent(ch(_?,_),ch(_?,_))|, providing two bidirectional channels: the first to the person operating the machine, and the second to the network. The first channel is a bidirectional stream of GLP messages between the machine and the person, mediated by runtime UI support below the GLP layer and by physical UI hardware; messages flowing toward the person are rendered as UI elements, and the person's responses flow back as GLP messages. With appropriate UI runtime support, a GLP message to the person may contain a writer, in which case the runtime presents it as a question and an eventual response from the person becomes the corresponding assignment. The second channel carries communication with other agents.

Second, the Cold-call transaction enables agents to bootstrap communication by establishing shared variables through the network infrastructure, realising the cold-call protocol for connecting previously-disconnected agents.
\fi

\begin{definition}[Multiagent GLP~\cite{shapiro2025glp}]\label{definition:maGLP}
Given agents $P\subset \Pi$ and GLP program $M$, the \temph{maGLP transition system} over $P$ and $M$ is the transactions-based multiagent transition system (Definition~\ref{def:tbmts}) over $P$, local states being asynchronous resolvents $(G_p, \sigma_p)$ over $M$, initial local state $s_0 = (\{\texttt{agent(ch(\_?,\_),\allowbreak ch(\_?,\_))}\}, \emptyset)$, the following transactions $c\rightarrow c'$, and the equivalence $\sim$ given below:
\begin{enumerate}
    \item \textbf{Reduce $p$:} A unary transaction with participant $p$ where $c_p\rightarrow c'_p$ is a cGLP Reduce transition (Definition~\ref{def:cglp-ts}).

    \item \textbf{Communicate $p$ to $q$:} A transaction with participants $p,q\in P$ where $\{X?:=T\} \in \sigma_p$, $X?$ occurs in $G_q$, $c'_p=(G_p,\sigma'_p)$, $\sigma'_p = \sigma_p \setminus \{X?:=T\}$, and $c'_q=(G_q\{X?:=T\},\sigma_q)$.

    \item \textbf{Cold-call $p$ to $q$:} A binary transaction with participants $p\ne q \in P$ where the network output stream in $c_p$ has a new message \verb|msg|$(q,X)$, $c'_p$ is the result of advancing the network output stream in $c_p$, and $c'_q$ is the result of adding \verb|msg|$(q,X)$ to the network input stream in $c_q$.

    \item ${\sim}$, the \temph{maGLP transaction equivalence} on the transactions above, relates $t_1 \sim t_2$ iff both are Reduce transactions at the same agent $p$ reducing the same goal (by identity, Definition~\ref{def:goal-identity}) with the same clause, both are Communicate transactions from the same $p$ to the same $q$ applying the counterpart of the same writer assignment to the same goal (by identity), or both are Cold-call transactions from the same $p$ to the same $q$ delivering the same message.
\end{enumerate}
\end{definition}

In every transaction the non-participating agents are stationary (Definition~\ref{def:closure}).

\ifappendix
\begin{example}[Cold-call]
\label{ex:cold-call}
With $G_p$ writing $\mathtt{msg}(q,\mathtt{intro}(X?))$ on $p$'s network output stream, Cold-call $p$ to $q$ advances that stream at $p$ and appends $\mathtt{msg}(q,\mathtt{intro}(X?))$ to $q$'s network input stream; the pair $(X,X?)$ is now shared, $X$ at $p$ and $X?$ at $q$, and subsequent Communicate transactions use it.
\end{example}
\fi

Note that Communicate may be unary or binary, depending on whether $p=q$.  Communicate transfers assignments from writers to readers within an agent or between agents sharing a paired reader and writer.  Cold-call transfers a term with its variables to $q$ through the network streams established in each agent's initial configuration, enabling the creation of paired variables among previously-disconnected agents.

The worked examples below start from application-level goals such as \texttt{client1(Xs)@p} and \texttt{monitor(Xs?)@q}, reachable from the universal initial state \texttt{agent(ch(\_?,\_),\allowbreak ch(\_?,\_))} by a boot phase.

\begin{example}[A maGLP Round-Trip]
\label{ex:maglp-roundtrip}
Suppose agents $p$ and $q$ already share a variable pair $(X, X?)$, with $X$ held by $p$ and $X?$ held by $q$, and that $q$'s resolvent contains a goal that consumes $X?$ (and so suspends until $X?$ is assigned). A maGLP run proceeds: (\ia)~Reduce $p$ assigns $X := T$ in $\sigma_p$; (\ib)~Communicate $p$ to $q$ delivers the assignment to $q$, removes it from $\sigma_p$, and substitutes $T$ for $X?$ in $G_q$; (\ic)~Reduce at $q$ now proceeds. Detailed traces, including the friend-mediated introduction of two agents that eschews Cold-call, appear in \appref{app:madglp-trace}{\arxivref}.
\end{example}

\ifappendix
\begin{remark}[Cold-call as the Exception]
\label{rem:cold-call-exception}
Cold-call is the exceptional means by which agents establish paired variables: once two agents share a pair, they can communicate via Communicate indefinitely; moreover, an agent with two friends (with which it shares channels) may introduce them to each other through its existing channels, eschewing Cold-call.
\end{remark}

\begin{remark}[Transaction Degrees]
\label{rem:maglp-transaction-degrees}
Reduce is unary; Communicate is binary, unary when $p = q$; Cold-call is binary.
\end{remark}
\fi

A \emph{maGLP run} is a run of the maGLP transition system, hence safe, collapsing Definition~\ref{def:ts-basic}'s run/safe-run distinction. Correctness of a maGLP run thus reduces to liveness.

\begin{definition}[maGLP Proper Run, Outcome~\cite{shapiro2025glp}]
\label{def:maglp-proper-run}
An maGLP run is \temph{proper} if the run at each agent is proper (Definition~\ref{def:proper-run}). The \temph{outcome} of a proper maGLP run is the tuple, indexed by $P$, of the outcomes of the per-agent runs (Definition~\ref{def:proper-run}).
\end{definition}

\begin{proposition}[maGLP SO Preservation~\cite{shapiro2025glp}]
\label{lem:maglp-so-preservation}
If the initial goals of all agents satisfy SO, then every goal in every agent's resolvent throughout a proper run satisfies SO.
\end{proposition}

\subsection{Implementing Multiagent GLP with Deterministic Agents}
\label{subsec:madglp}

\mypara{Motivation}
madGLP serves as the formal specification for AI to implement maGLP on smartphones communicating peer-to-peer, without reliance on central servers. The discipline is: from the mathematical definition, AI derives an English+code specification, and from this specification derives Dart code.

This section specifies madGLP, an agent-deterministic transition system that implements maGLP by replacing each cross-agent shared variable pair with two local pairs---the reader of one linked to the writer of the other by a \emph{global link}, realised through a global writers table and message passing, with forwarding handled by spawned GLP goals. A value crosses a link once, to the reader's current holder, and the holder's acknowledgement releases the link.

\mypara{Realising the UI Channel}
The maGLP initial state provides two bidirectional channels (Section~\ref{subsec:maglp}): one to the person operating the machine, one to the network. madGLP implements the second via the global-link machinery described below. The first is realised by the Dart runtime below the madGLP layer, which bridges GLP messages and the physical UI hardware: outgoing GLP messages on this channel are rendered as UI elements, and the person's responses, captured by event handlers, flow back as GLP messages. When a message to the person contains a writer, the runtime presents it as a prompt and treats the person's response as the corresponding assignment, performed on the agent's local state.

\mypara{Implementing Shared Variables via Local Pairs}
A maGLP shared variable pair $(X, X?)$ with writer $X$ at agent $p$ and reader $X?$ at agent $q$ can be implemented by two local pairs, the reader of $p$'s pair linked to the writer of $q$'s pair by a \emph{global link}:

\begin{itemize}
\item At agent $p$: a local pair $(X_p, X_p?)$ where both variables remain in $p$'s resolvent
\item At agent $q$: a local pair $(X_q, X_q?)$ where both variables remain in $q$'s resolvent
\item A global link connecting the reader $X_p?$ to the writer $X_q$, realised as a \texttt{global\_send} goal at $p$ that monitors $X_p?$ and an entry in $q$'s global writers table that maps the global name to $X_q$
\end{itemize}

When $X_p$ is assigned a term $T$ at agent $p$, the value becomes available via $X_p?$. A spawned \texttt{global\_send} goal detects this and sends the assignment message to $q$, where $T$ is \emph{localized} (Definition~\ref{def:localize}). Upon receipt, $q$ looks up the target writer in its global writers table, assigns $X_q := T_q^\downarrow$, and removes the entry. This makes the value available through $X_q?$ in $q$'s resolvent. Here $T^\uparrow$ and $T_q^\downarrow$ denote the \emph{globalized} and \emph{localized} forms of a term---local variables replaced by global names, and conversely. Beside its global writers table, each agent $p$ keeps a set $U_p$ of \emph{imported-writer records}: a record $(Y, o, k, s)$ marks the local writer $Y$ as the sending end of the link anchored by $\_w(o, k)$, with $s$ the agent the name came from.

\begin{definition}[Globalize]\label{def:globalize}
Given agent $p$, remote agent $q$, and term $T$, the \temph{globalization by $p$}, written $T_p^\uparrow$, may update the global writers table $W'_p$ and the imported-writer records $U'_p$, and spawn or remove goals in $p$'s resolvent, as follows. For each variable $Y$ occurring in $T$:
\begin{enumerate}
\item If $Y$ is a writer: allocate the next index $i$, create entry $(Y, q)$ at index $i$ in $W'_p$, and replace $Y$ in $T_p^\uparrow$ with $\_w(p, i)$.
\item If $Y?$ is a reader: allocate the next index $i$, replace $Y?$ in $T_p^\uparrow$ with $\_r(p, i)$, and spawn goal $\texttt{global\_send}(Y?, \_r(p,i), q)$ into $p$'s resolvent.
\item If $Y?$ is the reader of an imported pair---its paired writer carries a global writers table entry $(Y, o, k)$ with anchor $o$---then instead replace $Y?$ with the original name $\_r(o, k)$ and remove the entry: the name is \emph{forwarded}, and no goal is spawned.
\item If $Y$ is the writer of an imported pair---$U_p$ has the record $(Y, o, k, s)$ with anchor $o$---then instead replace $Y$ in $T_p^\uparrow$ with the original name $\_w(o, k)$, remove the record and the \texttt{global\_send} goal watching $Y?$, and create no entry: the name is \emph{forwarded}, and $p$ leaves the link. Single occurrence puts this case out of reach once the value has been sent: an assigned writer has no further occurrence to export.
\end{enumerate}
\end{definition}

Globalize and Localize follow a uniform principle: a global writers table entry is created at the agent that will receive an assignment on the link, and a \texttt{global\_send} goal is spawned at the agent that will send an assignment. Every global link has exactly one entry on one end and exactly one \texttt{global\_send} goal on the other. An imported end is marked---a reader by its entry, a writer by its record---so re-exporting it forwards the original name instead of anchoring a new link at the exporter.

\begin{definition}[Localize]\label{def:localize}
Given agent $q$, remote agent $p$, and globalized term $T_p^\uparrow$, the \temph{localization by $q$}, written $T_q^\downarrow$, may update the global writers table $W'_q$ and the imported-writer records $U'_q$, and spawn goals into $q$'s resolvent, as follows. For each global name in $T_p^\uparrow$:
\begin{enumerate}
\item If $\_w(p, i)$: create fresh local pair $(Y_q, Y_q?)$, replace $\_w(p, i)$ with $Y_q$ (the writer) in $T_q^\downarrow$, add the record $(Y_q, p, i, s)$ to $U'_q$ with $s$ the sending agent, and spawn goal $\texttt{global\_send}(Y_q?, \_w(p,i), p)$ into $q$'s resolvent.\ifappendix{} If the sending agent is not $p$, the name was forwarded, and the value that goal produces is held (\Cref{app:held-links}).\fi
\item If $\_r(p, i)$: create fresh local pair $(Z_q, Z_q?)$, allocate the next index $k$ in $W'_q$, add entry $(Z_q, p, i)$, and replace $\_r(p, i)$ with $Z_q?$ (the reader) in $T_q^\downarrow$. If the sending agent is not $p$, the name was forwarded, and $q$ adds the request $(\texttt{req}(\_r(p,i)), p)$ to its outgoing messages---the request is how the anchor $p$ learns the reader's current holder.\ifappendix{} The request is held (\Cref{app:held-links}).\fi
\end{enumerate}
\end{definition}

This design keeps all variables in the resolvent, enabling correct behaviour even when both ends of a variable pair are exported to the same or different agents. A re-exported reader or writer travels under its original name, so a link keeps one anchor however far its end moves.


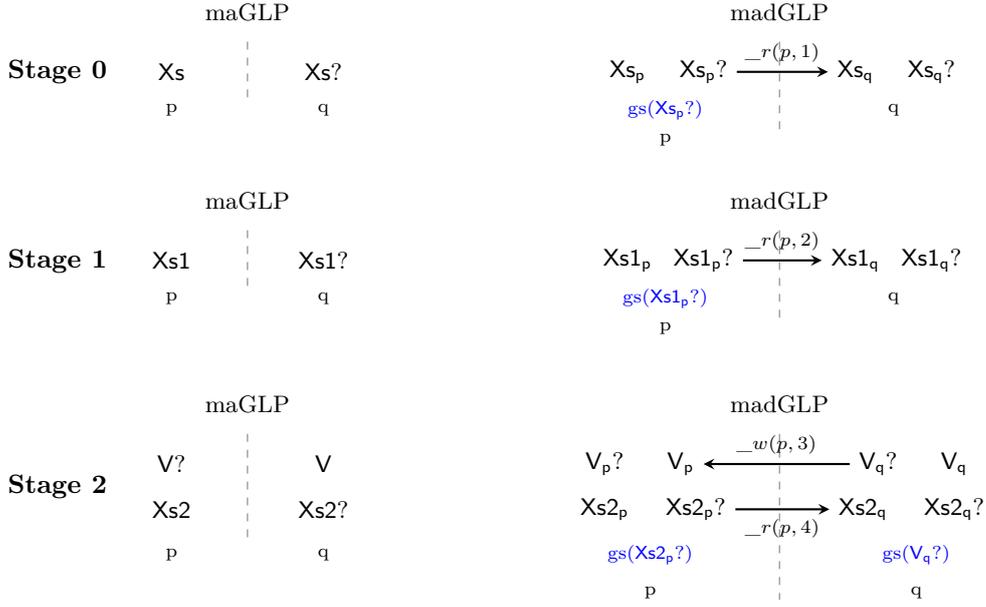
\begin{figure*}[t]
\centering
\resizebox{0.75\textwidth}{!}{%
\begin{tikzpicture}[
    var/.style={font=\small},
    gvar/.style={font=\scriptsize, text=blue},
    link/.style={->, >=stealth, thick},
    stage/.style={font=\bfseries},
    label/.style={font=\small},
    bound/.style={font=\small\itshape, text=gray},
    goal/.style={font=\scriptsize, text=blue},
    sep/.style={dashed, gray}
]

\node[stage] at (-6, 0) {Stage 0};

\node[label] at (-3.5, 0.8) {maGLP};
\draw[sep] (-3.5, 0.4) -- (-3.5, -0.4);
\node[var] (p0mX) at (-4.5, 0) {$\mathsf{Xs}$};
\node[var] (q0mX) at (-2.5, 0) {$\mathsf{Xs?}$};
\node[label] at (-4.5, -0.5) {\scriptsize p};
\node[label] at (-2.5, -0.5) {\scriptsize q};

\node[label] at (3.5, 0.8) {madGLP};
\draw[sep] (3.5, 0.4) -- (3.5, -0.8);
\node[var] (p0dXw) at (1.5, 0) {$\mathsf{Xs_p}$};
\node[var] (p0dXr) at (2.5, 0) {$\mathsf{Xs_p?}$};
\node[var] (q0dXw) at (4.5, 0) {$\mathsf{Xs_q}$};
\node[var] (q0dXr) at (5.5, 0) {$\mathsf{Xs_q?}$};
\draw[link] (p0dXr) -- node[above, font=\scriptsize] {$\_r(p,1)$} (q0dXw);
\node[goal] at (2, -0.5) {gs($\mathsf{Xs_p?}$)};
\node[label] at (2, -0.9) {\scriptsize p};
\node[label] at (5, -0.5) {\scriptsize q};

\node[stage] at (-6, -2.5) {Stage 1};

\node[label] at (-3.5, -1.7) {maGLP};
\draw[sep] (-3.5, -2.1) -- (-3.5, -2.9);
\node[var] (p1mX) at (-4.5, -2.5) {$\mathsf{Xs1}$};
\node[var] (q1mX) at (-2.5, -2.5) {$\mathsf{Xs1?}$};
\node[label] at (-4.5, -3) {\scriptsize p};
\node[label] at (-2.5, -3) {\scriptsize q};

\node[label] at (3.5, -1.7) {madGLP};
\draw[sep] (3.5, -2.1) -- (3.5, -3.3);
\node[var] (p1dXw) at (1.5, -2.5) {$\mathsf{Xs1_p}$};
\node[var] (p1dXr) at (2.5, -2.5) {$\mathsf{Xs1_p?}$};
\node[var] (q1dXw) at (4.5, -2.5) {$\mathsf{Xs1_q}$};
\node[var] (q1dXr) at (5.5, -2.5) {$\mathsf{Xs1_q?}$};
\draw[link] (p1dXr) -- node[above, font=\scriptsize] {$\_r(p,2)$} (q1dXw);
\node[goal] at (2, -3) {gs($\mathsf{Xs1_p?}$)};
\node[label] at (2, -3.4) {\scriptsize p};
\node[label] at (5, -3) {\scriptsize q};

\node[stage] at (-6, -5.5) {Stage 2};

\node[label] at (-3.5, -4.4) {maGLP};
\draw[sep] (-3.5, -4.8) -- (-3.5, -6.2);
\node[var] (p2mV) at (-4.5, -5.2) {$\mathsf{V?}$};
\node[var] (p2mX) at (-4.5, -5.8) {$\mathsf{Xs2}$};
\node[var] (q2mV) at (-2.5, -5.2) {$\mathsf{V}$};
\node[var] (q2mX) at (-2.5, -5.8) {$\mathsf{Xs2?}$};
\node[label] at (-4.5, -6.4) {\scriptsize p};
\node[label] at (-2.5, -6.4) {\scriptsize q};

\node[label] at (3.5, -4.4) {madGLP};
\draw[sep] (3.5, -4.8) -- (3.5, -7);
\node[var] (p2dVr) at (1.2, -5.2) {$\mathsf{V_p?}$};
\node[var] (p2dVw) at (2.2, -5.2) {$\mathsf{V_p}$};
\node[var] (q2dVr) at (4.8, -5.2) {$\mathsf{V_q?}$};
\node[var] (q2dVw) at (5.8, -5.2) {$\mathsf{V_q}$};
\node[var] (p2dXw) at (1.2, -5.8) {$\mathsf{Xs2_p}$};
\node[var] (p2dXr) at (2.4, -5.8) {$\mathsf{Xs2_p?}$};
\node[var] (q2dXw) at (4.6, -5.8) {$\mathsf{Xs2_q}$};
\node[var] (q2dXr) at (5.8, -5.8) {$\mathsf{Xs2_q?}$};
\draw[link] (q2dVr) -- node[above, font=\scriptsize] {$\_w(p,3)$} (p2dVw);
\draw[link] (p2dXr) -- node[below, font=\scriptsize] {$\_r(p,4)$} (q2dXw);
\node[goal] at (1.8, -6.4) {gs($\mathsf{Xs2_p?}$)};
\node[goal] at (5.3, -6.4) {gs($\mathsf{V_q?}$)};
\node[label] at (1.8, -6.9) {\scriptsize p};
\node[label] at (5.3, -6.9) {\scriptsize q};

\end{tikzpicture}%
}%
\caption{maGLP vs madGLP variable linking through three stages. Left: maGLP shared variable pairs. Right: madGLP local pairs with global links (arrows from reader to writer, labelled with global names). Dashed lines separate agents $p$ and $q$. ``gs'' denotes \texttt{global\_send} goals watching readers. In Stage 2, the V link arrow points from $q$ to $p$: when $q$ assigns $\mathsf{V_q}$, the goal gs($\mathsf{V_q?}$) sends the value back to $p$.}
\label{fig:madglp-example}
\end{figure*}

\begin{definition}[Global Link]\label{def:global-link-detail}
A \temph{global link} between agents $p$ and $q$ implements a maGLP shared variable pair $(X, X?)$ whose writer $X$ is held by $p$ and reader $X?$ by $q$; it consists of:
\begin{itemize}
\item At agent $p$: a local variable pair $(X_p, X_p?)$ where both occur in $p$'s resolvent, and a spawned goal $\texttt{global\_send}(X_p?, \_r(p,i), q)$
\item At agent $q$: a local variable pair $(X_q, X_q?)$ where both occur in $q$'s resolvent, and an entry $(X_q, p, i)$ at index $i$ in $q$'s global writers table
\end{itemize}
The maGLP writer $X$ corresponds to $X_p$ and the reader $X?$ to $X_q?$; the inner ends $X_p?$ and $X_q$ have no maGLP counterpart. The mirror case, with writer held by $q$ and reader by $p$, has a local pair $(Y_p, Y_p?)$ and entry $(Y_p, q)$ at $p$, and a local pair $(Y_q, Y_q?)$, the record $(Y_q, p, i, p)$ in $U_q$, and spawned goal $\texttt{global\_send}(Y_q?, \_w(p,i), p)$ at $q$.
\end{definition}

A sent value addressed by a reader name is retained by its sender as \emph{pending} until acknowledged: the receiver, upon assigning the entry's writer, sends an acknowledgement, which releases the pending value and closes the link, and records the acknowledging holder as the holder of every link whose reader name the value carried; a request for a pending value redirects it to the requester. Writer-name and serializer values are not acknowledged---their entries are fixed. Traffic matching no entry or pending value is stale (indices are never reused) and is dropped.

\ifappendix
\mypara{Held Links}
A link end whose other end is at an agent this runtime did not exchange the name with is \emph{held}: nothing is sent on it, no assignment arriving on it is applied, and the runtime reports the link on the agent's network input stream. Four cases arise. At the holder of a forwarded reader, the request it creates is held, and at the holder of a forwarded writer, the value its \texttt{global\_send} goal produces is held. At the anchor, an arriving request is held, and so is an assignment arriving from an agent the anchor did not export the name to. The system predicate \texttt{authorise\_link/2} answers: authorise, and the runtime releases what it held; refuse, and it drops the link. The code obtains the channel first, by a \emph{warm call} to the anchor---a third form of introduction beside the cold call and the friend-mediated introduction---specified in the Secure GLP paper~\cite{keidar2026secure}. A forwarded link is held at both ends, so two people are asked: the holder before anything is sent on the link, and the anchor when the forwarded traffic reaches it. \Cref{app:held-links} gives the definitions.
\fi

Global links parallel the encodings of synchronous into asynchronous communication in the $\pi$-calculus~\cite{milner1992calculus,honda1991object,boudol1992asynchrony}\ifappendix; Remark~\ref{rem:global-link-design} records the design alternative and the cost\fi.
\ifappendix
\begin{remark}[Global-Link Design and Cost]
\label{rem:global-link-design}
The alternative we developed first and eliminated is variable migration---moving one end of a shared pair to the other end's agent---whose defects are recorded in \Cref{app:dev-history}. A global link costs one \texttt{global\_send} goal at one end and one global writers table entry at the other; the goal and entry are released when the value is acknowledged. Message traffic per reader-name assignment is the value and its acknowledgement, plus one request after a re-export; a writer-name assignment is one message, however far the writer has been forwarded.
\end{remark}
\fi

\mypara{madGLP Local States and Transitions}
Each madGLP agent executes its local resolvent with FIFO scheduling, tracking three kinds of goals: \emph{active} (ready to reduce), \emph{suspended} (waiting on readers), and \emph{failed}. A reduction succeeds, suspends on a set $W$ of readers, or fails (Definition~\ref{def:reduction-suspend-fail}); when a reader in $W$ is later instantiated, the suspended goal is reactivated (Definition~\ref{def:reactivation-set}).

madGLP defines three unary transactions: \emph{Reduce} performs a local reduction step at one agent; \emph{Send} places a queued outgoing message on a communication channel; and \emph{Receive} applies an incoming message at its destination---a \emph{value} is assigned to the local writer recorded for it in the receiving agent's global writers table, reactivating any goals suspended on the corresponding reader, and is acknowledged; a \emph{request} directs a pending value to its sender; an \emph{acknowledgement} releases the pending value. Cold-calls---messages between agents with no prior shared variables---are realised uniformly through a reserved index-0 serializer that lets multiple senders append to an agent's network input stream, so that all three transaction kinds remain unary. We assume \emph{fair message delivery}: every message placed on a communication channel is eventually delivered to its recipient. Two madGLP transitions are $\sim$-equivalent when they perform the same Reduce (same agent, goal by identity, clause), the same Send (same agent, queued message), or the same Receive (same agent, incoming message). The formal Reduce, Send, and Receive transactions and this equivalence are defined in \appref{def:madglp-ts}{\arxivref}.

\mypara{Example: Client-Monitor Communication}
We illustrate madGLP's variable linking with a client-monitor scenario.

\Program{Concurrent Monitor}\label{program:monitor}
\begin{verbatim}
monitor(Reqs) :- monitor(Reqs?,0).
monitor([add|Reqs],Sum) :-
    Sum1 := Sum? + 1, monitor(Reqs?,Sum1?).
monitor([subtract|Reqs],Sum) :-
    Sum1 := Sum? - 1, monitor(Reqs?,Sum1?).
monitor([value(Sum?)|Reqs],Sum) :-
    integer(Sum?) |  monitor(Reqs?,Sum?).
monitor([],_).
\end{verbatim}

The notation \texttt{integer(Sum?) | Body} is a \emph{guard}: the body is executed only when the guard succeeds, i.e., when \texttt{Sum?} is instantiated to an integer.  The notation \texttt{goal@p} designates that \texttt{goal} runs at agent~$p$.

Consider the initial goal \texttt{client1(Xs)@p, monitor(Xs?)@q}, where \texttt{client1} stands for an arbitrary client process, external to the program, emitting \texttt{add}, \texttt{subtract}, and \texttt{value(\_)} messages on the stream \texttt{Xs}. The initial goal establishes a shared pair with writer $\mathsf{Xs}$ at $p$ and reader $\mathsf{Xs}?$ at $q$. The correspondence between maGLP shared variables and madGLP local pairs connected by global links proceeds through three stages (Figure~\ref{fig:madglp-example}):
\begin{itemize}
\item \textbf{Stage 0:} At boot time, the coordinator globalizes the initial goal, creating local pairs at each agent connected by global links. Agent $p$ gets $\texttt{client1}(\mathsf{Xs_p})$ and spawns $\texttt{global\_send}(\mathsf{Xs_p}?, \_r(p,1), q)$; agent $q$ gets $\texttt{monitor}(\mathsf{Xs_q}?)$ and creates entry $(\mathsf{Xs_q}, p, 1)$.
\item \textbf{Stage 1:} After $p$ assigns $\mathsf{Xs_p} := [\mathtt{add}|\mathsf{Xs1_p?}]$, the \texttt{global\_send} goal at $p$ fires, sending the value to $q$. A new link is created for $\mathsf{Xs1}$ with its own \texttt{global\_send} goal at $p$ and entry at $q$.
\item \textbf{Stage 2:} After $p$ assigns $\mathsf{Xs1_p} := [\mathtt{value}(\mathsf{V_p})|\mathsf{Xs2_p?}]$. Note that $p$ sends writer $\mathsf{V}$, so $p$ creates entry $(\mathsf{V_p}, q)$ and $q$ spawns $\texttt{global\_send}(\mathsf{V_q?}, \_w(p,k), p)$. When monitor assigns $\mathsf{V_q}$, the value flows $q{\to}p$.
\end{itemize}

Each delivered reader-name value is acknowledged, closing its link.\ifappendix{} Had $q$ passed $\mathsf{Xs?}$ on to a third agent before a value arrived, the new holder's request would have redirected the pending value to it, and the stream would continue to the new holder directly.\fi

\ifappendix
\section{madGLP Specification}
\label{app:madglp-spec}

This section contains the complete madGLP specification: full definitions, remarks, and correctness proofs. It serves as the formal specification from which AI derives the English+code spec and Dart implementation.

\subsection{Global Variable Names}
\label{app:global-variable-names}

\begin{definition}[Global Variable Name]\label{def:global-variable-name}
A \temph{global variable name} is a term of the form $\_w(p, i)$ or $\_r(p, i)$, where $p \in \Pi$ is an agent identifier and $i \in \mathbb{N}$ is an index allocated by $p$ during globalization. The name $\_w(p, i)$ denotes a writer globalized at $p$, and $\_r(p, i)$ denotes a reader globalized at $p$.
\end{definition}

A global name has two presentations. In a message between agents it stands in the place of a variable and is encoded as one (\S\ref{app:cf-terms}); it never stands for a variable in a resolvent. To a program it is a ground term---the 2-ary structures $\_r(P, I)$ and $\_w(P, I)$, the functor giving the polarity and the first argument the anchor---which is the form carried by \texttt{global\_send} goals, by \texttt{send\_to\_net} (Remark~\ref{rem:network-output-processing}), and by the reports and answers of \S\ref{app:held-links}. Global names identify the \emph{anchor} of a global link and enable message routing; an unbound imported reader, and an imported writer, are forwarded under their original name (Definition~\ref{def:globalize}), so a link's anchor is fixed however far its end travels. Over a real network, the agent name $p$ in $\_w(p, i)$ and $\_r(p, i)$ is the public key of $p$ (Remark~\ref{rem:agent-names}); the simulation realisation uses symbolic agent names.

\subsection{Global Writers Table}
\label{app:global-writers-table}

The global writers table tracks local writers that await incoming assignments from remote agents. Each entry maps a global name to the local writer that will be assigned when a message arrives.

\begin{definition}[Global Writers Table Entry]\label{def:global-writers-table-entry}
A \temph{global writers table entry} at agent $p$ is either:
\begin{itemize}
\item $(X, q)$ for entries created by Globalize: $X \in \calV$ is a local writer that will be assigned when a value arrives on the entry's name, and $q \in \Pi$ is the agent the name was exported to. A value from any other agent reaches the entry only after a forwarding, and is held (\S\ref{app:held-links})
\item $(X, q, i)$ for entries created by Localize: $X \in \calV$ is a local writer, $q \in \Pi$ is the remote agent, and $i \in \mathbb{N}$ is the index in $q$'s global name (needed to match incoming messages)
\end{itemize}
\end{definition}

\begin{definition}[Global Writers Table]\label{def:global-writers-table}
The \temph{global writers table} $W_p$ of agent $p$ is an array of global writers table entries. For entries created by Globalize at index $i$, the index $i$ is the index in the global name $\_w(p, i)$. For entries created by Localize, the entry stores the remote index explicitly.
\end{definition}

\begin{definition}[Imported-Writer Records]\label{def:imported-writer-records}
The \temph{imported-writer records} $U_p$ of agent $p$ is a set of quadruples $(Y, o, k, s)$, where $Y \in \calV$ is a local writer created by Localize from the global name $\_w(o, k)$, $o \in \Pi$ is the link's anchor, $k \in \mathbb{N}$ its index at $o$, and $s \in \Pi$ the agent the name came from. The record marks $Y$ as the sending end of that link. It is removed when the value is sent, and when $Y$ is re-exported and the name forwarded (Definition~\ref{def:globalize}).
\end{definition}

\ifappendix
\begin{remark}[What the Global Writers Table Stores]
\label{rem:global-writers-table-stores}
The table contains only writers that await incoming assignments. When agent $p$ globalizes a writer $Y$, it creates an entry for $Y$ because $p$ will receive the assignment: the agent the name was exported to holds the writer and sends the value back, or, after a forwarding, whichever agent then holds it. When agent $q$ localizes a reader global name $\_r(p, i)$, it creates an entry for its local writer because $q$ will receive the assignment. No entries are created for outgoing links---those are handled by \texttt{global\_send} goals; an outgoing link whose writer was imported also carries a record (Definition~\ref{def:imported-writer-records}).
\end{remark}
\fi

\subsection{The global\_send Predicate}
\label{app:global-send}

Outgoing communication is handled by spawned goals rather than by table entries. This uniform approach correctly handles all cases, including forwarding when both ends of a variable pair are exported. The system predicate \texttt{global\_send/3} takes a reader $T$ whose value will be sent once known, a global variable name $G$ (either $\_w(p,i)$ or $\_r(p,i)$) identifying the link, and a destination agent $Q$.

\begin{definition}[global\_send Predicate]\label{def:global-send}
The system predicate \texttt{global\_send/3} is defined as:
\begin{verbatim}
global_send(T, G, Q) :- known(T?) | '_send'(T?, G?, Q?).
\end{verbatim}
\end{definition}

The guard \texttt{known(T?)} succeeds when $T$ is assigned a non-variable term; the builtin \texttt{'\_send'(T?, G?, Q?)} globalizes $T$ and adds message $(G := T^\uparrow, Q)$ to the agent's outgoing message set. A sent reader-name value remains with its sender, \emph{pending}, until acknowledged (\S\ref{app:requests-acks}), so a later request can redirect it to the reader's current holder. The full guard and body-kernel tables, including the row for \texttt{'\_send'}, are catalogued in the companion GLP paper~\cite{shapiro2025glp}.

\ifappendix
\begin{remark}[Forwarding via global\_send]
\label{rem:forwarding-global-send}
When an agent exports both ends of a variable pair (e.g., sending $[X, X?]$ to another agent, or $X$ to one agent and $X?$ to another), the Globalize operation spawns a \texttt{global\_send} goal for the exported reader. If a value arrives on one global link and is assigned to a local writer $X$, then $X?$ becomes known, triggering any \texttt{global\_send} goal watching $X?$. This automatically forwards the value without requiring special forwarding logic in the Receive transaction.
\end{remark}
\fi

\subsection{Networking Seam Predicates}
\label{app:seam-predicates}

Four system predicates give GLP code the networking layer's services. Their effects lie outside the madGLP transition system, as those of \texttt{sign} and \texttt{self\_module} do: the transactions of \S\ref{app:madglp-local-states-transitions} do not mention them, and $\pi$ does not see them.

\begin{definition}[Seam Predicates]\label{def:seam-predicates}
The four seam predicates are defined as:
\begin{verbatim}
peer_address(P, A) :- ground(P?) | '_peer_address'(P?, A).
punch_udp(A)       :- ground(A?) | '_punch_udp'(A?).
place_declare(P, R, E) :-
    ground(P?), ground(R?) | '_place_declare'(P?, R?, E).
place_remove(P)    :- ground(P?) | '_place_remove'(P?).
\end{verbatim}
\end{definition}

\texttt{peer\_address} assigns $A$ the address at which the layer observes peer $P$; \texttt{punch\_udp} opens a path to address $A$ and returns nothing; \texttt{place\_declare} declares a place $P$ --- the region of radius $R$ about the declaring agent's location at the time of the call --- and assigns $E$ a stream of that agent's own \texttt{entered} and \texttt{exited} events for $P$, carrying also \texttt{unobservable} where the layer stops reporting crossings and \texttt{observable} where it resumes, between the two seeing no crossing and delivering none late; \texttt{place\_remove} ends the declaration of $P$, and does nothing where $P$ is not declared.

A place event is bridged below the madGLP layer, as the person's input on the UI channel is (\S\ref{subsec:madglp}): the layer reports the event, and the runtime assigns the tail writer of the named place's stream, reactivating whatever suspended on its reader. The stream is fed serializer-fashion, like the network input stream (Definition~\ref{def:index-0-serializer}): the runtime retains the new tail rather than closing the stream, so one declaration yields one stream however many entries and exits follow. The stream is closed, its tail assigned \texttt{[]}, in exactly two cases: \texttt{place\_remove}, and a further declaration of the same place, which supersedes the earlier one. An event for a place removed or superseded is dropped. A declaration the layer refuses is not one of the two: $E$ receives \texttt{unobservable}, the declaration stands, and \texttt{observable} follows if the layer later begins reporting. A refusal is the layer watching nothing, which is what \texttt{unobservable} already says; closing the stream instead would make a refusal permanent, which no refusal is. A refusal answers one declaration, so it is dropped where that declaration has been removed or superseded, as an event for such a place is. Neither an event nor the assignment it causes is a transaction, and neither has a maGLP counterpart; the correctness results concern the transition system, which external input does not enter.

The layer functions behind these predicates are the networking layer's, specified with the seam contract (\S\ref{app:in-networking}).

\subsection{Requests and Acknowledgements}
\label{app:requests-acks}

Three message kinds cross agent boundaries: a \emph{value} $G := T^\uparrow$, a \emph{request} $\texttt{req}(G)$, and an \emph{acknowledgement} $\texttt{ack}(G)$, the latter two carrying reader names only.
\begin{itemize}
\item A request is sent by an agent that localizes a forwarded reader name (Definition~\ref{def:localize}): the anchor did not itself export the name to that agent, and the request is how the anchor learns the reader's current holder.
\item On receiving $\texttt{req}(\_r(p, i))$, the anchor $p$ records the requester as the link's holder; if the value is pending, $p$ redirects it to the requester.
\item On receiving a value $\_r(p, i) := T^\uparrow$, the holder assigns the entry's writer, removes the entry, and acknowledges to the sender. Writer-name values are not acknowledged---a writer-name value always finds its fixed anchor entry---and neither are serializer values, the index-0 entry being permanent. A forwarded writer needs no request either: $\_w(o,k)$ names its destination, so the value reaches the anchor's entry from whichever agent holds the name.
\item On receiving $\texttt{ack}(\_r(p, i))$ from $s$, the anchor records $s$ as the holder of every link whose own reader name occurs in the acknowledged value---$s$ localized them---and re-addresses their pending values to $s$; it then removes the acknowledged value: the link is closed.
\item A value, request, or acknowledgement matching no entry and no pending value is dropped. Indices are never reused, so such traffic is stale: a value reaching a holder that forwarded the reader on, a request or acknowledgement for a closed link.
\end{itemize}
Single occurrence gives a reader one holder at any time, so a link's holders form a chain and exactly one holder acknowledges its value.

\subsection{Index-0 Serializer for Cold-Calls}
\label{app:index-0-serializer}

Cold-calls---messages between agents with no prior shared variables---are implemented using a reserved global name at index 0. Each agent $p$ has a well-known global reader $\_r(p, 0)$ that serves as a \emph{serializer} for incoming cold-call messages.

\begin{definition}[Index-0 Serializer]\label{def:index-0-serializer}
Index 0 in each agent's global writers table is reserved for the \emph{network input serializer}:
\begin{itemize}
\item At boot time, agent $p$ creates a permanent entry $(N_p, *)$ at index 0, where $N_p$ is the local writer for $p$'s network input stream
\item The global name $\_r(p, 0)$ is the serializer reader; any agent can send to $\_w(p, 0)$
\item Index allocation for regular global links starts at 1
\end{itemize}
\end{definition}

\ifappendix
\begin{remark}[Cold-Call via Serializer]
\label{rem:cold-call-serializer}
To send a cold-call message $T$ to agent $q$, agent $p$ spawns the goal $\texttt{global\_send}(T, \_w(q,0), q)$. When $T$ becomes known, the \texttt{'\_send'} builtin detects that index is 0 and uses serializer semantics:
\begin{enumerate}
\item The message sent is $\_w(q,0) := [T^\uparrow \mid \_w(q,0)]$, wrapping $T$ in a list cell with the serializer writer as tail
\item When $q$ receives this message, it assigns $N_q := [T_q^\downarrow \mid N'_q]$ where $N'_q$ is a fresh local writer
\item The entry at index 0 is updated to $(N'_q, *)$, not removed
\end{enumerate}
This allows multiple agents to send cold-calls to $q$ concurrently; messages are serialized onto $q$'s network input stream in an arbitrary order determined by message arrival.
\end{remark}

\begin{remark}[Network Output Processing]
\label{rem:network-output-processing}
Each agent spawns a \texttt{send\_to\_net} goal to process its network output stream:
\begin{verbatim}
send_to_net([msg(Q, T) | In]) :- ground(Q?) |
    global_send(msg(Q?, T?), '_w'(Q?, 0), Q?), send_to_net(In?).
send_to_net([]).
\end{verbatim}
When the agent writes \texttt{msg(q, T)} to its network output stream, the \texttt{send\_to\_net} goal observes this and spawns $\texttt{global\_send}(\texttt{msg}(q,T), \_w(q,0), q)$, forwarding the whole message to $q$'s index-0 serializer; the groundness-implying guard \texttt{ground(Q?)} licenses the multiple occurrences of \texttt{Q}~\cite{shapiro2025glp}. The serializer thus appends \texttt{msg}$(q,T)$, localized, to $q$'s network input stream, matching the maGLP Cold-call (Definition~\ref{definition:maGLP}).
\end{remark}

\begin{remark}[Serializer vs Regular Global Links]
\label{rem:serializer-vs-regular}
For index $i > 0$, the \texttt{'\_send'} builtin sends $G := T^\uparrow$ directly and the receiver removes the entry after assignment and acknowledges. For index 0 (serializer), the builtin wraps the term in a list cell $[T^\uparrow \mid G]$ reusing the serializer writer, and the entry is updated rather than removed. This distinction is the only special case for cold-calls; all other machinery (globalization, localization, message routing) is shared with regular global links.
\end{remark}

\begin{remark}[Potential Generalisation]
\label{rem:potential-generalization}
The many-to-one serializer mechanism at index 0 implements an efficient multiway merge for network input. This mechanism could be generalised as a language extension allowing programs to create additional serializers, but we leave such extensions to future work.
\end{remark}
\fi

\subsection{Held Links}
\label{app:held-links}

An agent \emph{exchanges} a global name with the agent it sends that name to, or receives it from. A link end at agent $p$ is \emph{held} when the agent at its other end is not the one $p$ exchanged the name with. While a link is held, $p$ places nothing on the communication channel for it and applies no assignment arriving on it. A held outgoing message remains in $M_p$, marked \emph{held}, and a held incoming assignment is retained keyed by its global name.

A held link is \emph{reported} to the program, and the runtime acts on it only once the program answers. The report is $\texttt{pending\_link}(G, S)$, with $G$ the link's global name in its ground presentation (\S\ref{app:global-variable-names}) and $S$ the agent the name or the traffic came from; it is appended to $p$'s network input stream through the index-0 serializer entry (Definition~\ref{def:index-0-serializer}), within the transaction that would otherwise have acted. $G$ names the anchor and $S$ the agent $p$ exchanged the end with, so the program has both and the runtime retains nothing further. The report is appended last in that transaction, after anything the transaction itself appends, so a link is reported only once the term carrying it can be read. The report records the module of the goal it reaches, and the record is released with the held link.

Four cases arise:
\begin{enumerate}
\item At the holder of a forwarded reader, the request of Definition~\ref{def:localize} is held. A request is created exactly when the name was forwarded, so every request a holder creates is held.
\item At the holder of a forwarded writer, the value its \texttt{global\_send} goal produces is held.
\item At the anchor, an arriving request is held. The agent the anchor exported the name to sends no request, so every request the anchor receives is held.
\item At the anchor, an assignment arriving from an agent it did not export the name to is held.
\end{enumerate}

The program answers a report with the system predicate \texttt{authorise\_link/2}, which takes the global name $L$ of a reported link and an answer $A$, either \texttt{authorise} or \texttt{refuse}.

\begin{definition}[authorise\_link Predicate]\label{def:authorise-link}
The system predicate \texttt{authorise\_link/2} is defined as:
\begin{verbatim}
authorise_link(L, A) :- known(A?) | '_authorise_link'(L?, A?).
\end{verbatim}
\end{definition}

On \texttt{authorise}, the builtin \texttt{'\_authorise\_link'} releases what was held: an outgoing message becomes unsent, and a retained assignment is applied. On \texttt{refuse} it drops the link: the held message is removed from $M_p$, a retained assignment is discarded, and the link's entry, record, and \texttt{global\_send} goal at $p$ are released. A refusal is accepted in every case, the program refusing for reasons the runtime does not see.

A global name is a ground term (\S\ref{app:global-variable-names}), so a program can write one it was never given, and a term arriving from another agent can carry one: the name is not a secret and cannot be made one. A kernel taking a global name from the program therefore acts only on a name the runtime reported to that goal's module, and does nothing otherwise. \texttt{authorise\_link/2} is the only such kernel at present; any later one inherits the rule. The check matters once an agent runs a module it adopted from another (\S\ref{app:program-shipment}): adopted code shares the runtime, and without it could answer for links reported to its host.

A forwarded link is held at both of its ends, so two people are asked and each decides independently. The holder is asked before anything is sent on the link. The anchor is asked only when the forwarded traffic reaches it, after the holder has sent it. When a link may be authorised, and the \emph{warm call} by which the code obtains the channel---a third form of introduction beside the cold call and the friend-mediated introduction---are specified in the Secure GLP paper~\cite{keidar2026secure}; madGLP fixes the mechanism only.

\begin{remark}[Authorised Runs]\label{rem:authorised-runs}
A refusal drops a link, so the value it would have carried never arrives and the maGLP Communicate it implements is never taken. The correctness results below concern \temph{authorised runs}---runs in which every reported link is authorised---as they concern runs with fair message delivery.
\end{remark}

\subsection{madGLP Local States and Transitions}
\label{app:madglp-local-states-transitions}

Each madGLP agent executes its local resolvent with FIFO scheduling, tracking three kinds of goals: \emph{active} (ready to reduce), \emph{suspended} (waiting on readers), and \emph{failed}. A reduction succeeds, suspends on a set $W$ of readers, or fails (Definition~\ref{def:reduction-suspend-fail}); when a reader in $W$ is later instantiated, the suspended goal is reactivated (Definition~\ref{def:reactivation-set}).

\begin{definition}[madGLP Local State]\label{def:madglp-local-state}
The local state of agent $p \in \Pi$ is a tuple $s_p = (R_p, W_p, U_p, M_p)$ where:
\begin{enumerate}
\item $R_p = (A_p, S_p, F_p)$ is a \emph{deterministic resolvent}: $A_p$ is a queue of active goals, $S_p$ contains suspended goals paired with blocking readers, and $F_p$ contains failed goals
\item $W_p$ is a \emph{global writers table}
\item $U_p$ is a set of \emph{imported-writer records} (Definition~\ref{def:imported-writer-records})
\item $M_p$ is a set of \emph{outgoing messages}; a sent reader-name value remains in $M_p$, marked \emph{pending}, until acknowledged, and a held message remains marked \emph{held} until authorised (\S\ref{app:held-links})
\end{enumerate}
\end{definition}

\begin{definition}[madGLP Transition System]\label{def:madglp-ts}
The \emph{madGLP transition system} over agents $P \subset \Pi$ and a GLP program is the multiagent transition system $\textrm{madGLP} = (C, c_0, T, {\sim})$ where $C$ is the set of madGLP configurations, $c_0$ is the initial configuration with empty resolvents except for the initial \texttt{agent} goal and index-0 serializer entry, $T$ is the set of Reduce, Send, and Receive transactions defined below, and two transitions are $\sim$-equivalent when they perform the same Reduce (same agent, goal by identity, clause), the same Send (same agent, queued message), or the same Receive (same agent, incoming message).
\end{definition}

Entries in the global writers table $W_p$ route incoming assignments to local writers: $(X, q)$ at index $i$ binds $X$ to assignments arriving via $\_w(p,i)$ from $q$, and $(X_q, p, i)$ binds $X_q$ to assignments arriving via $\_r(p,i)$ from $p$.  Index~0 holds the serializer entry $(N_p, *)$, which is updated (not removed) on each cold-call so that multiple senders can append to the network input stream.

\begin{definition}[madGLP Reduce Transaction]\label{def:madglp-reduce}
The \emph{unary} Reduce transaction for agent $p$ where $A_p = A \cdot A_r$ transitions $s_p \rightarrow s'_p$ with exactly one of the following holding:
\begin{enumerate}
\item \textbf{Reduce:} The GLP reduction of $A$ with the first applicable clause $C$ succeeds with $(B, \hat\sigma)$. Let $R = \mathrm{reactivate}(S_p, \hat\sigma?)$. Then:
\[
A'_p = (A_r \cdot B \cdot R)\hat\sigma\hat\sigma?, \quad S'_p = S_p \setminus \{(G, W) : G \in R\}, \quad F'_p = F_p
\]
\item \textbf{Suspend:} No clause succeeds, but the GLP reduction of $A$ with at least one clause suspends. Let $W = \bigcup_{C} W_C$ where $W_C$ is the suspension set for clause $C$. Then:
\[
A'_p = A_r, \quad S'_p = S_p \cup \{(A, W)\}, \quad F'_p = F_p
\]
\item \textbf{Fail:} No clause succeeds and no clause suspends. Then:
\[
A'_p = A_r, \quad S'_p = S_p, \quad F'_p = F_p \cup \{A\}
\]
\end{enumerate}
\end{definition}

\begin{definition}[madGLP Send Transaction]\label{def:madglp-send}
The \emph{unary} Send transaction for agent $p$ is enabled when $(m, q) \in M_p$ is unsent and not held (\S\ref{app:held-links}).  It places $m$ in the communication channel to $q$; a reader-name value remains in $M_p$, marked \emph{pending}, until acknowledged (\S\ref{app:requests-acks}), and every other message is removed from $M_p$ when sent.
\end{definition}

\begin{definition}[madGLP Receive Transaction]\label{def:madglp-receive}
The \emph{unary} Receive transaction processes a message $m$ from the communication channel, with exactly one of the following holding:
\begin{enumerate}
\item \textbf{Serializer ($m = (\_w(q, 0) := [T^\uparrow \mid \_w(q,0)])$):} The receiving agent is $q$, with entry $(N_q, *)$ at index~0 in $W_q$.  Let $N'_q$ be a fresh local writer, $\hat\sigma = \{N_q := [T_q^\downarrow \mid N'_q?]\}$, and $R = \mathrm{reactivate}(S_q, \hat\sigma?)$.  Then:
\[
A'_q = (A_q \cdot R)\hat\sigma\hat\sigma?, \quad S'_q = S_q \setminus \{(G, W) : G \in R\}, \quad F'_q = F_q
\]
The entry at index~0 in $W'_q$ becomes $(N'_q, *)$.
\item \textbf{Writer ($m = (\_w(p, i) := T^\uparrow)$, $i > 0$):} The receiving agent is $p$, with entry $(X, q)$ at index $i$ in $W_p$; the sender need not be $q$, the name having been forwarded (Definition~\ref{def:globalize}); an assignment from an agent other than $q$ is held and applied on authorisation (\S\ref{app:held-links}).  Let $\hat\sigma = \{X := T_p^\downarrow\}$ and $R = \mathrm{reactivate}(S_p, \hat\sigma?)$.  Then:
\[
A'_p = (A_p \cdot R)\hat\sigma\hat\sigma?, \quad S'_p = S_p \setminus \{(G, W) : G \in R\}, \quad F'_p = F_p
\]
The entry at index $i$ is removed from $W'_p$.
\item \textbf{Reader ($m = (\_r(p, i) := T^\uparrow)$):} The receiving agent $q$ is the one that localized $\_r(p,i)$, with entry $(X_q, p, i)$ in $W_q$.  Let $\hat\sigma = \{X_q := T_q^\downarrow\}$ and $R = \mathrm{reactivate}(S_q, \hat\sigma?)$.  Then:
\[
A'_q = (A_q \cdot R)\hat\sigma\hat\sigma?, \quad S'_q = S_q \setminus \{(G, W) : G \in R\}, \quad F'_q = F_q
\]
The entry matching $(p, i)$ is removed from $W'_q$, and $(\texttt{ack}(\_r(p,i)), s)$ is added to $M'_q$, where $s$ is the agent that sent $m$.
\item \textbf{Request ($m = \texttt{req}(\_r(p, i))$):} The receiving agent is $p$, the anchor. The request is held (\S\ref{app:held-links}); on authorisation $p$ records the requester as the holder of the link $\_r(p,i)$, and if a pending value addressed by $\_r(p,i)$ is in $M_p$, it is re-addressed to the requester and becomes unsent. A request for a closed link is dropped.
\item \textbf{Acknowledgement ($m = \texttt{ack}(\_r(p, i))$):} The receiving agent is the sender of the acknowledged value; let $s$ be the acknowledging agent and $V$ the pending value addressed by $\_r(p,i)$. The receiver records $s$ as the holder of every link whose reader name anchored at the receiver occurs in $V$, re-addressing any pending value of such a link to $s$; it then removes $V$ from its $M$. An acknowledgement for a closed link is dropped.
\end{enumerate}
\end{definition}

\begin{remark}[Fair Message Delivery]
\label{rem:fair-message-delivery}
All madGLP transactions are unary. Cold-calls, which in maGLP require binary transactions, are implemented using the index-0 serializer, decomposing into unary Send and Receive transactions. We assume \temph{fair message delivery}---a standard assumption for distributed systems with reliable channels: every message placed in a communication channel is eventually delivered to its recipient.
\end{remark}

\ifappendix
\begin{remark}[Globalize-Localize Correspondence]
\label{rem:globalize-localize-correspondence}
The pairing between Globalize and Localize ensures correct dataflow:
\begin{itemize}
\item \textbf{Writer globalized at $p$:} Globalize creates entry $(Y, q)$ at $p$. Localize records $(Y_q, p, i, p)$, spawns $\texttt{global\_send}(Y_q?, \_w(p,i), p)$ at $q$, and puts $Y_q$ (writer) in $q$'s term. When $q$ assigns $Y_q$, the spawned goal fires and sends the value to $p$, where the entry routes it to $Y$, making it available via $Y?$.
\item \textbf{Reader globalized at $p$:} Globalize spawns $\texttt{global\_send}(Y?, \_r(p,i), q)$ at $p$. Localize adds entry $(Z_q, p, i)$ at $q$ and puts $Z_q?$ (reader) in $q$'s term. When $p$ assigns $Y$, the spawned goal fires and sends the value to $q$, where the entry routes it to $Z_q$, making it available via $Z_q?$.
\end{itemize}
\end{remark}
\fi

\ifappendix
\subsection{madGLP Remarks}
\label{app:madglp-remarks}

\begin{remark}[Outgoing Messages via global\_send]
\label{rem:outgoing-messages}
The Reduce transaction does not directly generate outgoing messages. Instead, when a writer $X$ is assigned, $X?$ becomes known, which may trigger a \texttt{global\_send} goal watching $X?$. That goal's reduction (via the \texttt{'\_send'} builtin) adds the message to $M_p$. This uniform approach handles both direct sends and forwarding.
\end{remark}

\begin{remark}[Automatic Forwarding]
\label{rem:automatic-forwarding}
The Receive transaction assigns the local writer, removes the entry, and acknowledges. If the assigned writer's reader ($X?$) is being watched by a \texttt{global\_send} goal (because it was also exported), that goal will fire on a subsequent Reduce, automatically forwarding the value. No special forwarding logic is needed in Receive.
\end{remark}

\begin{remark}[Transaction Degrees]
\label{rem:madglp-transaction-degrees}
All madGLP transactions are unary: Reduce operates on a single agent's resolvent, Send removes a message from one agent's outbox, and Receive processes an incoming message at one agent. Cold-calls, which in maGLP require binary transactions, are implemented in madGLP using the index-0 serializer mechanism, which decomposes into unary Send and Receive transactions.
\end{remark}

\begin{remark}[Determinism]
\label{rem:madglp-determinism}
madGLP is deterministic at the agent level: given a local state, exactly one Reduce transition is enabled (FIFO goal selection, first-clause matching). However, madGLP is nondeterministic at the system level due to communication asynchrony: the order in which messages arrive at an agent depends on network timing, and different arrival orders may lead to different computations.
\end{remark}

\begin{remark}[Correspondence to maGLP]
\label{rem:correspondence-maglp}
The maGLP transactions are implemented in madGLP as follows:
\begin{itemize}
\item \textbf{Communicate:} The maGLP binary Communicate transaction, which atomically transfers an assignment from one agent's writer to another agent's reader, is implemented by the sequence: Reduce (assigns writer, triggering \texttt{global\_send}) $\rightarrow$ Send $\rightarrow$ Receive (applies assignment).
\item \textbf{Cold-call:} The maGLP binary Cold-call transaction is implemented using the index-0 serializer: the sender spawns $\texttt{global\_send}(T, \_w(q,0), q)$, which fires when $T$ is known, followed by Send $\rightarrow$ Receive (serializer appends to $q$'s network input stream).
\end{itemize}
The correctness of this implementation is proved below.
\end{remark}

\begin{remark}[Global Writers Table Lifecycle]
\label{rem:global-writers-lifecycle}
An entry is added to the global writers table when a global link is established with the agent as the receiver: either when globalizing a writer (expecting an incoming assignment from the remote agent that received the writer) or when localizing a reader global name (expecting an incoming assignment from the globalizer). The entry is removed when the assignment arrives and the writer is assigned. The table tracks only those writers awaiting incoming assignments from remote agents. An imported-writer record has the mirror lifecycle: added when a writer global name is localized, removed when the value is sent, or when the writer is re-exported and its name forwarded.
\end{remark}

\begin{remark}[Message Routing]
\label{rem:message-routing}
Messages use global names to enable routing:
\begin{itemize}
\item Message $\_w(p, i) := T$ is sent to $p$ (the original globalizer who exported a writer and created the entry); $p$ has an entry at index $i$ in its global writers table
\item Message $\_r(p, i) := T$ is sent to whoever localized $\_r(p, i)$; that agent searches for an entry with remote agent $p$ and remote index $i$
\end{itemize}
\end{remark}
\fi

\subsection{Canonical Encoding}
\label{app:canonical-encoding}

The Send transaction places assignment messages in communication channels (Definition~\ref{def:madglp-send}); over a network, a message travels as a byte string. Signatures and hashes of terms and artefacts are computed over these bytes independently at different agents and must agree byte for byte; the encoding is therefore part of the specification. The protocols that sign and hash shipped content are specified in the Secure GLP paper~\cite{keidar2026secure}.

\begin{definition}[Canonical Encoding]\label{def:canonical-encoding}
A \temph{canonical encoding} is an injective function $e$ from globalized terms and assignment messages to byte strings.
\end{definition}

Injectivity makes every message decodable, as localization requires. The canonical encoding has the following properties:
\begin{enumerate}
\item \emph{Hardware-independent:} integer widths and byte order are pinned; a term yields the same byte string on every machine.
\item \emph{Address-free:} in a globalized term every variable has been replaced by a global name (Definition~\ref{def:globalize}), so a variable appears in the encoding as its global name---agent and index---with its reader/writer polarity; no heap address appears.
\item \emph{Globalizer's names:} the encoding is of the globalized term exactly per Definition~\ref{def:globalize}---names are the allocating agent's: a forwarded reader keeps its original anchor's name.
\item \emph{Canonical on ground terms:} globalization is the identity on ground terms, so a ground term $T$ has one encoding $e(T)$, identical at every agent---the byte string over which terms are signed and hashed.
\end{enumerate}
\ifappendix The byte-level layout of $e$---primitive encodings, term and message grammar, the instruction encoding, the artefact, and the loader---is given in Appendix~\ref{app:code-format}.\fi

\subsection{Compiled Programs as Shippable Values}
\label{app:program-shipment}

madGLP ships code the way it ships terms: a compiled program is a value of the primitive type Module. A GLP program is a hierarchy of modules~\cite{shapiro2026types}; the sender links it into a single flat program and compiles it to an artefact, carried by a Module constant. One canonical encoding (Definition~\ref{def:canonical-encoding}) thus serves terms and code. The compiled program ships inside an assignment message like any other value, and the receiving agent loads it at runtime; a goal is posted to a loaded module by \texttt{run/3}, whose type discipline is the module system's~\cite{shapiro2026types}. The correctness theorems of this paper fix the program; runtime loading sits outside them, treated where composition is treated, in the Secure GLP paper~\cite{keidar2026secure}.

The engine follows the sequential abstract machine of Flat Concurrent Prolog~\cite{houri1989sequential}: a program is compiled to a string of instructions that the engine interprets directly. GLP departs only where the languages differ---the instruction set is GLP's (it matches rather than unifies and pairs each reader with a writer), and the byte order is fixed to one canonical order rather than the running machine's, so an artefact's bytes are read the same way on every machine and are what its compiled identity hashes. A recipient whose machine uses a different order inverts the bytes locally before executing and inverts them again before shipping the program onward, leaving the shipped, hashed bytes unchanged (Appendix~\ref{app:code-format}). Which programs are shipped when, the identities and hashes that name them, and trust in shipped code are specified in the Secure GLP paper~\cite{keidar2026secure}.

\subsection{Correctness Proofs}
\label{app:madglp-correctness-proofs}

\begin{lemma}[Globalize-Localize Correspondence]\label{lem:globalize-localize}
Let $T$ be a term with variables at agent $p$, let $T_p^\uparrow$ be its globalization by $p$ for remote agent $q$, and let $T_q^\downarrow$ be the localization by $q$. Then:
\begin{enumerate}
\item For each writer $Y$ in $T$: globalization creates entry $(Y, q)$ at index $i$ in $p$'s global writers table; localization creates fresh pair $(Y_q, Y_q?)$ at $q$ and spawns $\texttt{global\_send}(Y_q?, \_w(p,i), p)$. These form a global link implementing shared pair $(Y, Y?)$.

\item For each reader $Y?$ in $T$: globalization spawns $\texttt{global\_send}(Y?, \_r(p,i), q)$ at $p$; localization creates fresh pair $(Z_q, Z_q?)$ at $q$ with entry $(Z_q, p, i)$. These form a global link implementing shared pair $(Y, Y?)$.

\item The structure of $T$ is preserved: $T_q^\downarrow$ has the same functor/arity structure as $T$, with each variable replaced by its corresponding local variable at $q$.

\item For each unbound imported reader $Y?$ in $T$ with entry $(Y, o, k)$ at $p$: globalization by $p$ emits the original name $\_r(o, k)$ and removes $p$'s entry; localization at $q$ creates a fresh pair with an entry recording $(o, k)$ and sends $\texttt{req}(\_r(o,k))$ to $o$. The link's anchor is unchanged and its holder is now $q$.

\item For each imported writer $Y$ in $T$ with record $(Y, o, k, s)$ at $p$: globalization by $p$ emits the original name $\_w(o, k)$ and removes $p$'s record and the \texttt{global\_send} goal watching $Y?$; localization at $q$ creates a fresh pair with the record $(Y_q, o, k, p)$ and spawns $\texttt{global\_send}(Y_q?, \_w(o,k), o)$. The link's anchor is unchanged and its sending end is now $q$.
\end{enumerate}
\end{lemma}

\begin{proof}
By the definitions of Globalize (Definition~\ref{def:globalize}) and Localize (Definition~\ref{def:localize}):

For (1): When $Y$ is a writer in $T$, globalization replaces $Y$ with $\_w(p, i)$ and creates entry $(Y, q)$ at index $i$ in $p$'s global writers table. Localization at $q$ sees $\_w(p, i)$, creates fresh pair $(Y_q, Y_q?)$, spawns $\texttt{global\_send}(Y_q?, \_w(p,i), p)$, and replaces $\_w(p, i)$ with $Y_q$ (writer) in $T_q^\downarrow$. When $q$ assigns $Y_q$, the value becomes available via $Y_q?$, the \texttt{global\_send} goal fires and sends message $\_w(p,i) := T^\uparrow$ to $p$. Agent $p$'s Receive finds entry $(Y, q)$ at index $i$ and assigns $Y := T^\downarrow$, making the value available via $Y?$. This matches maGLP semantics where the exported writer is assigned by the remote agent.

For (2): When $Y?$ is a reader in $T$, globalization replaces $Y?$ with $\_r(p, i)$ and spawns $\texttt{global\_send}(Y?, \_r(p,i), q)$ at $p$. Localization at $q$ creates fresh pair $(Z_q, Z_q?)$, creates entry $(Z_q, p, i)$, and replaces $\_r(p, i)$ with $Z_q?$ (reader) in $T_q^\downarrow$. When $p$ assigns $Y$, the value becomes available via $Y?$, the \texttt{global\_send} goal fires and sends $\_r(p,i) := T^\uparrow$ to $q$. Agent $q$'s Receive finds entry $(Z_q, p, i)$ and assigns $Z_q := T^\downarrow$, making the value available via $Z_q?$. This implements the semantics where the exported reader receives the value from the globalizer.

For (3): Both globalization and localization traverse $T$ structurally, replacing only variables while preserving all functors and their arities.

For (4): By the reader-forwarding case of Definition~\ref{def:globalize} and by Definition~\ref{def:localize}, the entry moves with the reader occurrence while the name and its anchor are fixed. The request reaches $o$ by fair delivery; $o$ redirects the pending value, or sends it once produced, to $q$; $q$'s entry consumes it and $q$ acknowledges, closing the link (\S\ref{app:requests-acks}). Single occurrence makes the holder unique at any time, so exactly one acknowledgement closes the link.

For (5): By the writer-forwarding case of Definition~\ref{def:globalize} and case 1 of Definition~\ref{def:localize}, the record and its goal move with the writer occurrence while the name and its anchor are fixed. Single occurrence makes the sending end unique at any time, and $\_w(o,k)$ names its destination, so the value $q$ produces reaches $o$'s entry in one message, without request or acknowledgement.
\end{proof}

\ifappendix
\begin{lemma}[Both Ends Exported]\label{lem:both-ends}
When agent $p$ globalizes term $T$ containing both $X$ and $X?$ for agent $q$, the resulting madGLP configuration correctly implements the maGLP semantics where both ends of the pair are available to $q$.
\end{lemma}

\begin{proof}
Globalization of $T = [\ldots, X, \ldots, X?, \ldots]$ at $p$ for $q$ proceeds as follows:
\begin{itemize}
\item For writer $X$: entry $(X, q)$ at index 1, no spawn, term gets $\_w(p,1)$
\item For reader $X?$: spawns $\texttt{global\_send}(X?, \_r(p,2), q)$, no entry, term gets $\_r(p,2)$
\end{itemize}

Localization at $q$ creates:
\begin{itemize}
\item For $\_w(p,1)$: pair $(Y_q, Y_q?)$, record $(Y_q, p, 1, p)$, spawns $\texttt{global\_send}(Y_q?, \_w(p,1), p)$, term gets $Y_q$ (writer)
\item For $\_r(p,2)$: pair $(Z_q, Z_q?)$, entry $(Z_q, p, 2)$, term gets $Z_q?$ (reader)
\end{itemize}

The term at $q$ is $[\ldots, Y_q, \ldots, Z_q?, \ldots]$ with two independent pairs. Agent $q$ does not know they are connected---that connection exists only through $p$'s local pair $(X, X?)$.

When $q$ assigns $Y_q := T$:
\begin{enumerate}
\item $Y_q?$ becomes known
\item $\texttt{global\_send}(Y_q?, \_w(p,1), p)$ fires, sends $\_w(p,1) := T^\uparrow$ to $p$
\item $p$'s Receive finds entry $(X, q)$ at index 1 and assigns $X := T^\downarrow$
\item $X?$ becomes known at $p$
\item $\texttt{global\_send}(X?, \_r(p,2), q)$ fires, sends $\_r(p,2) := T^\uparrow$ to $q$
\item $q$'s Receive finds entry $(Z_q, p, 2)$ and assigns $Z_q := T^\downarrow$
\item $Z_q?$ becomes known at $q$
\end{enumerate}

The assignment to $Y_q$ flows through $p$'s local pair to $Z_q?$. The same analysis applies when $X$ and $X?$ are sent to different agents.
\end{proof}

\fi

\mypara{The Implementation Mapping}

\begin{definition}[Reconstructed Resolvent]\label{def:reconstructed-resolvent}
Given a madGLP local state $(R_p, W_p, U_p, M_p)$ where $R_p = (A_p, S_p, F_p)$, the \temph{reconstructed resolvent} at agent $p$ is:
$$G_p = A_p \cup \{A : (A, W) \in S_p\} \cup F_p$$
\end{definition}

\begin{definition}[Reconstructed Readers Substitution]\label{def:reconstructed-sigma}
Given a madGLP configuration $c'$, the \temph{reconstructed readers substitution} $\sigma_p$ for agent $p$ is derived from undelivered values: for each value produced at $p$ whose link is still open---queued, pending, or in a communication channel, its addressed entry not yet consumed---include the corresponding reader assignment in $\sigma_p$ according to the variable correspondence.
\end{definition}

\begin{definition}[Implementation Mapping $\pi$]\label{def:madglp-pi-full}
The implementation mapping $\pi$ from madGLP configurations to maGLP configurations is defined as follows. For each agent $p$:
$$\pi_p(R_p, W_p, U_p, M_p) = (G_p, \sigma_p)$$
where $G_p$ is the reconstructed resolvent (Definition~\ref{def:reconstructed-resolvent}) with \texttt{global\_send} goals removed and each local variable replaced by its corresponding maGLP variable according to the variable correspondence (Definition~\ref{def:global-link-detail}), and $\sigma_p$ is the reconstructed readers substitution (Definition~\ref{def:reconstructed-sigma}).

The global mapping $\pi: C' \rightarrow C$ applies this construction to each agent's local state.
\end{definition}

The mapping $\pi$ forgets:
\begin{itemize}
\item The partition of goals into active, suspended, and failed
\item The suspension sets $W$ for suspended goals
\item The global writers tables $W_p$ and the imported-writer records $U_p$ (routing information)
\item The message queues $M_p$ (except as they contribute to $\sigma$)
\item The \texttt{global\_send} goals (implementation mechanism)
\item The pairing of local variables---corresponding variables at different agents map to the same maGLP variable
\end{itemize}

\mypara{Transaction Correspondence}

\begin{proposition}[Reduce Implementation]\label{prop:reduce-implementation}
The madGLP Reduce transaction for agent $p$ projects via $\pi$ to a maGLP Reduce transaction for $p$.
\end{proposition}

\begin{proof}
madGLP Reduce selects the head of the active queue and attempts reduction with the first applicable clause. If reduction succeeds with $(B, \hat\sigma)$, the substitutions $\hat\sigma$ and $\hat\sigma?$ are applied to local goals; $\pi$ maps the result to the maGLP Reduce result. If Suspend or Fail, the goal is reclassified but remains in $G_p = \pi(R_p)$.
\end{proof}

\begin{proposition}[Communicate Implementation]\label{prop:communicate-implementation}
maGLP's binary Communicate transaction from $p$ to $q$ is correctly implemented by madGLP's mechanism: a \texttt{global\_send} goal fires (via Reduce) at the writer-owner, then Send, then Receive at the reader-owner.
\end{proposition}

\begin{proof}
maGLP Communicate$(p,q)$ transfers assignment $\{X? := T\}$ from $\sigma_p$ to $G_q$. In madGLP: (1) the writer-owner's Reduce assigns the local writer; (2) a subsequent Reduce fires the \texttt{global\_send} goal watching the paired reader, adding message $(G := T^\uparrow)$ to the outbox; (3) Send moves the message to the channel; (4) Receive at the reader's current holder finds the entry in its global writers table, localizes, applies the assignment, and acknowledges, closing the link. If the reader was forwarded, the holder's request precedes or redirects the send (\S\ref{app:requests-acks}); a value reaching a past holder is dropped and the pending value is redirected to the current one. By Lemma~\ref{lem:globalize-localize}, $\pi$ maps the result to the maGLP Communicate result.
\end{proof}

\begin{proposition}[Cold-call Implementation]\label{prop:cold-call-implementation}
maGLP's binary Cold-call transaction from $p$ to $q$ is correctly implemented by madGLP's Send and Receive transactions via the index-0 serializer.
\end{proposition}

\begin{proof}
maGLP Cold-call$(p,q)$ transfers \texttt{msg}$(q,X)$ from $p$'s network output stream to $q$'s network input stream, establishing shared variable pairs. In madGLP: $p$ globalizes \texttt{msg}$(q,X)$ (Remark~\ref{rem:network-output-processing}); Send at $p$ places the serializer message in the channel; Receive at $q$ localizes it and appends it to $q$'s network input stream; global links for the variables of $X$ are established. By Lemma~\ref{lem:globalize-localize}, $\pi$ maps the result to the maGLP Cold-call result.
\end{proof}

\begin{lemma}[$\pi$ is Correct]\label{lem:madglp-live}
The implementation mapping $\pi$ is correct.
\end{lemma}

\begin{proof}
Each agent's madGLP execution is deterministic at the local level (FIFO scheduling). For any correct madGLP run $r'$, every madGLP Reduce that becomes enabled is eventually taken. Each such Reduce projects to a maGLP Reduce via $\pi$. A produced value stays pending at its producer until acknowledged; the current holder's entry exists from its localization, and the producer learns the holder from its own export, from the holder's request when the name was forwarded (fair delivery), or from the acknowledgement of the enclosing value when that value was redirected (Definition~\ref{def:madglp-receive}); the value is then sent to, applied at, and acknowledged by that holder---so every maGLP Communicate class that stays enabled is eventually taken. Since madGLP can perform any set of reductions that maGLP can (by Lemma~\ref{lem:madglp-same-reductions}), and, since $r'$ is correct, every enabled equivalence class is eventually taken, $\pi(r')$ is a correct maGLP run.
\end{proof}

\mypara{Outcome-Completeness}

\begin{lemma}[madGLP Performs Equivalent Reductions]
\label{lem:madglp-same-reductions}
For every proper maGLP run $r$ performing reductions $\{R_1, \ldots, R_k\}$, there exists a madGLP run $r'$ performing equivalent reductions.
\end{lemma}

\begin{proof}
By induction on spawn order. Goals in each agent's initial goal are in the initial queue. For any $R_i$ reducing $A_i$ at agent $p$ with $C_i$: if $A_i$ was spawned by an earlier reduction $R_j$, by induction madGLP performs $R_j$, so $A_i$ enters $p$'s queue. Since $R_i$ succeeds in $r$, $A_i$ and the head of $C_i$ have a writer mgu. Since the writer mgu assigns only writers (Definition~\ref{def:writers-assignment}), it is invariant under any reader substitutions applied to $A_i$ in the intervening run. When $A_i$ reaches the queue head, the reduction succeeds.
\end{proof}

\begin{lemma}[madGLP Outcome-Completeness]
\label{lem:madglp-complete}
For every proper maGLP run with outcome $O$, there exists a madGLP run with outcome $O$.
\end{lemma}

\begin{proof}
Let $r$ be a proper maGLP run with outcome $O$, performing reductions $\{R_1, \ldots, R_k\}$.

By Lemma~\ref{lem:madglp-same-reductions}, some madGLP run $r'$ performs equivalent reductions.

By Lemma~\ref{lem:substitution-commutativity} and Lemma~\ref{lem:reduction-set-outcome}, since $r$ and $r'$ perform equivalent reductions (with disjoint substitutions that commute), they have the same outcome $O$.
\end{proof}

\fi
\mypara{Correctness}
Global links implement maGLP shared variable pairs via message passing; Definition~\ref{def:global-link-detail} treats the case of a writer at $p$ and reader at $q$ (Section~\ref{subsec:madglp}), with the mirror case alongside. The implementation mapping $\pi$ from madGLP configurations to maGLP configurations reconstructs the resolvent and readers substitution while forgetting implementation details (goal classification, global writers tables, message queues, \texttt{global\_send} goals).

Correctness rests on \emph{disjoint substitution commutativity}, a consequence of the single-occurrence invariant: reductions at distinct goals assign disjoint variables and therefore commute, so the agent-deterministic, message-passing execution reaches the same outcomes as maGLP's nondeterministic shared-variable semantics.

\begin{theorem}[madGLP Correctly Implements maGLP]
\label{thm:madglp-implements-maglp}
\ifappendix For authorised runs (Remark~\ref{rem:authorised-runs}), the\fi\ implementation $(\textrm{madGLP}, \pi)$ of maGLP is correct and outcome-complete.
\end{theorem}
\begin{proof}
\emph{Correct:} By Lemma~\ref{lem:madglp-live}. \emph{Outcome-complete:} By Lemma~\ref{lem:madglp-complete}.
\end{proof}

\begin{corollary}[Global Links Correctly Implement Shared Variables]
\label{cor:global-links-correct}
Two local pairs joined by a global link together with its \texttt{global\_send} goal correctly implement a maGLP shared variable pair\ifappendix, including the case where both ends of a pair are exported\fi.
\end{corollary}
\begin{proof}
Follows from Theorem~\ref{thm:madglp-implements-maglp} together with Lemma~\ref{lem:globalize-localize}\ifappendix\ and Lemma~\ref{lem:both-ends}\fi.
\end{proof}


\section{AI Development Methodology}
\label{sec:ai-methodology}

This paper offers a precise interface between the human designer and the AI programmer: formal operational semantics serve as a specification language, constraining the space of legal implementations. The emerging discipline~\cite{fowler2025sdd,mundler2025type,blinn2024typed} we advocate and employ is for the human designer and AI to jointly develop and agree upon (\ia)~formal semantics (this paper), (\ib)~an informal (English+code) specification derived by AI from the math, and only then let AI attempt to write (\ic)~Dart code that complies with both.

\ifappendix
For example, the spec defines outgoing communication between madGLP agents through a system predicate:
\begin{verbatim}
global_send(T, G, Q) :- known(T?) | '_send'(T?, G?, Q?).
\end{verbatim}
The guard \texttt{known(T?)} succeeds once $T$ is assigned; the builtin \texttt{'\_send'} globalizes $T$ and queues the message $(G := T^\uparrow, Q)$ for delivery to $Q$.
\fi

The AI used throughout is Claude (Anthropic). While authority flows math$\to$spec$\to$Dart, the three layers were harmonised via extensive back-and-forth: running the Dart implementation and deriving the specification surfaced defects, several of which proved to be at the mathematical level rather than the code level, and human review caught a conceptual error the formalism had absorbed from an earlier draft. \ifappendix The remarks in this paper record design rationale developed during spec derivation. The English+code specifications served the derivation; their content is now absorbed into this paper's appendices.\fi \ifappendix \Cref{app:dev-history} records the substantive revisions and the diagnostics that drove them.\fi

\ifappendix
\Cref{app:dev-history} gives the full account.
\fi

\section{madGLP is Grassroots}
\label{sec:grassroots}

A protocol is \temph{grassroots}~\cite{lewis2026volitional} if any two disjoint groups of agents can each run it independently---every interleaving of their independent correct runs is a correct run of the combined group (\emph{obliviousness})---yet when brought together the combined group can produce behaviours that neither group could on its own (\emph{interactivity}).  The maGLP protocol is grassroots~\cite{shapiro2025glp}. By the \emph{madGLP protocol} we mean the family with one madGLP transition system per finite set of agents (\appref{def:madglp-ts}{\arxivref}), and likewise for maGLP.

\begin{theorem}[madGLP is Grassroots]\label{thm:madglp-grassroots}
The madGLP protocol is grassroots.
\end{theorem}

\ifappendix

\subsection{Protocols and the Grassroots Property}
\label{subsec:protocols}

A protocol is a family of multiagent transition systems, one for each finite set of agents $P \subset \Pi$, sharing an underlying set of local states with a designated initial state.

\begin{definition}[Local-States Function]\label{def:local-states-function}
A \temph{local-states function} $S: 2^\Pi \to 2^\calS$ maps every finite set of agents $P \subset \Pi$ to a set of local states $S(P) \subset \calS$ that includes $s_0$ and satisfies $P \subset P' \implies S(P) \subset S(P')$.
\end{definition}

Given a local-states function $S$, we use $C(P) := S(P)^P$ for configurations over $P$ and $c_0(P) := \{s_0\}^P$ for the initial configuration.

\begin{definition}[Protocol]\label{def:protocol}
A \temph{protocol} $\calF$ over a local-states function $S$ is a family of multiagent transition systems with exactly one transition system $\calF(P) = (C(P), c_0(P), T(P), {\sim(P)})$ for every finite $P \subset \Pi$.
\end{definition}

A \emph{correct} run of a transition system is safe and live (Definition~\ref{def:ts-basic}); the runs here are safe by construction, so correctness coincides with liveness.

To define grassroots formally, we first define the interleaving of runs of two disjoint groups.

\begin{definition}[Interleaving]\label{def:interleaving}
Let $P, P' \subset \Pi$ be disjoint nonempty sets of agents, $r = c_0, c_1, \ldots$ a run of $\calF(P)$, and $r' = d_0, d_1, \ldots$ a run of $\calF(P')$.  An \temph{interleaving} of $r$ and $r'$ is a sequence $e_0, e_1, \ldots$ of configurations in $C(P \cup P')$ for which there exist non-decreasing sequences of indices $(i_k)_{k \geq 0}$ and $(j_k)_{k \geq 0}$ with $i_0 = j_0 = 0$ such that for every $k \geq 0$:
\begin{enumerate}
\item $(e_k)_p = (c_{i_k})_p$ for every $p \in P$,
\item $(e_k)_q = (d_{j_k})_q$ for every $q \in P'$,
\item if $e_{k+1}$ exists, then exactly one of:
  (a) $i_{k+1} = i_k + 1$ and $j_{k+1} = j_k$ (a $P$-step), or
  (b) $i_{k+1} = i_k$ and $j_{k+1} = j_k + 1$ (a $P'$-step).
\end{enumerate}
Moreover, if $r$ is finite of length $n$ then $i_k = n$ for some $k$, and if $r$ is infinite then for every $m \ge 0$ there is a $k$ with $i_k = m$; likewise for $r'$ and $(j_k)$.
\end{definition}

An interleaving is well-defined: since $P \subset P \cup P'$, the local-states function ensures $S(P) \subset S(P \cup P')$, and similarly for $P'$, so each $e_k$ is a valid configuration in $C(P \cup P')$.  Also, $e_0 = c_0(P \cup P')$, since $(e_0)_p = (c_0)_p = s_0$ for all $p \in P$ and $(e_0)_q = (d_0)_q = s_0$ for all $q \in P'$.

\begin{definition}[Oblivious, Interactive, Grassroots]\label{def:grassroots}
A protocol $\calF$ is:
\begin{enumerate}
\item \temph{oblivious} if for every disjoint nonempty $P, P' \subset \Pi$,
    every interleaving of a correct run of $\calF(P)$ and a correct run of $\calF(P')$ is a correct run of $\calF(P \cup P')$.
\item \temph{interactive} if for every disjoint nonempty $P, P' \subset \Pi$, there exists a correct run $\hat{r}$ of $\calF(P \cup P')$ such that for every correct run $r$ of $\calF(P)$, every correct run $r'$ of $\calF(P')$, and every interleaving $e$ of $r$ and $r'$, $\hat{r} \neq e$.
\item \temph{grassroots} if it is oblivious and interactive.
\end{enumerate}
\end{definition}

In an interleaving, each step changes the local states of agents in only one group; therefore, any transaction whose active participants span both groups yields a step that cannot occur in any interleaving. The interleaving and obliviousness definitions are the machine-layer specialisation, without guards, of the volitional grassroots framework of~\cite{lewis2026volitional}; interactivity is stated via runs that are not interleavings, as in~\cite{shapiro2025glp}, since all madGLP transactions are unary and none is an interaction in the sense of~\cite{lewis2026volitional}.

\subsection{Transactions-Based Grassroots Protocols}
\label{subsec:tb-protocols}

\begin{definition}[Transactions Over a Local-States Function]\label{def:transactions-lsf}
Let $S$ be a local-states function. A set of transactions $R$ is \temph{over $S$} if every transaction $t \in R$ is a transition over participants $Q$ and local states $S(Q')$ for some $Q \subseteq Q' \subset \Pi$. Given such $R$ and $P \subset \Pi$:
$$R(P) := \{t \in R : t \text{ has participants } Q \subseteq P \text{ and is over } S(Q') \text{ for some } Q' \subseteq P\}$$
\end{definition}

\begin{definition}[Transactions-Based Protocol]\label{def:tb-protocol}
Let $S$ be a local-states function and $R$ a set of transactions over $S$ with equivalence $\sim$. The \temph{protocol $\calF$ over $R$, $S$, and $\sim$} is defined by $\calF(P) := (C(P), c_0(P), R(P){\uparrow}P, {\sim}{\uparrow}P)$ for each $P \subset \Pi$, where $R(P){\uparrow}P$ and ${\sim}{\uparrow}P$ are the $P$-closures (Definition~\ref{def:closure}) of $R(P)$ and of ${\sim}$ restricted to $R(P)$.
\end{definition}

Since liveness applies to all equivalence classes, any class $[t]$ in $R(P \cup P')/\!\sim$ whose transactions have participants spanning both $P$ and $P'$ could obstruct obliviousness if enabled in an interleaving.  The following proposition identifies the condition under which this does not occur.

\begin{proposition}[Obliviousness Condition]\label{prop:tb-oblivious}
A transactions-based protocol is oblivious provided that for every disjoint nonempty $P, P' \subset \Pi$, no equivalence class whose transactions have participants spanning both $P$ and $P'$ is ever enabled in any interleaving of correct runs of $\calF(P)$ and $\calF(P')$.
\end{proposition}

\begin{proof}
Let $\calF$ be a transactions-based protocol over transactions $R$, local-states function $S$, and equivalence $\sim$.  Let $P, P' \subset \Pi$ be disjoint and nonempty, $r = c_0, c_1, \ldots$ a correct run of $\calF(P)$, $r' = d_0, d_1, \ldots$ a correct run of $\calF(P')$, and $e = e_0, e_1, \ldots$ an interleaving of $r$ and $r'$.

\emph{Safety.}  We show that $e$ is a run of $\calF(P \cup P')$.  We have $e_0 = c_0(P \cup P')$ as noted above.  Consider a $P$-step $e_k \rightarrow e_{k+1}$.  The corresponding transition $c_{i_k} \rightarrow c_{i_k+1}$ is in $T(P) = R(P){\uparrow}P$, so there exists a transaction $t \in R(P)$ with participants $Q \subseteq P$.  Since $P \subseteq P \cup P'$, we have $R(P) \subseteq R(P \cup P')$, so $t \in R(P \cup P')$.  In the step $e_k \rightarrow e_{k+1}$, the $Q$-agents change as specified by $t$ and all other agents---both $P \setminus Q$ and $P'$---remain stationary.  By the definition of closure, $e_k \rightarrow e_{k+1} \in t{\uparrow}(P \cup P') \subseteq T(P \cup P')$.  The case of a $P'$-step is symmetric.

\emph{Liveness.}  Suppose for contradiction that some class $[t]$ in $R(P \cup P')/\!\sim$ is enabled in every configuration of some suffix of $e$ in which no member of $[t]$ occurs.  By assumption, no class with participants spanning both $P$ and $P'$ is ever enabled in the interleaving, and equivalent transactions have the same participants (Section~\ref{sec:multiagent-ts}), so the participants $Q$ of the transactions of $[t]$ satisfy $Q \subseteq P$ or $Q \subseteq P'$; without loss of generality $Q \subseteq P$.  Enablement of $[t]$ at $e_k$ depends only on the states of the agents in $P$, which match those of $r$ at index $i_k$.  Since the interleaving exhausts $r$ (Definition~\ref{def:interleaving}), the indices $i_k$ over the suffix cover a tail of $r$, so $[t]$ is enabled at every configuration of that tail; as no member of $[t]$ occurs in the suffix of $e$, none occurs along that tail of $r$.  This contradicts the correctness of $r$.  The case $Q \subseteq P'$ is symmetric.
\end{proof}

\begin{theorem}[Transactions-Based Grassroots]\label{thm:interactive-grassroots}
A transactions-based protocol that satisfies the condition of Proposition~\ref{prop:tb-oblivious} and is interactive is grassroots.
\end{theorem}

\begin{proof}
By Proposition~\ref{prop:tb-oblivious}, the protocol is oblivious.  If it is also interactive, it is grassroots by Definition~\ref{def:grassroots}.
\end{proof}

\begin{proof}[Proof of Theorem~\ref{thm:madglp-grassroots}]
By Theorem~\ref{thm:interactive-grassroots}, it suffices to show that madGLP is a transactions-based protocol satisfying the obliviousness condition and that it is interactive.

\emph{madGLP is transactions-based}: By Definition~\ref{def:madglp-ts}, madGLP is defined via transactions (Reduce, Send, Receive) over a local-states function with a common initial state and transaction equivalence (two transactions are equivalent if they are the same type at the same agent, reducing the same goal with the same clause for Reduce, or processing the same message for Send and Receive).

\emph{Obliviousness condition}: All madGLP transactions are unary (Reduce, Send, and Receive each act on a single agent), so no transaction has participants spanning both $P$ and $P'$.  Therefore no equivalence class with spanning participants exists, and the condition of Proposition~\ref{prop:tb-oblivious} is satisfied vacuously.

\emph{Interactivity}: Let $P, P' \subset \Pi$ be disjoint and nonempty.  Let $\hat{r}$ be a correct run of madGLP$(P \cup P')$ in which some $p \in P$ executes a Reduce that outputs $\mathit{msg}(q, X?)$ on its network output stream for some $q \in P'$, triggering a \texttt{global\_send} goal addressed to $q$'s index-0 serializer.  A subsequent Send at $p$ places the cold-call message in the communication channel to $q$, and a Receive at $q$ localizes it: the reader global name $\_r(p,i)$ in the message yields a global writers table entry $(Z_q, p, i)$ at $q$ recording $p$ (Definition~\ref{def:localize}).  The run $\hat{r}$ is correct: the cold-call is delivered (by fair message delivery) and processed, and every enabled equivalence class is eventually taken.

In madGLP$(P')$, every message in a communication channel is placed there by a Send of an agent of $P'$, and the global names it carries are the sending agent's (Definition~\ref{def:globalize}); hence every global writers table entry created by Localize records a remote agent in $P'$.  The entry $(Z_q, p, i)$ records $p \notin P'$, so the local state of $q$ following the Receive occurs in no run of madGLP$(P')$, and no interleaving contains this configuration of $\hat{r}$.  Hence $\hat{r}$ is not an interleaving, and madGLP is interactive.
\end{proof}

\fi

\ifappendix
\section{Related Work}
\label{sec:related-work}

\mypara{Concurrent Constraint Programming}
Concurrent constraint programming~\cite{saraswat1991semantic} organises computation around a shared store accessed by concurrent agents, accumulating constraints rather than assigning single-occurrence variables. GLP belongs to the family of committed-choice concurrent logic languages~\cite{shapiro1983subset,ueda1986guarded,clark1986parlog} surveyed in~\cite{shapiro1989family}, and may be understood as Flat Concurrent Prolog~\cite{mierowsky1985fcp} with SRSW added---resolving the semantic difficulties of read-only unification~\cite{levi1985readonly} so that term matching suffices.

\mypara{Modes and Linearity}
Ueda's moded Flat GHC~\cite{ueda1994moded,ueda1995io} assigns polarity to every variable occurrence and guarantees each variable is written exactly once, establishing ``programming as wiring''; his subsequent linearity analysis~\cite{ueda2001resource} identifies single-reader variables for compile-time garbage collection. GLP enforces both single-reader and single-writer universally as syntactic restrictions, enabling its distributed execution model. The SRSW discipline connects to linear logic~\cite{girard1987linear,wadler1993taste,hodas1994lolli}: each assignment is produced and consumed at most once.

\mypara{Session Types and Linear Logic}
GLP's writer/reader pairs connect to session types via the Curry-Howard correspondence~\cite{caires2010session,wadler2012propositions,wadler2014propositions}: linear logic propositions interpret as session types where channels are used exactly once. Multiparty session types~\cite{honda2016multiparty} coordinate communication among multiple agents; a companion paper~\cite{shapiro2026types} develops a type system for GLP along these lines. Whereas session types address communication-protocol safety, this paper addresses correctness of a distributed implementation against an abstract specification. The forwarder role of \texttt{global\_send} parallels the forwarding processes of the $\pi$-calculus encodings of synchronous into asynchronous communication~\cite{milner1992calculus,honda1991object,boudol1992asynchrony}; a global link differs in carrying exactly one write-once assignment from its fixed anchor to the reader's current holder.

\mypara{Futures and Promises}
A GLP reader/writer pair is operationally a future or promise~\cite{baker1977future,friedman1976impact,flanagan1999semantics}: the writer creates a promise, the reader awaits fulfilment. Pruiksma and Pfenning~\cite{pruiksma2022futures} develop linear futures with write-once shared-memory semantics analogous to GLP's. Unlike mainstream futures, GLP's may carry terms with further reader/writer pairs, enabling recursive communication structures shared with channel-passing calculi~\cite{milner1999communicating}; none of these works address distributed implementation across independent agents---the specific contribution of madGLP's global links.

\mypara{Transition Systems and Refinement}
Our correctness framework uses simulation relations between transition systems, following the I/O automata tradition~\cite{lynch1987hierarchical,lynch1988introduction,lynch1995forward,lynch1996forward} and Abadi and Lamport's refinement-mapping theory~\cite{abadi1991existence}. We use \emph{outcome equivalence} rather than trace refinement, and forward simulation directly (without history or prophecy variables): the SO invariant provides sufficient structure. We use simulation rather than bisimulation~\cite{sangiorgi1998bisimulation} because committed-choice clause selection means dGLP and cGLP need not be bisimilar even though they produce the same outcomes.

\mypara{Atomic Transactions}
Our showing that binary transactions can be implemented by sequences of unary transactions builds on the theory of atomic transactions~\cite{lynch1988atomic,lynch1994atomic,bernstein1987concurrency} and on linearizability~\cite{herlihy1990linearizability}, which defines concurrent correctness via linearization points. Our decomposition follows a different route: disjoint substitution commutativity (a consequence of SO) makes the order of unary transactions irrelevant to outcomes.

\mypara{Grassroots and Monotonic Distributed Systems}
The grassroots property---that any agent subset forms a functioning subsystem---was introduced for distributed platforms~\cite{shapiro2023grassrootsBA} via multiagent transition systems with monotonicity. The definition we use follows~\cite{lewis2026volitional}. Applications include grassroots social networks~\cite{shapiro2023gsn}, cryptocurrencies~\cite{shapiro2024gc,lewis2023grassroots}, and federated governance~\cite{halpern2024federated,shapiro2025GF}. CRDTs~\cite{shapiro2011conflict,kleppmann2022making} achieve eventual consistency via commuting updates; the Blocklace~\cite{almeida2024blocklace} demonstrates the compositionality grassroots systems require.

\mypara{Related Language Implementations}
Several concurrent languages share design goals with GLP: Erlang~\cite{armstrong2010erlang,lindahl2006practical}, Go~\cite{pike2012go}, Rust~\cite{matsakis2014rust}, session-typed languages~\cite{cooper2006links,lindley2017lightweight,toninho2013higher}, capability-based actors~\cite{clebsch2015deny}, and active-object languages with futures~\cite{johnsen2011abs}. GLP differs in its logic programming foundation, with SRSW emerging from logical structure rather than a separate type system; none provide formal proofs that their distributed implementations correctly implement their abstract semantics.

\mypara{AI-Assisted Formal Methods}
AI-driven derivation of implementations from formal specifications is an emerging discipline~\cite{fowler2025sdd,mundler2025type,blinn2024typed}. Our approach shares the premise that formal specifications constrain AI-generated code, but employs three layers of abstraction (math$\to$informal spec$\to$Dart) with bidirectional feedback, and at greater scale: the specifications in this paper underwent several major revisions driven by difficulties discovered during AI implementation (Section~\ref{sec:ai-methodology}).

\fi

\section{Conclusion}
\label{sec:conclusion}

\mypara{Summary}
\ifappendix We have presented dGLP and madGLP---deterministic operational semantics that implement cGLP and maGLP respectively---and proved them correct with respect to their abstract counterparts, madGLP for authorised runs.\fi Disjoint substitution commutativity (from the SO invariant)\ifappendix\ enables correctness proofs that carry over directly from the single-agent to the multiagent case\fi. The madGLP design---local variable pairs connected by global links---provides a clean separation between local computation and inter-agent communication, supporting both the formal correctness proofs and practical implementation.

\begin{credits}
\subsubsection{\discintname}
The author has no competing interests to declare that are relevant to the content of this article.
\end{credits}

\clearpage
\bibliography{bib}

@article{keidar2026secure,
  author    = {Idit Keidar and Ehud Shapiro},
  title     = {Secure GLP: Executing Digital Social Contracts without Consensus},
  year      = {2026},
  journal      = {In preparation}
}

@article{lewis2026volitional,
  title={Volition-Guarded Multiagent Atomic Transactions: Describing People and their Machines},
  author={Lewis-Pye, Andy and Shapiro, Ehud},
  journal={arXiv preprint arXiv:2604.25596},
  year={2026},
  url = {https://arxiv.org/abs/2604.25596}
}

@inproceedings{saraswat1991semantic,
  author    = {Saraswat, Vijay A. and Rinard, Martin and Panangaden, Prakash},
  title     = {Semantic Foundations of Concurrent Constraint Programming},
  booktitle = {Proceedings of the 18th ACM SIGPLAN-SIGACT Symposium on Principles of Programming Languages (POPL)},
  pages     = {333--352},
  publisher = {ACM},
  year      = {1991},
  doi       = {10.1145/99583.99627}
}

@article{sangiorgi1998bisimulation,
  author    = {Sangiorgi, Davide},
  title     = {On the Bisimulation Proof Method},
  journal   = {Mathematical Structures in Computer Science},
  volume    = {8},
  number    = {5},
  pages     = {447--479},
  year      = {1998},
  publisher = {Cambridge University Press},
  doi       = {10.1017/S0960129598002527}
}

@article{herlihy1990linearizability,
  author    = {Herlihy, Maurice P. and Wing, Jeannette M.},
  title     = {Linearizability: A Correctness Condition for Concurrent Objects},
  journal   = {ACM Transactions on Programming Languages and Systems},
  volume    = {12},
  number    = {3},
  pages     = {463--492},
  year      = {1990},
  publisher = {ACM},
  doi       = {10.1145/78969.78972}
}

@article{honda2016multiparty,
  author    = {Honda, Kohei and Yoshida, Nobuko and Carbone, Marco},
  title     = {Multiparty Asynchronous Session Types},
  journal   = {Journal of the ACM},
  volume    = {63},
  number    = {1},
  pages     = {9:1--9:67},
  year      = {2016},
  publisher = {ACM},
  doi       = {10.1145/2827695}
}

@article{flanagan1999semantics,
  author    = {Flanagan, Cormac and Felleisen, Matthias},
  title     = {The Semantics of Future and an Application},
  journal   = {Journal of Functional Programming},
  volume    = {9},
  number    = {1},
  pages     = {1--31},
  year      = {1999},
  publisher = {Cambridge University Press},
  doi       = {10.1017/S0956796899003329}
}

@article{shapiro2026types,
  title={Types for Grassroots Logic Programs (Full version)},
  author={Shapiro, Ehud},
  journal={arXiv preprint arXiv:2601.17957},
  year={2026},
  url = {https://arxiv.org/abs/2601.17957}
}

@article{shapiro2026implementing,
  title={Implementing Grassroots Logic Programs with Multiagent Transition Systems and AI (Full Version)},
  author={Shapiro, Ehud},
  journal={arXiv:2602.06934.  Summary to appear in Proc. of LOPSTR+PPDP'26},
  year={2026},
  url = {https://arxiv.org/abs/2602.06934}
}

@article{friedman1976impact,
  author = {Friedman, Daniel P. and Wise, David S.},
  title = {The Impact of Applicative Programming on Multiprocessing},
  journal = {Indiana University Computer Science Department Technical Report},
  number = {TR-26},
  year = {1976}
}

@misc{dart2024,
     author       = {{Google}},
     title        = {Dart Programming Language},
     howpublished = {\url{https://dart.dev}},
     year         = {2024},
  url = {https://dart.dev/}
}

@inproceedings{mundler2025type,
  author    = {M\"{u}ndler, Niels and Guerraoui, Rachid and Vechev, Martin},
  title     = {Type-Constrained Code Generation with Language Models},
  booktitle = {Proceedings of the 46th ACM SIGPLAN Conference on Programming Language Design and Implementation (PLDI)},
  publisher = {ACM},
  year      = {2025},
  note      = {Available at \url{https://arxiv.org/abs/2504.09246}}
}

@inproceedings{blinn2024typed,
  author    = {Blinn, Andrew and Li, Xiang and Kim, June Hyung and Omar, Cyrus},
  title     = {Statically Contextualizing Large Language Models with Typed Holes},
  booktitle = {Proceedings of the ACM on Programming Languages (OOPSLA)},
  volume    = {8},
  number    = {OOPSLA2},
  pages     = {1--29},
  publisher = {ACM},
  year      = {2024},
  doi       = {10.1145/3689746}
}

@article{fowler2025sdd,
  title   = {Understanding Spec-Driven-Development: Kiro, spec-kit, and Tessl},
  author  = {Fowler, Martin},
  journal = {MartinFowler.com},
  year    = {2025},
  month   = {October},
  day     = {15},
  url     = {https://martinfowler.com/articles/exploring-gen-ai/sdd-3-tools.html},
  note    = {Analyzes the shift toward using formal specifications and types as the primary interface for AI code generation}
}

@inproceedings{lynch1987hierarchical,
  author    = {Lynch, Nancy A. and Tuttle, Mark R.},
  title     = {Hierarchical Correctness Proofs for Distributed Algorithms},
  booktitle = {Proceedings of the 6th Annual ACM Symposium on Principles of Distributed Computing (PODC)},
  pages     = {137--151},
  year      = {1987},
  publisher = {ACM},
  doi       = {10.1145/41840.41852}
}

@techreport{lynch1988introduction,
  author      = {Lynch, Nancy A. and Tuttle, Mark R.},
  title       = {An Introduction to Input/Output Automata},
  institution = {Laboratory for Computer Science, Massachusetts Institute of Technology},
  number      = {MIT/LCS/TM-373},
  year        = {1988}
}

@article{lynch1995forward,
  author    = {Lynch, Nancy A. and Vaandrager, Frits W.},
  title     = {Forward and Backward Simulations: {I}. {U}ntimed Systems},
  journal   = {Information and Computation},
  volume    = {121},
  number    = {2},
  pages     = {214--233},
  year      = {1995},
  publisher = {Elsevier},
  doi       = {10.1006/inco.1995.1134}
}

@article{lynch1996forward,
  author    = {Lynch, Nancy A. and Vaandrager, Frits W.},
  title     = {Forward and Backward Simulations: {II}. {T}iming-Based Systems},
  journal   = {Information and Computation},
  volume    = {128},
  number    = {1},
  pages     = {1--25},
  year      = {1996},
  publisher = {Elsevier},
  doi       = {10.1006/inco.1996.0060}
}

@article{abadi1991existence,
  author    = {Abadi, Mart\'{\i}n and Lamport, Leslie},
  title     = {The Existence of Refinement Mappings},
  journal   = {Theoretical Computer Science},
  volume    = {82},
  number    = {2},
  pages     = {253--284},
  year      = {1991},
  publisher = {Elsevier},
  doi       = {10.1016/0304-3975(91)90224-P}
}

@article{mierowsky1985fcp,
  title={On the implementation of Flat Concurrent Prolog},
  author={Mierowsky, C. and Taylor, S. and Shapiro, E. and Levy, J. and Safra, M.},
  journal={Proceedings of the 1985 Symposium on Logic Programming},
  pages={276--286},
  year={1985},
  publisher={IEEE}
}

@article{ueda1994moded,
  doi = {10.1007/BF03038307},
  title={Moded Flat GHC and Its Message-Oriented Implementation Technique},
  author={Ueda, Kazunori and Morita, Masao},
  journal={New Generation Computing},
  volume={12},
  number={4},
  pages={337--368},
  year={1994}
}

@inproceedings{ueda1995io,
  doi = {10.1007/BFb0026579},
  title={I/O mode analysis in concurrent logic programming},
  author={Ueda, Kazunori and Morita, Masao},
  booktitle={Proceedings of the International Symposium on Theory and Practice of Parallel Programming},
  pages={356--368},
  year={1995},
  publisher={Springer}
}

@article{ueda2001resource,
  doi = {10.1007/3-540-45500-0\_5},
  title={Resource-passing concurrent programming},
  author={Ueda, Kazunori},
  journal={Proceedings of TACS 2001},
  pages={95--126},
  year={2001},
  publisher={Springer}
}

@inproceedings{levi1985readonly,
  author = {Levi, Giorgio and Palamidessi, Catuscia},
  title = {The Semantics of the Read-Only Variable},
  booktitle = {Proc. Symposium on Logic Programming},
  publisher = {IEEE},
  pages = {128--137},
  year = {1985}
}

@inproceedings{johnsen2011abs,
  author    = {Johnsen, Einar Broch and H{\"a}hnle, Reiner and Sch{\"a}fer, Jan and Schlatte, Rudolf and Steffen, Martin},
  title     = {{ABS}: A Core Language for Abstract Behavioral Specification},
  booktitle = {Proceedings of the 9th International Symposium on Formal Methods for Components and Objects (FMCO)},
  series    = {Lecture Notes in Computer Science},
  volume    = {6957},
  pages     = {142--164},
  publisher = {Springer},
  year      = {2011},
  doi       = {10.1007/978-3-642-25271-6\_8}
}

@inproceedings{matsakis2014rust,
  author    = {Matsakis, Nicholas D. and Klock II, Felix S.},
  title     = {The {R}ust Language},
  booktitle = {Proceedings of the 2014 ACM SIGAda Annual Conference on High Integrity Language Technology (HILT)},
  pages     = {103--104},
  publisher = {ACM},
  year      = {2014},
  doi       = {10.1145/2663171.2663188}
}

@misc{pike2012go,
  author    = {Pike, Rob},
  title     = {Go at {G}oogle: Language Design in the Service of Software Engineering},
  howpublished = {Keynote at SPLASH 2012},
  year      = {2012},
  note      = {Available at \url{https://go.dev/talks/2012/splash.article}}
}

@article{armstrong2010erlang,
  author    = {Armstrong, Joe},
  title     = {Erlang},
  journal   = {Communications of the ACM},
  volume    = {53},
  number    = {9},
  pages     = {68--75},
  year      = {2010},
  publisher = {ACM},
  doi       = {10.1145/1810891.1810910}
}

@inproceedings{lindahl2006practical,
  author    = {Lindahl, Tobias and Sagonas, Konstantinos},
  title     = {Practical Type Inference Based on Success Typings},
  booktitle = {Proceedings of the 8th ACM SIGPLAN International Conference on Principles and Practice of Declarative Programming (PPDP)},
  pages     = {167--178},
  publisher = {ACM},
  year      = {2006},
  doi       = {10.1145/1140335.1140356}
}

@book{milner1999communicating,
  author    = {Milner, Robin},
  title     = {Communicating and Mobile Systems: The $\pi$-Calculus},
  publisher = {Cambridge University Press},
  year      = {1999},
  isbn      = {0-521-65869-1},
  url = {https://www.cambridge.org/9780521658690}
}

@inproceedings{cooper2006links,
  author    = {Cooper, Ezra and Lindley, Sam and Wadler, Philip and Yallop, Jeremy},
  title     = {Links: Web Programming Without Tiers},
  booktitle = {Proceedings of the 5th International Symposium on Formal Methods for Components and Objects (FMCO)},
  series    = {Lecture Notes in Computer Science},
  volume    = {4709},
  pages     = {266--296},
  publisher = {Springer},
  year      = {2006},
  doi       = {10.1007/978-3-540-74792-5\_12}
}

@inproceedings{toninho2013higher,
  author    = {Toninho, Bernardo and Caires, Lu\'{\i}s and Pfenning, Frank},
  title     = {Higher-Order Processes, Functions, and Sessions: A Monadic Integration},
  booktitle = {Proceedings of the 22nd European Symposium on Programming (ESOP)},
  series    = {Lecture Notes in Computer Science},
  volume    = {7792},
  pages     = {350--369},
  publisher = {Springer},
  year      = {2013},
  doi       = {10.1007/978-3-642-37036-6\_20}
}

@inproceedings{clebsch2015deny,
  author    = {Clebsch, Sylvan and Drossopoulou, Sophia and Blessing, Sebastian and McNeil, Andy},
  title     = {Deny Capabilities for Safe, Fast Actors},
  booktitle = {Proceedings of the 5th International Workshop on Programming Based on Actors, Agents, and Decentralized Control (AGERE!)},
  pages     = {1--12},
  publisher = {ACM},
  year      = {2015},
  doi       = {10.1145/2824815.2824816}
}

@inproceedings{caires2010session,
  author    = {Caires, Lu\'{\i}s and Pfenning, Frank},
  title     = {Session Types as Intuitionistic Linear Propositions},
  booktitle = {Proceedings of the 21st International Conference on Concurrency Theory (CONCUR)},
  series    = {Lecture Notes in Computer Science},
  volume    = {6269},
  pages     = {222--236},
  publisher = {Springer},
  year      = {2010},
  doi       = {10.1007/978-3-642-15375-4\_16}
}

@article{pruiksma2022futures,
  author    = {Pruiksma, Klaas and Pfenning, Frank},
  title     = {Back to Futures},
  journal   = {Journal of Functional Programming},
  volume    = {32},
  pages     = {e4},
  year      = {2022},
  publisher = {Cambridge University Press},
  doi       = {10.1017/S0956796822000016}
}

@inproceedings{wadler2012propositions,
  author    = {Wadler, Philip},
  title     = {Propositions as Sessions},
  booktitle = {Proceedings of the 17th ACM SIGPLAN International Conference on Functional Programming (ICFP)},
  pages     = {273--286},
  publisher = {ACM},
  year      = {2012},
  doi       = {10.1145/2364527.2364568}
}

@article{wadler2014propositions,
  author    = {Wadler, Philip},
  title     = {Propositions as Sessions},
  journal   = {Journal of Functional Programming},
  volume    = {24},
  number    = {2--3},
  pages     = {384--418},
  year      = {2014},
  publisher = {Cambridge University Press},
  doi       = {10.1017/S095679681400001X}
}

@article{lindley2017lightweight,
  author    = {Lindley, Sam and Morris, J. Garrett},
  title     = {Lightweight Functional Session Types},
  journal   = {Behavioural Types: From Theory to Tools},
  pages     = {265--286},
  year      = {2017},
  publisher = {River Publishers},
  note      = {Chapter in edited volume}
}

@article{girard1987linear,
  author    = {Girard, Jean-Yves},
  title     = {Linear Logic},
  journal   = {Theoretical Computer Science},
  volume    = {50},
  number    = {1},
  pages     = {1--101},
  year      = {1987},
  publisher = {Elsevier},
  doi       = {10.1016/0304-3975(87)90045-4}
}

@article{wadler1993taste,
  author    = {Wadler, Philip},
  title     = {A Taste of Linear Logic},
  journal   = {Mathematical Structures in Computer Science},
  volume    = {3},
  number    = {4},
  pages     = {367--392},
  year      = {1993},
  publisher = {Cambridge University Press},
  doi       = {10.1017/S0960129500000268}
}

@article{hodas1994lolli,
  author    = {Hodas, Joshua S. and Miller, Dale},
  title     = {Logic Programming in a Fragment of Intuitionistic Linear Logic},
  journal   = {Information and Computation},
  volume    = {110},
  number    = {2},
  pages     = {327--365},
  year      = {1994},
  doi       = {10.1006/inco.1994.1036}
}

@article{lynch1988atomic,
  title={Atomic Transactions},
  author={Lynch, Nancy A. and Merritt, Michael and Weihl, William and Fekete, Alan},
  journal={Morgan Kaufmann},
  year={1988},
  publisher={Morgan Kaufmann Publishers}
}

@book{lynch1994atomic,
  author={Lynch, Nancy A. and Merritt, Michael and Weihl, William E. and Fekete, Alan},
  title={Atomic Transactions: In Concurrent and Distributed Systems},
  year={1994},
  publisher={Morgan Kaufmann Publishers},
  address={San Francisco, CA}
}

@book{bernstein1987concurrency,
  title={Concurrency Control and Recovery in Database Systems},
  author={Bernstein, Philip A. and Hadzilacos, Vassos and Goodman, Nathan},
  year={1987},
  publisher={Addison-Wesley}
}

@article{shapiro2025glp,
  title={GLP: A Grassroots, Multiagent, Concurrent, Logic Programming Language for AI (full version)},
  author={Shapiro, Ehud},
  journal={arXiv preprint arXiv:2510.15747, Summary to appear in Proc. of ICLP'26},
  year={2025},
  url = {https://arxiv.org/abs/2510.15747}
}

@inproceedings{ueda1986guarded,
  doi = {10.1007/3-540-16479-0\_17},
  title={Guarded Horn Clauses},
  author={Ueda, Kazunori},
  booktitle={Logic Programming '85},
  series={Lecture Notes in Computer Science},
  volume={221},
  pages={168--179},
  year={1986},
  publisher={Springer}
}

@article{shapiro2025GF,
  title={Grassroots Federation: Fair Democratic Governance at Scale},
  author={Talmon, Nimrod and Shapiro, Ehud},
  journal={arXiv preprint arXiv:2505.02208; also Proc. of AAMAS'26},
  year={2025},
  url = {https://arxiv.org/abs/2505.02208}
}

@inproceedings{shapiro2025atomic,
  title={Grassroots Platforms with Atomic Transactions: Social Graphs, Cryptocurrencies, and Democratic Federations},
  author={Shapiro, Ehud},
  booktitle={Proceedings of the 27th International Conference on Distributed Computing and Networking},
  pages={71--81},
  doi = {10.1145/3772290.3772309},
  note = {arXiv preprint arXiv:2502.11299},
  year={2026}
}

@article{clark1986parlog,
  doi = {10.1145/5001.5390},
  title={PARLOG: parallel programming in logic},
  author={Clark, Keith and Gregory, Steve},
  journal={ACM Transactions on Programming Languages and Systems (TOPLAS)},
  volume={8},
  number={1},
  pages={1--49},
  year={1986},
  publisher={ACM New York, NY, USA}
}

@article{shapiro1983subset,
  title={A subset of Concurrent Prolog and its interpreter},
  author={Shapiro, Ehud},
  journal={ICOT Technical Report, TR-003},
  year={1983}
}

@inproceedings{shapiro2011conflict,
  title={Conflict-free replicated data types},
  author={Shapiro, Marc and Pregui{\c{c}}a, Nuno and Baquero, Carlos and Zawirski, Marek},
  booktitle={Stabilization, Safety, and Security of Distributed Systems: 13th International Symposium, SSS 2011, Grenoble, France, October 10-12, 2011. Proceedings 13},
  pages={386--400},
  year={2011},
  organization={Springer}
}

@inproceedings{baker1977future,
  author    = {Henry G. Baker and Carl Hewitt},
  title     = {The Incremental Garbage Collection of Processes},
  booktitle = {Proceedings of the 1977 Symposium on Artificial Intelligence and Programming Languages},
  year      = {1977},
  pages     = {55--59},
  publisher = {ACM},
  doi       = {10.1145/800228.806932}
}

@article{houri1989sequential,
  doi = {10.1016/0743-1066(89)90011-3},
  title={A sequential abstract machine for Flat Concurrent Prolog},
  author={Houri, Avshalom and Shapiro, Ehud},
  journal={The Journal of Logic Programming},
  volume={7},
  number={2},
  pages={85--123},
  year={1989},
  publisher={Elsevier}
}

@inproceedings{kleppmann2022making,
  title={{Making CRDTs Byzantine fault tolerant}},
  author={Kleppmann, Martin},
  booktitle={Proceedings of the 9th Workshop on Principles and Practice of Consistency for Distributed Data},
  pages={8--15},
  year={2022}
}

@article{halpern2024federated,
  title={Federated Assemblies},
  author={Halpern, Daniel and Procaccia, Ariel D and Shapiro, Ehud and Talmon, Nimrod},
  booktitle={AAAI '25},
  journal={Proc AAAI 2025; arXiv preprint arXiv:2405.19129},
  year={2024},
  url = {https://arxiv.org/abs/2405.19129}
}

@article{almeida2024blocklace,
  title={The Blocklace: A  Byzantine-repelling and Universal Conflict-free Replicated Data Type},
  author={Almeida, Paulo S{\'e}rgio and Shapiro, Ehud},
  journal={arXiv preprint arXiv:2402.08068},
  year={2024},
  url = {https://arxiv.org/abs/2402.08068}
}

@article{lewis2023grassroots,
  title={Grassroots Flash: A Payment System for Grassroots Cryptocurrencies},
  author={Lewis-Pye, Andrew and Naor, Oded and Shapiro, Ehud},
  journal={arXiv preprint arXiv:2309.13191},
  year={2023},
  url = {https://arxiv.org/abs/2309.13191}
}

@inproceedings{shapiro2023gsn, 
author={Shapiro, Ehud},
title={Grassroots Social Networking: Serverless, Permissionless Protocols for Twitter/LinkedIn/WhatsApp},
year = {2023},
isbn = {979-8-4007-0225-9/23/09},
publisher = {Association for Computing Machinery},
doi = {10.1145/3599696.3612898},
booktitle = {OASIS ’23},
location = {Rome, Italy},
}

@inproceedings{shapiro2023grassrootsBA, 
author={Shapiro, Ehud},
 title={Grassroots Distributed Systems: Concept, Examples, Implementation and Applications (Brief Announcement)
},year = {2023},
publisher = {LIPICS},
booktitle = {37th International Symposium on Distributed Computing (DISC 2023). (Extended version: arXiv:2301.04391)},
address = {Italy},
notes={Extended version: arXiv preprint arXiv:2301.04391},
pages = {47:1, 47:7},
  url = {https://arxiv.org/abs/2301.04391}
}

@article{shapiro2024gc,
  title={Grassroots Currencies: A Foundation for a Grassroots Digital Economy},
  author={Shapiro, Ehud},
  journal={arXiv preprint arXiv:2202.05619},
  year={2022},
  url = {https://arxiv.org/abs/2202.05619}
}

@article{shapiro1989family,
  doi = {10.1145/72551.72555},
  title={The family of concurrent logic programming languages},
  author={Shapiro, Ehud},
  journal={ACM Computing Surveys (CSUR)},
  volume={21},
  number={3},
  pages={413--510},
  year={1989},
  publisher={ACM New York, NY, USA}
}

@article{shapiro2021multiagent,
  title={Multiagent Transition Systems: Protocol-Stack Mathematics for Distributed Computing},
  author={Shapiro, Ehud},
  journal={arXiv preprint arXiv:2112.13650},
  year={2021},
  url = {https://arxiv.org/abs/2112.13650}
}

@misc{proof,
  title={Proof of Stake FAQ, Ethereum Wiki},
  author={Nakamoto, Satoshi and Bitcoin, A},
  howpublished = { \\ \MYhref{https://eth.wiki/en/concepts/proof-of-stake-faqs}{https://eth.wiki/en/concepts/proof-of-stake-faqs}},
  year={2019}
}

@article{milner1992calculus,
  author    = {Robin Milner and Joachim Parrow and David Walker},
  title     = {A Calculus of Mobile Processes, {I} and {II}},
  journal   = {Information and Computation},
  volume    = {100},
  number    = {1},
  pages     = {1--77},
  year      = {1992}
}

@inproceedings{honda1991object,
  author    = {Kohei Honda and Mario Tokoro},
  title     = {An Object Calculus for Asynchronous Communication},
  booktitle = {Proceedings of the European Conference on Object-Oriented Programming (ECOOP)},
  series    = {Lecture Notes in Computer Science},
  volume    = {512},
  pages     = {133--147},
  publisher = {Springer},
  year      = {1991}
}

@techreport{boudol1992asynchrony,
  author      = {G{\'e}rard Boudol},
  title       = {Asynchrony and the $\pi$-calculus},
  institution = {INRIA},
  number      = {Research Report RR-1702},
  year        = {1992}
}

\ifappendix
\appendix

\section{Detailed madGLP Example Traces}\label{app:madglp-trace}

This appendix provides complete formal traces for three madGLP examples: the client-monitor example (Program~\ref{program:monitor}), the friend-mediated introduction scenario (Bob sends $X?$ to Alice and $X$ to Charlie), and the edge case where both ends of a variable pair are sent to the same agent.

\subsection{Example 1: Client-Monitor Communication}
\label{app:client-monitor}

This example corresponds to Figure~\ref{fig:madglp-example}.

Initial goal: \texttt{client1(Xs)@p, monitor(Xs?)@q}

\mypara{Stage 0: After Initial Cold-call}

\begin{center}
\begin{tabular}{|c|c|}
\hline
\multicolumn{2}{|c|}{\textbf{maGLP State}} \\
\hline
\textbf{Agent $p$ Resolvent} & \textbf{Agent $q$ Resolvent} \\
\hline
$\mathsf{Xs}$ & $\mathsf{Xs?}$ \\
\hline
\multicolumn{2}{|c|}{\textbf{Shared Pairs}: $(\mathsf{Xs}@p, \mathsf{Xs?}@q)$} \\
\hline
\end{tabular}
\end{center}

\begin{center}
\begin{tabular}{|c|c|}
\hline
\multicolumn{2}{|c|}{\textbf{madGLP State}} \\
\hline
\textbf{Agent $p$} & \textbf{Agent $q$} \\
\hline
\multicolumn{2}{|c|}{\textbf{Resolvents (including global\_send goals)}} \\
\hline
$\mathsf{Xs_p}, \mathsf{Xs_p?}$, & $\mathsf{Xs_q}, \mathsf{Xs_q?}$ \\
\texttt{global\_send}$(\mathsf{Xs_p?}, \_r(p,1), q)$ & \\
\hline
\multicolumn{2}{|c|}{\textbf{Assignments}} \\
\hline
(none) & (none) \\
\hline
\multicolumn{2}{|c|}{\textbf{Global Writers Tables}} \\
\hline
(empty) &
\begin{tabular}{ccc}
$k$ & entry \\
\hline
1 & $(\mathsf{Xs_q}, p, 1)$
\end{tabular}
\\
\hline
\multicolumn{2}{|c|}{\textbf{Pending Messages}: (none)} \\
\hline
\end{tabular}
\end{center}

During Globalize, the reader $\mathsf{Xs_p?}$ was globalized as $\_r(p,1)$, spawning \texttt{global\_send}$(\mathsf{Xs_p?}, \_r(p,1), q)$ at $p$ (Definition~\ref{def:globalize}, case~2). No entry is created at $p$---the \texttt{global\_send} goal handles outgoing communication. During Localize at $q$, the reader global name $\_r(p,1)$ created pair $(\mathsf{Xs_q}, \mathsf{Xs_q?})$ with entry $(\mathsf{Xs_q}, p, 1)$ at index~1 in $q$'s global writers table (Definition~\ref{def:localize}, case~2). Agent $q$ will receive the assignment on this link.

\mypara{Stage 1a: After $p$ Reduces (before $q$ Receives)}

Agent $p$ reduces with a clause whose head is $[\mathtt{add}|\mathsf{Xs1_p?}]$, producing the assignment $\mathsf{Xs_p} := [\mathtt{add}|\mathsf{Xs1_p?}]$. This makes $\mathsf{Xs_p?}$ known, triggering the \texttt{global\_send} goal.

\begin{center}
\begin{tabular}{|c|c|}
\hline
\multicolumn{2}{|c|}{\textbf{madGLP State}} \\
\hline
\textbf{Agent $p$} & \textbf{Agent $q$} \\
\hline
\multicolumn{2}{|c|}{\textbf{Resolvents (including global\_send goals)}} \\
\hline
$\mathsf{Xs1_p}, \mathsf{Xs1_p?}$, & $\mathsf{Xs_q}, \mathsf{Xs_q?}$ \\
\texttt{global\_send}$(\mathsf{Xs1_p?}, \_r(p,2), q)$ & \\
\hline
\multicolumn{2}{|c|}{\textbf{Assignments}} \\
\hline
$\mathsf{Xs_p} := [\mathtt{add}|\mathsf{Xs1_p?}]$ & (none) \\
$\mathsf{Xs_p?} := [\mathtt{add}|\mathsf{Xs1_p?}]$ & \\
\hline
\multicolumn{2}{|c|}{\textbf{Global Writers Tables}} \\
\hline
(empty) &
\begin{tabular}{cc}
$k$ & entry \\
\hline
1 & $(\mathsf{Xs_q}, p, 1)$
\end{tabular}
\\
\hline
\multicolumn{2}{|c|}{\textbf{Pending Messages}} \\
\hline
\multicolumn{2}{|c|}{$(\_r(p,1) := [\mathtt{add}|\_r(p,2)],\ q)$} \\
\hline
\end{tabular}
\end{center}

\paragraph{Globalization performed:}
\begin{center}
\begin{tabular}{|l|l|}
\hline
Term before & $[\mathtt{add}|\mathsf{Xs1_p?}]$ \\
\hline
Term after & $[\mathtt{add}|\_r(p,2)]$ \\
\hline
New goal at $p$ & \texttt{global\_send}$(\mathsf{Xs1_p?}, \_r(p,2), q)$ \\
\hline
New entry at $p$ & (none---outgoing link) \\
\hline
\end{tabular}
\end{center}

The old \texttt{global\_send} goal for $\mathsf{Xs_p?}$ fired (since $\mathsf{Xs_p?}$ became known) and completed after adding the message to $M_p$. The value $[\mathtt{add}|\mathsf{Xs1_p?}]$ contains reader $\mathsf{Xs1_p?}$, which is globalized as $\_r(p,2)$ with a new \texttt{global\_send} goal spawned (Definition~\ref{def:globalize}, case~2).

\mypara{Stage 1b: After $q$ Receives}

Agent $q$ receives message and localizes.

\begin{center}
\begin{tabular}{|c|c|}
\hline
\multicolumn{2}{|c|}{\textbf{maGLP State}} \\
\hline
\textbf{Agent $p$ Resolvent} & \textbf{Agent $q$ Resolvent} \\
\hline
$\mathsf{Xs1}$ & $\mathsf{Xs1?}$ \\
\hline
\multicolumn{2}{|c|}{\textbf{Assignments}} \\
\hline
$\mathsf{Xs} := [\mathtt{add}|\mathsf{Xs1}]$ & $\mathsf{Xs?} := [\mathtt{add}|\mathsf{Xs1?}]$ \\
\hline
\multicolumn{2}{|c|}{\textbf{Shared Pairs}: $(\mathsf{Xs1}@p, \mathsf{Xs1?}@q)$} \\
\hline
\end{tabular}
\end{center}

\begin{center}
\begin{tabular}{|c|c|}
\hline
\multicolumn{2}{|c|}{\textbf{madGLP State}} \\
\hline
\textbf{Agent $p$} & \textbf{Agent $q$} \\
\hline
\multicolumn{2}{|c|}{\textbf{Resolvents (including global\_send goals)}} \\
\hline
$\mathsf{Xs1_p}, \mathsf{Xs1_p?}$, & $\mathsf{Xs1_q}, \mathsf{Xs1_q?}$ \\
\texttt{global\_send}$(\mathsf{Xs1_p?}, \_r(p,2), q)$ & \\
\hline
\multicolumn{2}{|c|}{\textbf{Assignments}} \\
\hline
$\mathsf{Xs_p} := [\mathtt{add}|\mathsf{Xs1_p?}]$ & $\mathsf{Xs_q} := [\mathtt{add}|\mathsf{Xs1_q?}]$ \\
$\mathsf{Xs_p?} := [\mathtt{add}|\mathsf{Xs1_p?}]$ & $\mathsf{Xs_q?} := [\mathtt{add}|\mathsf{Xs1_q?}]$ \\
\hline
\multicolumn{2}{|c|}{\textbf{Global Writers Tables}} \\
\hline
(empty) &
\begin{tabular}{cc}
$k$ & entry \\
\hline
2 & $(\mathsf{Xs1_q}, p, 2)$
\end{tabular}
\\
\hline
\multicolumn{2}{|c|}{\textbf{Pending Messages}: (none)} \\
\hline
\end{tabular}
\end{center}

\paragraph{Localization performed:}
\begin{center}
\begin{tabular}{|l|l|}
\hline
Term before & $[\mathtt{add}|\_r(p,2)]$ \\
\hline
Term after & $[\mathtt{add}|\mathsf{Xs1_q?}]$ \\
\hline
New pair at $q$ & $(\mathsf{Xs1_q}, \mathsf{Xs1_q?})$ \\
\hline
New entry at $q$ & $(2, (\mathsf{Xs1_q}, p, 2))$ \\
\hline
New goal at $q$ & (none---incoming link) \\
\hline
\end{tabular}
\end{center}

Reader global name $\_r(p,2)$ is localized by creating pair $(\mathsf{Xs1_q}, \mathsf{Xs1_q?})$, placing $\mathsf{Xs1_q?}$ (reader) in the term, and adding entry $(\mathsf{Xs1_q}, p, 2)$ (Definition~\ref{def:localize}, case~2). Entry~1 at $q$ was removed after assigning $\mathsf{Xs_q}$, and $q$ acknowledged $\_r(p,1)$; on its arrival $p$ releases the pending value, closing link~1 (\S\ref{app:requests-acks}).

\mypara{Stage 2a: After $p$ Reduces (before $q$ Receives)}

Agent $p$ reduces with a clause whose head is $[\mathtt{value}(\mathsf{V_p})|\mathsf{Xs2_p?}]$, producing the assignment $\mathsf{Xs1_p} := [\mathtt{value}(\mathsf{V_p})|\mathsf{Xs2_p?}]$. The clause head contains writer $\mathsf{V_p}$ (the monitor will write the value back) and reader $\mathsf{Xs2_p?}$ (the stream continuation that $p$ will write to).

\begin{center}
\begin{tabular}{|c|c|}
\hline
\multicolumn{2}{|c|}{\textbf{madGLP State}} \\
\hline
\textbf{Agent $p$} & \textbf{Agent $q$} \\
\hline
\multicolumn{2}{|c|}{\textbf{Resolvents (including global\_send goals)}} \\
\hline
$\mathsf{V_p}, \mathsf{V_p?}, \mathsf{Xs2_p}, \mathsf{Xs2_p?}$, & $\mathsf{Xs1_q}, \mathsf{Xs1_q?}$ \\
\texttt{global\_send}$(\mathsf{Xs2_p?}, \_r(p,4), q)$ & \\
\hline
\multicolumn{2}{|c|}{\textbf{Assignments}} \\
\hline
$\mathsf{Xs_p} := [\mathtt{add}|\mathsf{Xs1_p?}]$ & $\mathsf{Xs_q} := [\mathtt{add}|\mathsf{Xs1_q?}]$ \\
$\mathsf{Xs_p?} := [\mathtt{add}|\mathsf{Xs1_p?}]$ & $\mathsf{Xs_q?} := [\mathtt{add}|\mathsf{Xs1_q?}]$ \\
$\mathsf{Xs1_p} := [\mathtt{value}(\mathsf{V_p})|\mathsf{Xs2_p?}]$ & \\
$\mathsf{Xs1_p?} := [\mathtt{value}(\mathsf{V_p})|\mathsf{Xs2_p?}]$ & \\
\hline
\multicolumn{2}{|c|}{\textbf{Global Writers Tables}} \\
\hline
\begin{tabular}{cc}
$i$ & entry \\
\hline
3 & $(\mathsf{V_p}, q)$
\end{tabular}
&
\begin{tabular}{cc}
$k$ & entry \\
\hline
2 & $(\mathsf{Xs1_q}, p, 2)$
\end{tabular}
\\
\hline
\multicolumn{2}{|c|}{\textbf{Pending Messages}} \\
\hline
\multicolumn{2}{|c|}{$(\_r(p,2) := [\mathtt{value}(\_w(p,3))|\_r(p,4)],\ q)$} \\
\hline
\end{tabular}
\end{center}

\paragraph{Globalization performed:}
\begin{center}
\begin{tabular}{|l|l|}
\hline
Term before & $[\mathtt{value}(\mathsf{V_p})|\mathsf{Xs2_p?}]$ \\
\hline
Term after & $[\mathtt{value}(\_w(p,3))|\_r(p,4)]$ \\
\hline
\end{tabular}
\end{center}

$\mathsf{V_p}$ is a writer, so it is globalized as $\_w(p,3)$ and entry $(\mathsf{V_p}, q)$ is created at index~3 in $p$'s global writers table (Definition~\ref{def:globalize}, case~1). Agent $p$ will receive the assignment on this link. $\mathsf{Xs2_p?}$ is a reader, so it is globalized as $\_r(p,4)$ and \texttt{global\_send}$(\mathsf{Xs2_p?}, \_r(p,4), q)$ is spawned (Definition~\ref{def:globalize}, case~2); no entry is created for this outgoing link.

\mypara{Stage 2b: After $q$ Receives}

Agent $q$ receives message and localizes.

\begin{center}
\begin{tabular}{|c|c|}
\hline
\multicolumn{2}{|c|}{\textbf{maGLP State}} \\
\hline
\textbf{Agent $p$ Resolvent} & \textbf{Agent $q$ Resolvent} \\
\hline
$\mathsf{V?},\ \mathsf{Xs2}$ & $\mathsf{V},\ \mathsf{Xs2?}$ \\
\hline
\multicolumn{2}{|c|}{\textbf{Assignments}} \\
\hline
$\mathsf{Xs} := [\mathtt{add}|\mathsf{Xs1}]$ & $\mathsf{Xs?} := [\mathtt{add}|\mathsf{Xs1?}]$ \\
$\mathsf{Xs1} := [\mathtt{value}(\mathsf{V?})|\mathsf{Xs2}]$ & $\mathsf{Xs1?} := [\mathtt{value}(\mathsf{V})|\mathsf{Xs2?}]$ \\
\hline
\multicolumn{2}{|c|}{\textbf{Shared Pairs}: $(\mathsf{V}@q, \mathsf{V?}@p)$, $(\mathsf{Xs2}@p, \mathsf{Xs2?}@q)$} \\
\hline
\end{tabular}
\end{center}

\begin{center}
\begin{tabular}{|c|c|}
\hline
\multicolumn{2}{|c|}{\textbf{madGLP State}} \\
\hline
\textbf{Agent $p$} & \textbf{Agent $q$} \\
\hline
\multicolumn{2}{|c|}{\textbf{Resolvents (including global\_send goals)}} \\
\hline
$\mathsf{V_p}, \mathsf{V_p?}, \mathsf{Xs2_p}, \mathsf{Xs2_p?}$, & $\mathsf{V_q}, \mathsf{V_q?}, \mathsf{Xs2_q}, \mathsf{Xs2_q?}$, \\
\texttt{global\_send}$(\mathsf{Xs2_p?}, \_r(p,4), q)$ & \texttt{global\_send}$(\mathsf{V_q?}, \_w(p,3), p)$ \\
\hline
\multicolumn{2}{|c|}{\textbf{Assignments}} \\
\hline
$\mathsf{Xs_p} := [\mathtt{add}|\mathsf{Xs1_p?}]$ & $\mathsf{Xs_q} := [\mathtt{add}|\mathsf{Xs1_q?}]$ \\
$\mathsf{Xs_p?} := [\mathtt{add}|\mathsf{Xs1_p?}]$ & $\mathsf{Xs_q?} := [\mathtt{add}|\mathsf{Xs1_q?}]$ \\
$\mathsf{Xs1_p} := [\mathtt{value}(\mathsf{V_p})|\mathsf{Xs2_p?}]$ & $\mathsf{Xs1_q} := [\mathtt{value}(\mathsf{V_q})|\mathsf{Xs2_q?}]$ \\
$\mathsf{Xs1_p?} := [\mathtt{value}(\mathsf{V_p})|\mathsf{Xs2_p?}]$ & $\mathsf{Xs1_q?} := [\mathtt{value}(\mathsf{V_q})|\mathsf{Xs2_q?}]$ \\
\hline
\multicolumn{2}{|c|}{\textbf{Global Writers Tables}} \\
\hline
\begin{tabular}{cc}
$i$ & entry \\
\hline
3 & $(\mathsf{V_p}, q)$
\end{tabular}
&
\begin{tabular}{cc}
$k$ & entry \\
\hline
3 & $(\mathsf{Xs2_q}, p, 4)$
\end{tabular}
\\
\hline
\multicolumn{2}{|c|}{\textbf{Pending Messages}: (none)} \\
\hline
\end{tabular}
\end{center}

\paragraph{Localization performed:}
\begin{center}
\begin{tabular}{|l|l|}
\hline
Term before & $[\mathtt{value}(\_w(p,3))|\_r(p,4)]$ \\
\hline
Term after & $[\mathtt{value}(\mathsf{V_q})|\mathsf{Xs2_q?}]$ \\
\hline
\end{tabular}
\end{center}

$\_w(p,3)$ is a writer global name, so localization creates pair $(\mathsf{V_q}, \mathsf{V_q?})$, puts $\mathsf{V_q}$ (the writer) in the term, and spawns \texttt{global\_send}$(\mathsf{V_q?}, \_w(p,3), p)$ with the record $(\mathsf{V_q}, p, 3, p)$ (Definition~\ref{def:localize}, case~1). No entry is created---this is an outgoing link from $q$ to $p$. $\_r(p,4)$ is a reader global name, so localization creates pair $(\mathsf{Xs2_q}, \mathsf{Xs2_q?})$, puts $\mathsf{Xs2_q?}$ (the reader) in the term, and creates entry $(\mathsf{Xs2_q}, p, 4)$ at index~3 (Definition~\ref{def:localize}, case~2)---this is an incoming link. Entry~2 at $q$ was removed after assigning $\mathsf{Xs1_q}$, and $q$ acknowledged $\_r(p,2)$, releasing $p$'s pending value.

\mypara{Global Link Summary}

After Stage 2b, the active global links are:

\begin{center}
\begin{tabular}{|c|c|c|c|c|}
\hline
\textbf{Link} & \textbf{Entry at} & \textbf{Entry} & \textbf{Direction} & \textbf{global\_send at} \\
\hline
V & $p$ (index 3) & $(\mathsf{V_p}, q)$ & $q \to p$ & $q$ \\
\hline
Xs2 & $q$ (index 3) & $(\mathsf{Xs2_q}, p, 4)$ & $p \to q$ & $p$ \\
\hline
\end{tabular}
\end{center}

The ``Direction'' column indicates dataflow: V flows from $q$ to $p$ (the monitor writes the value), while Xs2 flows from $p$ to $q$ (the client writes events). Entries exist only at the receiving end; \texttt{global\_send} goals exist only at the sending end. Each reader-name value is acknowledged by its receiver, and its sender holds it pending until then (\S\ref{app:requests-acks}).

\subsection{Example 2: Friend-Mediated Introduction}
\label{app:friend-intro}

This example demonstrates the critical case where Bob (agent $b$) introduces Alice (agent $a$) to Charlie (agent $c$) by sending them opposite ends of a variable pair. Bob sends reader $X?$ to Alice and writer $X$ to Charlie. Bob's local pair $(X_b, X_b?)$ acts as a relay: when Charlie assigns his writer, the value flows through Bob to Alice.

\mypara{Stage 0: Initial State at Bob}

Bob has a local variable pair $(X_b, X_b?)$ that will be shared with Alice and Charlie.

\begin{center}
\begin{tabular}{|c|c|c|}
\hline
\multicolumn{3}{|c|}{\textbf{madGLP State}} \\
\hline
\textbf{Agent $a$ (Alice)} & \textbf{Agent $b$ (Bob)} & \textbf{Agent $c$ (Charlie)} \\
\hline
\multicolumn{3}{|c|}{\textbf{Resolvents}} \\
\hline
(empty) & $X_b, X_b?$ & (empty) \\
\hline
\multicolumn{3}{|c|}{\textbf{Global Writers Tables}} \\
\hline
(empty) & (empty) & (empty) \\
\hline
\end{tabular}
\end{center}

\mypara{Stage 1: Bob Sends $X?$ to Alice via Cold-call}

Bob globalizes $X_b?$ (a reader) for Alice.

\paragraph{Globalization at Bob for Alice:}
\begin{center}
\begin{tabular}{|l|l|}
\hline
Variable & $X_b?$ (reader) \\
\hline
Global name & $\_r(b, 1)$ \\
\hline
Entry at Bob & (none---outgoing link) \\
\hline
Spawned goal & \texttt{global\_send}$(X_b?, \_r(b,1), a)$ \\
\hline
\end{tabular}
\end{center}

\paragraph{Localization at Alice:}
\begin{center}
\begin{tabular}{|l|l|}
\hline
Global name & $\_r(b, 1)$ \\
\hline
New pair & $(X_a, X_a?)$ \\
\hline
In term & $X_a?$ (reader) \\
\hline
Entry at Alice & $(X_a, b, 1)$ \\
\hline
Spawned goal & (none---incoming link) \\
\hline
\end{tabular}
\end{center}

\begin{center}
\begin{tabular}{|c|c|c|}
\hline
\multicolumn{3}{|c|}{\textbf{madGLP State after Stage 1}} \\
\hline
\textbf{Agent $a$ (Alice)} & \textbf{Agent $b$ (Bob)} & \textbf{Agent $c$ (Charlie)} \\
\hline
\multicolumn{3}{|c|}{\textbf{Resolvents (including global\_send goals)}} \\
\hline
$X_a, X_a?$ & $X_b, X_b?$, & (empty) \\
& \texttt{global\_send}$(X_b?, \_r(b,1), a)$ & \\
\hline
\multicolumn{3}{|c|}{\textbf{Global Writers Tables}} \\
\hline
\begin{tabular}{cc}
$k$ & entry \\
\hline
1 & $(X_a, b, 1)$
\end{tabular}
&
(empty)
& (empty) \\
\hline
\end{tabular}
\end{center}

Bob globalized reader $X_b?$ as $\_r(b,1)$, spawning a \texttt{global\_send} goal at Bob (Definition~\ref{def:globalize}, case~2). Alice localized $\_r(b,1)$, creating entry $(X_a, b, 1)$ and placing reader $X_a?$ in her term (Definition~\ref{def:localize}, case~2). Alice will receive values on this link.

\mypara{Stage 2: Bob Sends $X$ to Charlie via Cold-call}

Bob globalizes $X_b$ (a writer) for Charlie.

\paragraph{Globalization at Bob for Charlie:}
\begin{center}
\begin{tabular}{|l|l|}
\hline
Variable & $X_b$ (writer) \\
\hline
Global name & $\_w(b, 2)$ \\
\hline
Entry at Bob & $(X_b, c)$ at index 2 \\
\hline
Spawned goal & (none---incoming link) \\
\hline
\end{tabular}
\end{center}

\paragraph{Localization at Charlie:}
\begin{center}
\begin{tabular}{|l|l|}
\hline
Global name & $\_w(b, 2)$ \\
\hline
New pair & $(X_c, X_c?)$ \\
\hline
In term & $X_c$ (writer) \\
\hline
Entry at Charlie & (none---outgoing link) \\
\hline
Record at Charlie & $(X_c, b, 2, b)$ \\
\hline
Spawned goal & \texttt{global\_send}$(X_c?, \_w(b,2), b)$ \\
\hline
\end{tabular}
\end{center}

\begin{center}
\begin{tabular}{|c|c|c|}
\hline
\multicolumn{3}{|c|}{\textbf{madGLP State after Stage 2}} \\
\hline
\textbf{Agent $a$ (Alice)} & \textbf{Agent $b$ (Bob)} & \textbf{Agent $c$ (Charlie)} \\
\hline
\multicolumn{3}{|c|}{\textbf{Resolvents (including global\_send goals)}} \\
\hline
$X_a, X_a?$ & $X_b, X_b?$, & $X_c, X_c?$, \\
& \texttt{global\_send}$(X_b?, \_r(b,1), a)$ & \texttt{global\_send}$(X_c?, \_w(b,2), b)$ \\
\hline
\multicolumn{3}{|c|}{\textbf{Global Writers Tables}} \\
\hline
\begin{tabular}{cc}
$k$ & entry \\
\hline
1 & $(X_a, b, 1)$
\end{tabular}
&
\begin{tabular}{cc}
$i$ & entry \\
\hline
2 & $(X_b, c)$
\end{tabular}
&
(empty)
\\
\hline
\end{tabular}
\end{center}

Bob globalized writer $X_b$ as $\_w(b,2)$, creating entry $(X_b, c)$ at index~2 (Definition~\ref{def:globalize}, case~1). Bob will receive the assignment on this link. Charlie localized $\_w(b,2)$, placing writer $X_c$ in his term, recording $(X_c, b, 2, b)$, and spawning \texttt{global\_send}$(X_c?, \_w(b,2), b)$ (Definition~\ref{def:localize}, case~1). Charlie can assign $X_c$.

\paragraph{Global Link Summary after Stage 2:}
\begin{center}
\begin{tabular}{|c|c|c|c|}
\hline
\textbf{Link} & \textbf{Entry at} & \textbf{Entry} & \textbf{global\_send at} \\
\hline
$\_r(b,1)$ & Alice & $(X_a, b, 1)$ & Bob \\
\hline
$\_w(b,2)$ & Bob (index 2) & $(X_b, c)$ & Charlie \\
\hline
\end{tabular}
\end{center}

\mypara{Stage 3: Charlie Assigns $X_c := T$}

Charlie assigns $X_c := T$ for some term $T$. This makes $X_c?$ known.

\begin{center}
\begin{tabular}{|c|c|c|}
\hline
\multicolumn{3}{|c|}{\textbf{madGLP State after Charlie's assignment}} \\
\hline
\textbf{Agent $a$ (Alice)} & \textbf{Agent $b$ (Bob)} & \textbf{Agent $c$ (Charlie)} \\
\hline
\multicolumn{3}{|c|}{\textbf{Resolvents}} \\
\hline
$X_a, X_a?$ & $X_b, X_b?$, & (goal completed) \\
& \texttt{global\_send}$(X_b?, \_r(b,1), a)$ & \\
\hline
\multicolumn{3}{|c|}{\textbf{Assignments}} \\
\hline
(none) & (none) & $X_c := T$, $X_c? := T$ \\
\hline
\multicolumn{3}{|c|}{\textbf{Pending Messages}} \\
\hline
\multicolumn{3}{|c|}{$(\_w(b,2) := T^\uparrow,\ b)$ from Charlie} \\
\hline
\end{tabular}
\end{center}

The \texttt{global\_send}$(X_c?, \_w(b,2), b)$ goal fired because $X_c?$ became known. It sent message $(\_w(b,2) := T^\uparrow, b)$.

\mypara{Stage 4: Bob Receives from Charlie}

Bob receives message $(\_w(b,2) := T^\uparrow)$. He finds entry $(X_b, c)$ at index 2 and assigns $X_b := T^\downarrow$.

This makes $X_b?$ known, which triggers the \texttt{global\_send}$(X_b?, \_r(b,1), a)$ goal.

\begin{center}
\begin{tabular}{|c|c|c|}
\hline
\multicolumn{3}{|c|}{\textbf{madGLP State after Bob receives}} \\
\hline
\textbf{Agent $a$ (Alice)} & \textbf{Agent $b$ (Bob)} & \textbf{Agent $c$ (Charlie)} \\
\hline
\multicolumn{3}{|c|}{\textbf{Resolvents}} \\
\hline
$X_a, X_a?$ & (goal completed) & (empty) \\
\hline
\multicolumn{3}{|c|}{\textbf{Assignments}} \\
\hline
(none) & $X_b := T$, $X_b? := T$ & $X_c := T$, $X_c? := T$ \\
\hline
\multicolumn{3}{|c|}{\textbf{Global Writers Tables}} \\
\hline
\begin{tabular}{cc}
$k$ & entry \\
\hline
1 & $(X_a, b, 1)$
\end{tabular}
&
(empty)
& (empty) \\
\hline
\multicolumn{3}{|c|}{\textbf{Pending Messages}} \\
\hline
\multicolumn{3}{|c|}{$(\_r(b,1) := T^\uparrow,\ a)$ from Bob} \\
\hline
\end{tabular}
\end{center}

Entry at index 2 at Bob was removed after the assignment; a writer-name value is not acknowledged---its anchor entry is fixed. The \texttt{global\_send} goal fired and sent message $(\_r(b,1) := T^\uparrow, a)$.

\mypara{Stage 5: Alice Receives from Bob}

Alice receives message $(\_r(b,1) := T^\uparrow)$. She searches for entry with remote agent $b$ and remote index $1$, finds $(X_a, b, 1)$, and assigns $X_a := T^\downarrow$.

\begin{center}
\begin{tabular}{|c|c|c|}
\hline
\multicolumn{3}{|c|}{\textbf{Final madGLP State}} \\
\hline
\textbf{Agent $a$ (Alice)} & \textbf{Agent $b$ (Bob)} & \textbf{Agent $c$ (Charlie)} \\
\hline
\multicolumn{3}{|c|}{\textbf{Assignments}} \\
\hline
$X_a := T$, $X_a? := T$ & $X_b := T$, $X_b? := T$ & $X_c := T$, $X_c? := T$ \\
\hline
\multicolumn{3}{|c|}{\textbf{Global Writers Tables}} \\
\hline
(empty) & (empty) & (empty) \\
\hline
\end{tabular}
\end{center}

All entries have been removed. The value $T$ has flowed from Charlie through Bob to Alice; Alice acknowledges $\_r(b,1)$, releasing Bob's pending value and closing the link.

\paragraph{Summary:} Bob acted as a relay. Charlie assigned his writer $X_c$, the value flowed to Bob via the $\_w(b,2)$ link, Bob's local pair $(X_b, X_b?)$ received it, and the \texttt{global\_send} goal watching $X_b?$ forwarded it to Alice via the $\_r(b,1)$ link. Alice's reader $X_a?$ now has the value $T$.

\paragraph{The stream continues directly:} Suppose $T = [\mathtt{hello}|\mathsf{Rest?}]$ with $\mathsf{Rest}$ a fresh pair at Charlie. Charlie's value to Bob carries $\_r(c,1)$; Bob localizes it with entry $(\mathsf{Rest_b}, c, 1)$ and no request---the sender is the anchor. Forwarding the value to Alice, Bob globalizes the imported reader $\mathsf{Rest_b?}$: the original name $\_r(c,1)$ is emitted and Bob's entry removed (Definition~\ref{def:globalize}). Alice localizes $\_r(c,1)$ from sender Bob $\ne c$, so she sends $\texttt{req}(\_r(c,1))$ to Charlie, who records her as the holder. The next element flows Charlie$\to$Alice directly: the introducer relays one value and drops out.

\subsection{Example 3: Both Ends to Same Agent}
\label{app:both-ends-same}

This example demonstrates the edge case where Bob sends both writer $X$ and reader $X?$ to Alice in the same term $[X, X?]$. Bob's local pair acts as a relay: when Alice assigns her writer $Y_a$ (from the writer global name $\_w(b,1)$), the value flows through Bob back to Alice's reader $Z_a?$ (from the reader global name $\_r(b,2)$).

\mypara{Stage 0: Initial State}

Bob has a local variable pair $(X_b, X_b?)$ and will send term $[X_b, X_b?]$ to Alice.

\begin{center}
\begin{tabular}{|c|c|}
\hline
\multicolumn{2}{|c|}{\textbf{madGLP State}} \\
\hline
\textbf{Agent $a$ (Alice)} & \textbf{Agent $b$ (Bob)} \\
\hline
\multicolumn{2}{|c|}{\textbf{Resolvents}} \\
\hline
(empty) & $X_b, X_b?$ \\
\hline
\multicolumn{2}{|c|}{\textbf{Global Writers Tables}} \\
\hline
(empty) & (empty) \\
\hline
\end{tabular}
\end{center}

\mypara{Stage 1: Bob Sends $[X_b, X_b?]$ to Alice via Cold-call}

Bob globalizes term $[X_b, X_b?]$ for Alice. Both the writer and reader are processed:

\paragraph{Globalization at Bob:}
\begin{center}
\begin{tabular}{|l|l|l|}
\hline
\textbf{Variable} & \textbf{$X_b$ (writer)} & \textbf{$X_b?$ (reader)} \\
\hline
Global name & $\_w(b, 1)$ & $\_r(b, 2)$ \\
\hline
Entry at Bob & $(X_b, a)$ at index 1 & (none---outgoing) \\
\hline
Spawned goal & (none---incoming) & \texttt{global\_send}$(X_b?, \_r(b,2), a)$ \\
\hline
\end{tabular}
\end{center}

\paragraph{Globalized term:} $[\_w(b,1), \_r(b,2)]$

\paragraph{Localization at Alice:}
\begin{center}
\begin{tabular}{|l|l|l|}
\hline
\textbf{Global name} & \textbf{$\_w(b, 1)$} & \textbf{$\_r(b, 2)$} \\
\hline
New pair & $(Y_a, Y_a?)$ & $(Z_a, Z_a?)$ \\
\hline
In term & $Y_a$ (writer) & $Z_a?$ (reader) \\
\hline
Entry at Alice & (none---outgoing) & $(Z_a, b, 2)$ \\
\hline
Record at Alice & $(Y_a, b, 1, b)$ & (none) \\
\hline
Spawned goal & \texttt{global\_send}$(Y_a?, \_w(b,1), b)$ & (none---incoming) \\
\hline
\end{tabular}
\end{center}

\paragraph{Localized term:} $[Y_a, Z_a?]$

\begin{center}
\begin{tabular}{|c|c|}
\hline
\multicolumn{2}{|c|}{\textbf{madGLP State after Stage 1}} \\
\hline
\textbf{Agent $a$ (Alice)} & \textbf{Agent $b$ (Bob)} \\
\hline
\multicolumn{2}{|c|}{\textbf{Resolvents (including global\_send goals)}} \\
\hline
$Y_a, Y_a?, Z_a, Z_a?$, & $X_b, X_b?$, \\
\texttt{global\_send}$(Y_a?, \_w(b,1), b)$ & \texttt{global\_send}$(X_b?, \_r(b,2), a)$ \\
\hline
\multicolumn{2}{|c|}{\textbf{Global Writers Tables}} \\
\hline
\begin{tabular}{cc}
$k$ & entry \\
\hline
1 & $(Z_a, b, 2)$
\end{tabular}
&
\begin{tabular}{cc}
$i$ & entry \\
\hline
1 & $(X_b, a)$
\end{tabular}
\\
\hline
\end{tabular}
\end{center}

Alice's term is $[Y_a, Z_a?]$: she has writer $Y_a$ (can assign values) and reader $Z_a?$ (will receive values). These are two \emph{independent} pairs at Alice---she doesn't know they're connected through Bob.

Writer $X_b$ was globalized as $\_w(b,1)$ with entry $(X_b, a)$ at Bob (Definition~\ref{def:globalize}, case~1); Alice localized it as writer $Y_a$ with the record $(Y_a, b, 1, b)$ and a \texttt{global\_send} goal (Definition~\ref{def:localize}, case~1). Reader $X_b?$ was globalized as $\_r(b,2)$ with \texttt{global\_send} at Bob (Definition~\ref{def:globalize}, case~2); Alice localized it as reader $Z_a?$ with entry $(Z_a, b, 2)$ (Definition~\ref{def:localize}, case~2).

\mypara{Stage 2: Alice Assigns $Y_a := T$}

Alice assigns $Y_a := T$. This makes $Y_a?$ known, triggering her \texttt{global\_send} goal.

\begin{center}
\begin{tabular}{|c|c|}
\hline
\multicolumn{2}{|c|}{\textbf{madGLP State after Alice's assignment}} \\
\hline
\textbf{Agent $a$ (Alice)} & \textbf{Agent $b$ (Bob)} \\
\hline
\multicolumn{2}{|c|}{\textbf{Resolvents}} \\
\hline
$Z_a, Z_a?$ & $X_b, X_b?$, \\
(goal completed) & \texttt{global\_send}$(X_b?, \_r(b,2), a)$ \\
\hline
\multicolumn{2}{|c|}{\textbf{Assignments}} \\
\hline
$Y_a := T$, $Y_a? := T$ & (none) \\
\hline
\multicolumn{2}{|c|}{\textbf{Global Writers Tables}} \\
\hline
\begin{tabular}{cc}
$k$ & entry \\
\hline
1 & $(Z_a, b, 2)$
\end{tabular}
&
\begin{tabular}{cc}
$i$ & entry \\
\hline
1 & $(X_b, a)$
\end{tabular}
\\
\hline
\multicolumn{2}{|c|}{\textbf{Pending Messages}} \\
\hline
\multicolumn{2}{|c|}{$(\_w(b,1) := T^\uparrow,\ b)$ from Alice} \\
\hline
\end{tabular}
\end{center}

The \texttt{global\_send} goal sent message $(\_w(b,1) := T^\uparrow, b)$.

\mypara{Stage 3: Bob Receives from Alice}

Bob receives message $(\_w(b,1) := T^\uparrow)$. He finds entry $(X_b, a)$ at index 1 and assigns $X_b := T^\downarrow$.

This makes $X_b?$ known, triggering his \texttt{global\_send}$(X_b?, \_r(b,2), a)$ goal.

\begin{center}
\begin{tabular}{|c|c|}
\hline
\multicolumn{2}{|c|}{\textbf{madGLP State after Bob receives}} \\
\hline
\textbf{Agent $a$ (Alice)} & \textbf{Agent $b$ (Bob)} \\
\hline
\multicolumn{2}{|c|}{\textbf{Resolvents}} \\
\hline
$Z_a, Z_a?$ & (goal completed) \\
\hline
\multicolumn{2}{|c|}{\textbf{Assignments}} \\
\hline
$Y_a := T$, $Y_a? := T$ & $X_b := T$, $X_b? := T$ \\
\hline
\multicolumn{2}{|c|}{\textbf{Global Writers Tables}} \\
\hline
\begin{tabular}{cc}
$k$ & entry \\
\hline
1 & $(Z_a, b, 2)$
\end{tabular}
&
(empty)
\\
\hline
\multicolumn{2}{|c|}{\textbf{Pending Messages}} \\
\hline
\multicolumn{2}{|c|}{$(\_r(b,2) := T^\uparrow,\ a)$ from Bob} \\
\hline
\end{tabular}
\end{center}

Entry at index 1 at Bob was removed after the assignment. The \texttt{global\_send} goal fired and sent message $(\_r(b,2) := T^\uparrow, a)$.

\mypara{Stage 4: Alice Receives from Bob}

Alice receives message $(\_r(b,2) := T^\uparrow)$. She searches for entry with remote agent $b$ and remote index $2$, finds $(Z_a, b, 2)$, and assigns $Z_a := T^\downarrow$.

\begin{center}
\begin{tabular}{|c|c|}
\hline
\multicolumn{2}{|c|}{\textbf{Final madGLP State}} \\
\hline
\textbf{Agent $a$ (Alice)} & \textbf{Agent $b$ (Bob)} \\
\hline
\multicolumn{2}{|c|}{\textbf{Assignments}} \\
\hline
$Y_a := T$, $Y_a? := T$ & $X_b := T$, $X_b? := T$ \\
$Z_a := T$, $Z_a? := T$ & \\
\hline
\multicolumn{2}{|c|}{\textbf{Global Writers Tables}} \\
\hline
(empty) & (empty) \\
\hline
\end{tabular}
\end{center}

All entries have been removed. Bob's value to Alice, carried by the reader name $\_r(b,2)$, is acknowledged by Alice; Alice's value to Bob, carried by the writer name $\_w(b,1)$, needs no acknowledgement.

\paragraph{Summary:} Alice received term $[Y_a, Z_a?]$ with two independent pairs. When she assigned $Y_a := T$, the value flowed to Bob via the $\_w(b,1)$ link. Bob's local pair received it, and his \texttt{global\_send} goal forwarded it back to Alice via the $\_r(b,2)$ link. Now both $Y_a?$ and $Z_a?$ equal $T$---the two pairs at Alice are semantically linked through Bob's relay, even though Alice sees them as independent.

This correctly implements the maGLP semantics: in maGLP, Alice would have received $[X, X?]$ where both variables refer to the same shared pair. Assigning the writer makes the value available to the reader. In madGLP, this is achieved through Bob's relay.

\section{Code Format}
\label{app:code-format}

This appendix is the byte-level specification of the GLP code format: the canonical encoding of globalized terms and assignment messages, the encoding of the instruction set, the compiled-program artefact, deterministic flattening, the loader, and the offer and handshake messages. The engine implements the sequential abstract machine of Flat Concurrent Prolog~\cite{houri1989sequential}, whose reference implementation is available online.\footnote{\url{https://github.com/EShapiro2/FCP}} The format and engine follow that design, departing only where GLP differs from FCP---the instruction set below and the fixed byte order of \S\ref{app:cf-primitives}; departures are noted where they occur. The source and compiled identities, the module certificate, the conversation handshake, and the signature kernels are specified in the Secure GLP paper~\cite{keidar2026secure}; programs, exports, and the load-time interface check in the type and module system paper~\cite{shapiro2026types}; the instruction semantics in the companion GLP paper~\cite{shapiro2025glp}, this appendix assigning their bytes.

\subsection{Primitive Encodings}
\label{app:cf-primitives}

Fixed-width multi-byte integers (i64, f64) are little-endian, the one canonical byte order; clen is self-delimiting and read identically on every machine. A term yields the same byte string on every machine: the bytes that are hashed and signed.
\begin{itemize}
\item \temph{u8} --- one byte.
\item \temph{i64} --- 8 bytes, two's complement.
\item \temph{f64} --- 8 bytes, IEEE~754 binary64.
\item \temph{clen} (compact length) --- an unsigned integer below $2^{30}$ in 1, 2, or 4 bytes: values below 128 in one byte; below 16384 in two bytes, the first with high bits \texttt{10}; otherwise four bytes, the first with high bits \texttt{11}. The remaining bits, big-endian, carry the value. Encoders emit the shortest form; decoders reject longer-than-necessary forms (one value, one encoding).
\item \temph{string} --- clen byte count, then UTF-8 bytes.
\item \temph{bytes} --- clen byte count, then raw bytes.
\item \temph{hash} --- 32 raw bytes: SHA-256 of the hashed content. The \temph{source identity} $h(M)$ is the hash of the flattened source bytes (\S\ref{app:cf-flattening}); the \temph{compiled identity} is the hash of the artefact body (\S\ref{app:cf-artefact}).
\item \temph{agent} --- bytes; its content is the public key of the agent's person. The current realisation uses Ed25519 public keys (32 bytes); the simulation realisation uses UTF-8 symbolic names. The encoding does not interpret the content.
\end{itemize}

\subsection{Term and Message Encoding}
\label{app:cf-terms}

This section is the byte-level layout of the canonical encoding $e$: an injective function from globalized terms and assignment messages to byte strings.

\subsubsection{Terms.}
A globalized term is encoded by a tagged recursion. Each node opens with a u8 tag:
\begin{itemize}
\item \temph{1 constant} --- followed by a u8 constant tag and its payload: \temph{0 nil} (no payload, the empty list); \temph{1 integer} (i64); \temph{2 float} (f64); \temph{3 string} (string); \temph{4 boolean} (u8: 0 false, 1 true); \temph{5 bytes} (bytes; raw byte strings, used only in the runtime's own vocabulary---the 32-byte hashes of \S\ref{app:cf-handshake}---never in application terms); \temph{6 module} (bytes; the artefact bytes of \S\ref{app:cf-artefact}, decoding to the Module constant---the form in which compiled programs ship).
\item \temph{2 variable} --- a global name per Definition~\ref{def:globalize}: u8 polarity (0 writer $\_w(p,i)$, 1 reader $\_r(p,i)$), agent $p$, clen index $i$. Names are the allocating agent's---a forwarded reader keeps its anchor's name; no other variable representation exists in the code format.
\item \temph{3 structure} --- string functor, clen arity $n$, then the $n$ argument encodings in order.
\end{itemize}
Lists are structures: a cell is the structure \texttt{'.'/2}; the empty list is the constant nil. The four properties of the canonical encoding hold by this layout: hardware-independent (\S\ref{app:cf-primitives} pins widths and byte order); address-free (a variable appears only as its global name with polarity); globalizer's names (tag~2 admits nothing else); canonical on ground terms (a ground term contains tags~1 and~3 only, so its bytes are agent-independent --- these are the bytes over which terms are signed and hashed). Injectivity holds by unique decoding: every tag determines its payload form, every length is explicit, and clen is one-value-one-encoding.

\subsubsection{Messages.}
A message opens with a u8 kind: \temph{0 value}, \temph{1 request}, \temph{2 acknowledgement}. A value $G := T^\uparrow$ continues: u8 polarity of $G$, agent of $G$, clen index of $G$, then the encoding of $T^\uparrow$. A request $\texttt{req}(G)$ or an acknowledgement $\texttt{ack}(G)$ continues: u8 polarity (1, a reader name), agent of $G$, clen index of $G$; no term follows. The kind byte is new in code-format version~2 (\S\ref{app:cf-versioning}).

A serializer (cold-call) message to agent $q$ is the assignment $\_w(q,0) := [T^\uparrow \mid \_w(q,0)]$: polarity 0, agent $q$, index 0, then the encoding of the list cell whose head is $T^\uparrow$ and whose tail is the variable $\_w(q,0)$ (tag~2, polarity 0, agent $q$, index 0). The receiver treats index~0 by serializer semantics: append, update the entry, do not remove it.

\subsubsection{Signed and Hashed Content.}
The canonical bytes of a ground term $T$ are $e(T)$. Protocols that sign or hash compound content encode it as a ground term and apply $e$; the signature kernels of Secure GLP~\cite{keidar2026secure} sign $e$ of the pair of the signing instance's source identity and the term, and the signed term carries the signer's key alongside that pair, so that a verifier recovers all three (\S\ref{app:cf-handshake}).

\subsection{Instruction Encoding}
\label{app:cf-instructions}

This section assigns bytes to the instruction set; the instruction semantics are those of the companion GLP paper~\cite{shapiro2025glp}. The scheduling of \texttt{spawn} (0x50) and \texttt{requeue} (0x51) is dGLP's, not the companion paper's: \texttt{spawn} enqueues a body goal, \texttt{requeue} continues the clause's tail goal under the $b$-bounded tail discipline of Definition~\ref{def:bounded-tail}. An encoded instruction is a u8 opcode followed by its operands, in the order listed.

\subsubsection{Operand kinds.}
\begin{itemize}
\item \temph{polarity} --- u8: 0 writer, 1 reader (the \texttt{isReader} flag; same convention as the variable tag).
\item \temph{negated} --- u8: 0 plain, 1 negated guard ($\sim\!G$).
\item \temph{varIndex, argSlot, arity, count, slots, regIndex} --- clen.
\item \temph{constant} --- the constant payload of \S\ref{app:cf-terms} (u8 constant tag, then its payload); one constant representation serves terms and code.
\item \temph{functor} --- string.
\item \temph{proc} --- clen index into the artefact's symbol table (\S\ref{app:cf-artefact}). The entry's kind determines the target: a compiled procedure, or a codeless body kernel or builtin guard bound by name at load. Procedure references are symbol-table indices, never names; clause targets are byte offsets (\temph{ctarget}), not instruction indices.
\item \temph{ctarget} --- a clen byte offset to a clause target within the current procedure's code (clause targets are procedure-relative); the engine branches by program-counter offset, as in FCP~\cite{houri1989sequential}. Assembly labels do not exist in the code format.
\end{itemize}

\subsubsection{Opcode table.}
\begin{center}
\small
\begin{longtable}{lll}
\toprule
Opcode & Mnemonic & Operands \\
\midrule
\endfirsthead
\toprule
Opcode & Mnemonic & Operands \\
\midrule
\endhead
\texttt{0x01} & \texttt{clause\_try} & --- \\
\texttt{0x02} & \texttt{clause\_next} & ctarget \\
\texttt{0x03} & \texttt{no\_more\_clauses} & --- \\
\texttt{0x04} & \texttt{commit} & --- \\
\texttt{0x05} & \texttt{proceed} & --- \\
\texttt{0x06} & \texttt{halt} & --- \\
\texttt{0x07} & \texttt{nop} & --- \\
\texttt{0x10} & \texttt{head\_constant} & constant, argSlot \\
\texttt{0x11} & \texttt{head\_nil} & argSlot \\
\texttt{0x12} & \texttt{head\_structure} & functor, arity, argSlot \\
\texttt{0x13} & \texttt{head\_list} & argSlot \\
\texttt{0x14} & \texttt{head\_variable} & polarity, varIndex \\
\texttt{0x15} & \texttt{get\_variable} & polarity, varIndex, argSlot \\
\texttt{0x16} & \texttt{get\_value} & polarity, varIndex, argSlot \\
\texttt{0x20} & \texttt{unify\_variable} & polarity, varIndex \\
\texttt{0x21} & \texttt{unify\_constant} & constant \\
\texttt{0x22} & \texttt{unify\_void} & count \\
\texttt{0x23} & \texttt{unify\_structure} & functor, arity \\
\texttt{0x24} & \texttt{push} & regIndex \\
\texttt{0x25} & \texttt{pop} & regIndex \\
\texttt{0x30} & \texttt{put\_variable} & polarity, varIndex, argSlot \\
\texttt{0x31} & \texttt{put\_constant} & constant, argSlot \\
\texttt{0x32} & \texttt{put\_nil} & argSlot \\
\texttt{0x33} & \texttt{put\_list} & argSlot \\
\texttt{0x34} & \texttt{put\_structure} & functor, arity, argSlot \\
\texttt{0x35} & \texttt{set\_variable} & polarity, varIndex \\
\texttt{0x36} & \texttt{set\_constant} & constant \\
\texttt{0x37} & \texttt{allocate} & slots \\
\texttt{0x38} & \texttt{deallocate} & --- \\
\texttt{0x39} & \texttt{put\_bound\_const} & constant, argSlot \\
\texttt{0x3A} & \texttt{put\_bound\_nil} & argSlot \\
\texttt{0x40} & \texttt{guard} & proc, arity, negated \\
\texttt{0x41} & \texttt{ground} & varIndex, negated \\
\texttt{0x42} & \texttt{known} & varIndex, negated \\
\texttt{0x43} & \texttt{unknown} & varIndex \\
\texttt{0x44} & \texttt{no\_readers} & varIndex, negated \\
\texttt{0x45} & \texttt{ground\_equal} & varIndex, varIndex, negated \\
\texttt{0x46} & \texttt{otherwise} & --- \\
\texttt{0x50} & \texttt{spawn} & proc, arity \\
\texttt{0x51} & \texttt{requeue} & proc, arity \\
\bottomrule
\end{longtable}
\end{center}

\subsubsection{Reserved and excluded.}
Opcodes \texttt{0x47}--\texttt{0x4F} are reserved for the comparison guards (\texttt{guard\_less} and its five siblings), assigned when they are implemented; their arrival is an instruction-set version change carried in the artefact header (\S\ref{app:cf-artefact}). Opcodes \texttt{0x52}--\texttt{0x53} are reserved for runtime inter-unit dispatch: under static linking (\S\ref{app:cf-flattening}) every call is local and no dispatch instruction is emitted, so the code format carries none; dispatch of dynamically loaded modules returns with the trust machinery of Secure GLP~\cite{keidar2026secure}, as an instruction-set version change. The following implementation classes are not part of the code format: Label (assembly-time, erased by the operand kinds); TryNextClause (unused); GuardFail, SuspendEnd, TailStep, BodySetConst, BodySetStructConstArgs, HeadBindWriter, GuardNeedReader, and the \texttt{*Arg} test variants (legacy).

\subsection{Program Artefact}
\label{app:cf-artefact}

The artefact is the byte string a compiler produces from a flat program; it travels as a module constant (\S\ref{app:cf-terms}, constant tag~6). It is a \temph{body} followed by a \temph{certificate}. Its \temph{compiled identity} is the SHA-256 of the body; its \temph{source identity} $h(M)$ is the hash of the flattened source (\S\ref{app:cf-flattening}). The two identities, the certificate, and what a receiver concludes from it are specified in Secure GLP~\cite{keidar2026secure}; a compiled program carries what the load-time interface check needs~\cite{shapiro2026types}. The body's sections, in order:
\begin{enumerate}
\item \temph{Header} --- magic \texttt{GLPW} (4 bytes); u8 code-format version (this document: 2); string instruction-set version; string program name. Neither identity is in the header: both are in the certificate, where they are signed.
\item \temph{Interface table} --- the program interface, carried as declaration source text, from which the loader derives the type automata: a string giving the reachable type definitions, in canonical print (\S\ref{app:cf-flattening}); then a clen export count, and per export a string name, clen arity, and string declaration text. Carrying text rather than compiled automata keeps one source of truth: type identity is source-based throughout, and equal sources define equal automata.
\item \temph{Symbol table} --- clen count; per entry a string name, clen arity, and a u8 kind. Kind \temph{0 compiled} --- a procedure defined in this program --- is followed by clen byte offset into the code section and clen byte length; kind \temph{1 codeless} --- a body kernel or builtin guard with no code in this artefact, bound by name to the local runtime at load (\S\ref{app:cf-loader}) --- carries nothing further. The \temph{proc} operands of \S\ref{app:cf-instructions} index this table. Exported procedures are the compiled entries named in the interface table; the loader aliases exactly them.
\item \temph{Code section} --- clen byte count, then the concatenated procedure bodies in the instruction encoding of \S\ref{app:cf-instructions}. There is no constant pool: constants are inline in instructions.
\end{enumerate}
The \temph{certificate} follows the body: \temph{agent}, the public key of the person who compiled the module; then \temph{bytes}, that person's signature over $e(\texttt{ids}(\mathit{HSrc}, \mathit{HBin}))$, the canonical encoding of the 2-ary structure \texttt{ids} whose arguments are the source identity and the compiled identity, both hash constants. The functor \texttt{ids} is fixed by this specification. It is written at compilation, whether or not the module is ever sent, so a module states which source it was compiled from and who is answerable for it~\cite{keidar2026secure}. The compiled identity is the hash of the body alone, so the certificate sits inside the artefact without hashing itself, and \texttt{decompose\_module/4} reads the key and both identities from it.

The body is code only. Source distribution is an ordinary value exchange: the flattened source travels as a blob in a term, and the receiver verifies it by hashing against $h(M)$.

\subsection{Deterministic Flattening}
\label{app:cf-flattening}

The \temph{flattened source} of a program --- the preimage of the source identity $h(M)$ --- is the canonical print of the linked, pruned program: the program is discovered, type-checked, renamed, resolved, and pruned to the procedures reachable from the root's exported entry points~\cite{shapiro2026types}; the result is printed canonically and hashed~\cite{keidar2026secure}.

Canonical print: UTF-8, LF line ends, no comments, tokens separated by single spaces, one declaration or clause per line. Order: (\ia)~the reachable type definitions, lexicographically by type name; (\ib)~the exported procedure declarations, lexicographically by name then arity; (\ic)~the procedures, lexicographically by renamed name then arity --- within a procedure, clauses keep their source order, which is semantic (first-applicable-clause selection). Procedure order is not semantic and is fixed by sorting.

The same byte string is what a digital social contract's participants read and agree to, what the handshake compares by hash, and what versioning names.

\subsection{Loader}
\label{app:cf-loader}

Given a received artefact and the adoption context (the offered $h(M)$, and the sender's attestation)~\cite{keidar2026secure}, the loader:
\begin{enumerate}
\item Computes SHA-256 of the body and verifies it equals the compiled identity in the certificate; verifies the certificate's signature under the key the certificate carries; verifies the certificate's source identity equals the offered $h(M)$; refuses unsupported code-format or instruction-set versions (\S\ref{app:cf-versioning}).
\item Derives the type automata from the interface table's declaration text and runs the load-time interface checks when linking the program against local callers: the callee's exported type must accept the caller's imported type~\cite{shapiro2026types}.
\item Resolves \temph{proc} indices through the symbol table per \S\ref{app:cf-instructions}: a compiled entry resolves to the procedure at its offset; a codeless entry binds by name to the local runtime's kernel or builtin guard of that name, the instruction (\texttt{spawn}/\texttt{requeue} versus \texttt{guard}) selecting which, and an unknown name is a failsafe refusal (\S\ref{app:cf-versioning}). Generates unqualified aliases for the exported procedures only.
\item Registers the program under its two identities; artefacts are cached and deduplicated by compiled identity.
\end{enumerate}
Admission of conversation traffic is not the loader's: the runtime admits application traffic on a conversation only after its handshake succeeds, and aborts failsafe otherwise~\cite{keidar2026secure}.

After verification, if the running machine's byte order differs from the canonical order the loader inverts the artefact's fixed-width multi-byte fields to local order for execution; before the program is shipped onward it is inverted back, so the hashed artefact is always the canonical bytes (\S\ref{app:cf-primitives}).

\subsection{Offer and Handshake Messages}
\label{app:cf-handshake}

These are ground terms between runtimes on the attested channel, encoded per \S\ref{app:cf-terms}; hash values inside terms are bytes constants of 32 bytes. They are specified in Secure GLP~\cite{keidar2026secure}.
\begin{itemize}
\item \texttt{offer(HSrc, tau(HSrcDef, TypeName))} --- the adoption offer: the contract's source identity, and the root channel's type identity --- the pair of the source identity of the type's defining source and the type's name (a string).
\item \texttt{accept} / \texttt{decline} --- the consent outcome of the offeree's person.
\item \texttt{ship(Artefact)} --- the compiled program as a module constant. The two identities and the key answerable for them are in the artefact's own certificate (\S\ref{app:cf-artefact}), so the message no longer carries them.
\item \texttt{handshake(HSrc, tau(HSrcDef, TypeName))} --- exchanged by the two runtimes on a conversation before any application traffic; both ends must present equal values, or the runtime aborts the conversation failsafe. Derived variables undergo no handshake.
\item \temph{Signed content} --- \texttt{sign(T, K, S)} assigns $S$ the signed term of the ground term $T$ under the key $K$~\cite{keidar2026secure}. $S$ is \temph{agent}, the signer's public key; then \temph{bytes}, the signature; then $e(\texttt{sig}(\mathit{HSrc}, T))$, the canonical encoding of the 2-ary structure \texttt{sig} whose arguments are the source identity of the signing instance's module and the ground term. The signature is over those encoded bytes. The functor \texttt{sig} is fixed by this specification. \texttt{signed(S, K, H, T)} decodes $S$ and assigns the signer's key, the source identity and the term; \texttt{signed(S, K)} assigns the signer alone. The runtime supplies $\mathit{HSrc}$ and the caller cannot choose it, so neither verifying predicate runs the signing module.
\end{itemize}

\subsection{Format Versioning}
\label{app:cf-versioning}

The header carries two versions: the code-format version (u8) --- this document, version 2 --- and the instruction-set version (string). A loader refuses an artefact whose code-format version or instruction-set version it does not support; refusal at adoption time means the offer fails before any conversation exists. New instructions enter by instruction-set version (reserved ranges, \S\ref{app:cf-instructions}); changes to the encodings of \S\ref{app:cf-primitives}--\S\ref{app:cf-terms} or the artefact layout of \S\ref{app:cf-artefact} are code-format version changes. Opcode and name assignments are append-only: an assigned opcode never changes meaning, and a kernel or guard name, once in use, is never reassigned to a different operation; additions consume only reserved or unused opcodes and new names, so a newer runtime runs older artefacts unchanged. Forward incompatibility is always a clean refusal, never silent misexecution: an older runtime rejects a newer instruction-set version at adoption, and a codeless name the runtime does not provide fails to bind at load (\S\ref{app:cf-loader}). Module versioning --- a version is $h(M)$, a new version is a new contract --- is specified in Secure GLP~\cite{keidar2026secure} and is not restated here.

\section{Implementation Notes}
\label{app:implementation-notes}

This appendix records implementation decisions of the Dart runtime derived from dGLP and madGLP: the heap representation of variable pairs and suspensions, agent execution and boot, and the networking interface of the multiagent runtime. The code lives in the public repository; this appendix states the invariants it maintains.

\subsection{Heap: Variables, Dereferencing, Binding, Suspension}
\label{app:in-heap}

\mypara{Variable pairs}
The engine follows the FCP sequential abstract machine~\cite{houri1989sequential} in its heap design. A variable pair is two heap cells, writer and reader, each a tagged reference; while the pair is unbound the two cells point to each other, so either end reaches its counterpart by following a pointer --- no address arithmetic relates them. A variable occurrence in a goal or term is a cell address; the cell's tag gives its polarity.

\mypara{Dereferencing}
Dereferencing follows references until reaching a cell that is not a reference: a value, or an unbound writer (one whose pointer leads to its own reader). The starting cell is then updated to point directly to the final target; this path compression is integral to the design, not optional. A chain never runs writer to writer: a writer binds to a value or to a reader, never to a writer, and dereferencing checks this SRSW invariant.

\mypara{Binding}
A writer is bound at most once, to a value or to a reader. Binding to a value converts the writer cell to a value cell; binding to a reader stores a pointer to it, extending the dereference chain.

\mypara{Suspension}
A goal that suspends on a reader is registered on the paired writer's cell; an unbound writer carries its suspension list alongside its reader pointer. Binding the writer to a value activates the registered goals. Binding the writer to a reader routes its suspensions to wherever the reader's dereference chain leads: to the unbound writer at the end of the chain; and, when the chain already ends in a value, the binding determines the writer's value, so the suspensions are activated at once. The invariant: a suspension is activated exactly when its variable becomes determined, wherever along the chain that happens.

\mypara{Imported readers}
An imported reader --- one whose writer is at another agent --- is the reader of an ordinary local pair created by Localize (Definition~\ref{def:localize}), and its paired local writer is the one named by a global writers table entry. That writer carries the suspension queue and receives the incoming assignment (Definition~\ref{def:madglp-receive}), so dereferencing an imported reader reaches an unbound local writer like any other. The entry is found from the arriving name, by its remote agent and index.

\mypara{Imported writers}
An imported writer --- one localized from a writer global name --- is an ordinary local writer carrying its imported-writer record (Definition~\ref{def:imported-writer-records}), keyed by the writer's address and so reached from its cell in constant time: the anchor and index of the link its value must reach, the agent the name came from, and the \texttt{global\_send} goal watching its reader. The record is what makes re-export a forwarding: Globalize reads it from the writer's cell and removes both it and the goal, so the exporting agent leaves the link instead of relaying its values (Definition~\ref{def:globalize}).

\subsection{Agent Execution and Boot}
\label{app:in-execution}

\mypara{Clause try}
Head matching against a clause is two-phase: a collection pass traverses the arguments left to right, accumulating a tentative writers substitution and a preliminary set of blocking readers; a resolution pass removes from that set every reader whose paired writer the tentative substitution binds. An empty resolved set selects the clause; otherwise the resolved set is added to the goal's suspension set, the tentative substitution is discarded, and the next clause is tried --- when no clause selects, the goal suspends on the accumulated set if it is nonempty and fails otherwise (Definition~\ref{def:dglp-ts}). The tentative substitution reaches the heap only at commit, after the guards hold; a failed or suspended clause leaves no trace.

\mypara{Event-driven execution}
Each agent runs in its own isolate and is event-driven: an event is the initial start, an incoming assignment message (Definition~\ref{def:madglp-receive}), the person's input on the UI channel, or a place event from the networking layer (\S\ref{app:seam-predicates}). On each event the agent binds what the event delivers, reduces until quiescent --- FIFO with the $b$-bounded tail schedule of Definition~\ref{def:bounded-tail} --- and then performs its Sends (Definition~\ref{def:madglp-send}); there is no clock and no polling. An exhausted tail budget (Definition~\ref{def:bounded-tail}) yields to the host event loop, so timers, I/O, and the UI stay live. Every event that may unblock a goal is followed by such a cycle, and a suspended goal is re-enqueued solely through the activations its binding returns, so no goal is enqueued twice. Agents do not detect or report termination; shutdown is the embedding application's.

\mypara{Boot}
A program's \texttt{boot} procedure declares its agents: each body goal $G$\texttt{@}$p$ directs the runtime to spawn an isolate for agent $p$, create its index-0 serializer entry (Definition~\ref{def:index-0-serializer}), and spawn $G$ there with the network-input reader as its last argument. The runtime provides only that reader; every other channel, including the UI channel, is created by the GLP boot goal itself. Inter-agent message routing is the runtime's, by the destination in the message's global name. The \texttt{boot} clause is stripped on the boot path, so there the \texttt{@} notation never reaches the type checker and what remains of the file is an ordinary module.  A directory loaded as a program is not on that path and its boot clause is not stripped, so a program directory carries none. A boot program is placed under the program root, which is therefore its ancestor; placed outside it, the program inherits neither the root's declarations nor the modules the root exposes.

\subsection{Networking Seam}
\label{app:in-networking}

\mypara{One interface}
The multiagent runtime reaches the network through one interface, whose contract is stated below. The BLE/IP transport implementing it is a dependency; a simulation realisation of the same interface backs the tests and plays, so transport integration is a backend swap.

\mypara{The contract}
The runtime relies on this contract: an agent's identity is its public key, held with its private key by the networking layer; communication is point-to-point delivery of opaque payload bytes, the sender authenticated by the layer; delivery is fair --- every message sent is eventually delivered (Remark~\ref{rem:fair-message-delivery}) --- and unordered; the layer reports peer discovery, connection, disconnection, and reachability; and trust is open or closed, closed admitting no first contact from an unknown agent. The layer also observes a peer's address, opens a UDP path to an address, declares a place of a given radius, removes a declared place, and reports entry to and exit from a declared place together with when it stops and resumes reporting them; these five back the seam predicates (\S\ref{app:seam-predicates}).

\mypara{Payloads}
A payload is one assignment message in the canonical encoding (\S\ref{app:cf-terms}); a cold-call is an assignment to the index-0 serializer, distinguished by its global name, not by a message kind. The layer's signing primitives back the signature kernels; the signed content is specified in \S\ref{app:cf-handshake} and the Secure GLP paper~\cite{keidar2026secure}.

\mypara{Global-name indices}
Each agent allocates global-name indices from one counter shared by Globalize and Localize, starting at 1 --- index~0 is the serializer (Definition~\ref{def:index-0-serializer}). A removed entry's index is never reused, so a global name denotes one link for the whole run.

\mypara{Early messages}
Delivery is unordered, so an assignment to $\_r(p,i)$ may arrive before the message whose localization creates its global writers table entry. The Receive transaction is enabled only when the entry exists (Definition~\ref{def:madglp-receive}); the runtime holds such an early assignment, keyed by its global name, and applies it when localization creates the entry.

\mypara{Link release}
A reader-name value is retained by its sender, pending, until its acknowledgement arrives (\S\ref{app:requests-acks}); a request re-addresses the pending value to the reader's current holder. Traffic for a closed link --- a late value, request, or acknowledgement --- is dropped; indices are never reused, so the drop is failsafe.

\mypara{Held links}
A link end whose other end is at an agent the runtime did not exchange the name with is held and reported (\S\ref{app:held-links}). The test is structural --- the anchor named in the global name, or the agent an entry was exported to, against the agent that sent the traffic --- and the runtime does not consult the layer's reachability or trust state for it. A held outgoing message stays in $M_p$ marked held; a held incoming assignment is retained keyed by its global name, as an early assignment is. \texttt{authorise\_link/2} releases it, or drops the link.

\mypara{Simulation realisation}
The simulation realisation runs agents as isolates behind the same interface, with a router owning the identifier--key directory, the adjacency relation, and per-pair queues; agent identifiers in global names remain symbolic and are resolved to keys at the seam (Remark~\ref{rem:agent-names}). Harness controls partition pairs and defer delivery, exercising fair delivery and order independence.

\subsection{Self-Module}
\label{app:in-self-module}

The kernel predicate \texttt{self\_module}, whose body kernel is \texttt{'\_self\_module'}, returns, as a \texttt{Module} constant, the compiled artefact (\S\ref{app:cf-artefact}) of the app that calls it. An agent ships a running app to a friend by sending this value; the friend adopts it---offer and handshake are Secure GLP's~\cite{keidar2026secure}---and runs it. On the wire the constant travels as its artefact bytes under the module tag (\S\ref{app:cf-terms}), in an assignment message like any value (\S\ref{app:program-shipment}).

Every goal points to a module. An initial goal is started on a module by \texttt{run(Goal, Type, Module)}, and every goal spawned from it carries a program counter into that module. Whether an agent runs one module or a million, it is the same code. \texttt{self\_module} is the inverse of \texttt{run}: it takes no argument and returns the module its goal points to.

The \texttt{Module} constant carries the artefact --- the code, and the certificate over its source and compiled identities --- so the adopter verifies the certificate, checks its source identity against the offered $h(M)$, and runs the code (\S\ref{app:cf-loader}). The constant kind \texttt{Module} is the type system's~\cite{shapiro2026types}; the kernel predicate and the module-value representation are the runtime's.

\section{Development History}
\label{app:dev-history}

This appendix records how the madGLP specification and its Dart implementation co-evolved. The discipline of \cref{sec:ai-methodology} is bidirectional: authority flows from the mathematics to the specification to the code, but implementation runs and specification derivation surface defects, and several such defects proved to be at the mathematical level rather than the code level. Human review, in turn, caught a conceptual error the formalism had absorbed from an earlier draft of the paper. The episodes below are the substantive ones; each names what changed, how it was found, and which layer was at fault.

\mypara{The implementation arc}
The implementation passed through five stages: cGLP, the concurrent source language; dGLP, the single-agent deterministic implementation; maGLP, the abstract multiagent specification; irmaGLP, a first multiagent implementation organised around a binary network transaction; and madGLP, the current implementation, in which a shared variable pair is realised by two local pairs joined by a global link, with forwarding carried by spawned \texttt{global\_send} goals and a reserved index-0 serializer for cold-calls. The migration from irmaGLP to madGLP, in early 2026, is where most of the corrections below occurred.

\mypara{Localize stored the wrong half of the variable pair}
The Localize operation stored, in the global writers table entry, the wrong half of the local variable pair: for an incoming globalized writer $\_w(p,i)$, the entry must store the local writer $X_q$ (so that an arriving assignment can flow to its reader $X_q?$ in $q$'s resolvent), but it stored the reader instead. The error was found by tracing a complete communication scenario during specification derivation: $p$ assigns $X := T$, the message reaches $q$, and $q$ cannot find the local variable to assign. The correction touched the Localize and Receive definitions, the madGLP Reduce transition, and the Globalize--Localize correspondence lemma. The fault was at the specification level and was fixed there.

\mypara{Globalizing at production time --- a layering decision later reversed}
At one point the abstract maGLP semantics was revised to carry an outbox: each agent's local state became a pair $(G_p, O_p)$, Reduce queued each reader assignment into the outbox at production time, and Communicate was reduced to a passive sender (which merely removes the delivered message from its outbox) and an active receiver. The apparent purpose was to align the abstract semantics with the implementation, where outgoing communication is already produced at reduction time by \texttt{global\_send} goals; the change carried no recorded rationale beyond this reading. The change was subsequently reversed at the abstract level: the outbox was judged too concrete for the abstract semantics, which was restored to the simpler form in which a local state is an asynchronous resolvent $(G_p, \sigma_p)$. The production-time, passive-sender treatment was retained one layer down, in madGLP, where it belongs --- and it survives in the \texttt{global\_send} mechanism of \cref{sec:ai-methodology}, whose body predicate queues the outgoing message for delivery. The episode is an instance of the layering discipline deciding what belongs at which level, rather than a defect.

\mypara{Eliminating variable migration}
The earlier irmaGLP design tracked, per agent, which variables had their paired counterpart on another agent, and relayed readers leaving an agent's scope through fresh relay pairs with forwarding goals. An audit of the implementation against the specification during the irmaGLP-to-madGLP migration exposed edge cases in which the routing was missed or misordered. The design was replaced wholesale by the push-based model now in the paper: a shared pair is realised by two local pairs joined by a global link, and every variable remains in its originating agent's resolvent, with no migration. The defect surfaced in the implementation; the resolution was a redesign reflected in both specification and code.

\mypara{Unifying cold-calls with regular communication}
A cold-call --- a message between agents that share no variable yet --- was first handled by a dedicated binary network transaction. That mechanism was removed once as out of scope, then reintroduced, the separate path recurring because it duplicated machinery that regular communication already had. It was finally unified: cold-calls now use the same \texttt{global\_send} mechanism as established links, addressed to a reserved index-0 serializer, making every transaction unary and aligning the implementation with the source-language paper. The redundancy surfaced in the code; the unification was carried through specification and paper together.

\mypara{Removing the false ``callback'' asymmetry}
An earlier draft of the paper described the case of exporting a reader as a ``callback,'' a term that then propagated faithfully into the specification and the code comments. Human review caught that the term encoded a conceptual error, not merely an infelicity: replacing the word alone would have left the misunderstanding intact. Exporting a writer and exporting a reader are the same symmetric operation --- an agent holds a local pair and spawns a \texttt{global\_send} goal to send its reader's value outward once known, with a table entry at the receiving end. ``Callback'' wrongly implied that one direction was primary and the other a response. The term was removed throughout the paper, the appendix traces, and the implementation, and the symmetric account of global links put in its place. The fault originated in the paper and was corrected there first.

\mypara{Simplifying the global writers table}
The variable table initially stored four kinds of entry, distinguished by origin (Globalize versus Localize) and by variable kind (writer versus reader). Specification derivation made clear that entries are needed only for writers awaiting an incoming assignment, since all outgoing communication is handled uniformly by \texttt{global\_send} goals. The table was renamed the global writers table and its entry kinds reduced from four to two.

\mypara{A cold-call polarity error, found by running the code}
Running the multiagent implementation exposed a polarity error that was present in the abstract specifications, not in the code: the cold-call transaction transferred the reader $X?$ to the receiving agent $q$, whereas $q$ must receive the writer $X$. Because the single-occurrence restriction makes the globalize-reader dataflow always run from the exporting agent to the receiver, the reverse setup was simply wrong. The fix was applied across the madGLP paper, the source-language paper, and the specification together, with the affected single-occurrence-preservation proof reworked. The diagnostic that exposed it was a runtime trace in which a goal failed where it should have suspended:
\begin{verbatim}
[MAD agent2] send: globalized term =
  msg(agent3, data(.(got(1), .(got(2), _r(agent2, 1)))))
[MAD agent3] registered global_send goal: _r(agent2, 1) -> agent2
consumer_init(agent3, [msg(agent3, data([got(1), got(2) | X2])) | X3?]) -> failed
\end{verbatim}
Agent~3 received a writer where it expected a reader, so its \texttt{ground} guard failed definitively instead of suspending, and the goal terminated. This single line --- a goal reaching \texttt{failed} where the semantics required suspension --- drove the polarity correction at the paper level. It is the clearest case of the loop the methodology depends on: the implementation, by being run, falsified the specification.

\mypara{Diagnosing the fault at the right layer}
The same investigation first produced a fix at the code level --- a write-back mechanism at the receiver and additional table entries --- which was then itself judged wrong. Tracing the dataflow showed that the single-occurrence restriction forces the globalize-reader flow to run from the exporting agent to the receiver, whereas the specification had set up the reverse; the two failing tests were caused by this specification-level defect, not by the implementation. The provisional code-level fix was removed and the correction made in the specification and paper instead. This is the discipline's central caution in practice: a defect that first presents as an implementation bug is not patched around in the code once it is traced to the specification.

\mypara{A missing reactivation hook}
A related implementation bug blocked all inter-agent communication using a response variable: the spawned \texttt{global\_send} goal was never registered to fire when its watched reader became known, so a cold-call expecting a reply waited forever. This was a genuine code-level defect --- the specification was correct --- found by running a concrete request-and-response scenario. It is recorded here as the converse case: not every defect surfaced by running the code is a paper-level fault, and distinguishing the two is the point of keeping all three layers aligned.

\mypara{A suspension bug that clarified the specification}
In the core single-agent engine, the soft-fail step that advances a goal to its next candidate clause discarded the clause's suspension set instead of merging it into the goal's, so goals failed where they should have suspended. Fixing the code exposed that the specification had not stated the merge invariant explicitly; the invariant, together with the three-valued reading of guards, was then added to the runtime specification. Here a runtime bug drove a clarification of the specification rather than a change to the mathematics.

\mypara{From relayed links to anchored request/acknowledge transfer}
The relay design kept every past holder of a reader in the loop: a holder that passed a reader on chained a fresh pair to it, so a stream reaching a re-exported reader crossed every past holder forever. The 2026-07 redesign fixes a link's global name --- its \emph{anchor} --- for life: an unbound imported reader is forwarded under its original name, the new holder requests the value from the anchor, a sent value stays pending at the anchor until the current holder acknowledges, and stale traffic is dropped on never-reused indices. A value now crosses once, to the current holder; past holders drop out, and a redirected stream continues to the new holder directly. This partly reverses the ``eliminating variable migration'' episode above: the reader occurrence does move between agents, but against a fixed anchor and under the request/acknowledge discipline, so the routing races that sank irmaGLP's migration do not arise. Epochs on readers were considered and dropped: the acknowledgement alone closes a link, and a request from a holder the anchor did not itself export to is what reports the move.

\mypara{An internal trace inconsistency}
Reserving index~0 for the cold-call serializer left the worked traces and the figure still using index~0 for ordinary global links --- an inconsistency internal to the mathematical layer, introduced by the unification above. The traces, figure, and appendix proof were shifted to start at index~1.

\fi

\end{document}